\documentclass[11pt, reqno]{amsart}
\usepackage{amsmath}
\usepackage{amssymb}

\def\tensor{\,\raise2pt\hbox{${}_{\otimes}$}\,}% Tensor
\def\fdg{\,:\,}% ... fuer die gilt ...
\def\ptl{\partial}% Partial
\def\rest#1{\raise-2pt\hbox{${\lfloor_{#1}}$}}% Restringiert
\def\cal#1{\mathcal{#1}}% Kalligraphisch
\def\mbo#1{\boldsymbol{#1}}% Math-Bold-Symbol
\def\ip#1#2{\langle#1,#2\rangle}
\def\olin#1{\overline{#1}{}}% Oben-quer
\def\grad{{\nabla}}% Gradient
\newcommand{\leftexp}[2]{{\vphantom{#2}}^{#1}{#2}}
\def\halb{\frac{1}{2}}% 1/2
\def \gm{\gamma}

\def \a{\alpha}
\def \b {\beta}

\newtheorem{theorem}{Theorem}[section]
\newtheorem{lemma}[theorem]{Lemma}
\newtheorem{proposition}[theorem]{Proposition}

\newtheorem{remark}[theorem]{Remark}
\newtheorem{definition}[theorem]{Definition}
\newtheorem{claim}[theorem]{Claim}
\newcommand{\ba}{\begin{array}}
\newcommand{\ea}{\end{array}}
\newcommand{\bea}{\begin{eqnarray}}
\newcommand{\eea}{\end{eqnarray}}
\newcommand{\bee}{\begin{eqnarray*}}
\newcommand{\eee}{\end{eqnarray*}}

\newcommand{\R}{\mathbb{R}}

\renewcommand{\gg}{{\bf g}}
\newcommand{\bgg}{{\bar{\bf g}}}
\newcommand{\tgg}{{\tilde{\bf g}}}

\renewcommand{\H}{\mathcal{H}}

\renewcommand{\a}{\alpha}
\renewcommand{\b}{\beta}
\renewcommand{\d}{\delta}

\renewcommand{\r}{\rho}
\newcommand{\la}{\lambda}

\usepackage{color}

\newcommand{\green}[1]{{\color{green}#1}}

\newcounter{mnotecount}[section]

\renewcommand{\themnotecount}{\thesection.\arabic{mnotecount}}

\newcounter{mymnotecount}[section]
\renewcommand{\themymnotecount}{\thesection.\arabic{mymnotecount}}
\newcommand{\mymnote}[1]{\protect{\stepcounter{mymnotecount}}${\raisebox{0.5\baselineskip}[0pt]{\makebox[0pt][c]{\color{green}{\tiny\em$\bullet$\themnotecount}}}}$\marginpar{\raggedright\tiny\em$\!\bullet$\themymnotecount:

\green{#1}}\ignorespaces}

\renewcommand{\mymnote}[1]{}
\begin{document}

\title[Boundary Behaviour in the Orbit Space of the Kerr spacetime]{Axially Symmetric Perturbations of Kerr Black Holes II: Boundary Behaviour of the dynamics in the orbit space}
%\title[Boundary Behaviour in the Quotient of Kerr spacetime]{ Axially Symmetric Perturbations of Kerr Black Holes II: Boundary Behaviour in the quotient space and
%Strict Conservation}


%    author one information
\author{Nishanth Gudapati}
%\address{Center of Mathematical Sciences and Applications, Harvard University, 20 Garden Street, Cambridge, MA-02138, USA}
%\address{Albert Einstein Institute, Am M\"uhlenberg 1, Potsdam-Golm, D-14476, Germany}
\address{Department of Mathematics, Clark University, 950 Main Street,
	Worcester, MA-01610, USA}
%\curraddr{}
\email{ngudapati@clarku.edu}
\thanks{ This work was partly supported by the Deutsche Forschungsgemeinschaft (DFG) grant number GU 1513/2-1 at the Albert Einstein Institute in the fall of 2019} %; and  the Gordon and Betty Moore Foundation and the John Templeton Foundation, through the Black Hole Initiative of Harvard University}

%    \subjclass is required.

\date{\today}

%    Abstract is required.
\begin{abstract}
	In a previous work, we constructed a positive-definite total energy functional for the axially symmetric linear perturbative theory of Kerr black hole spacetimes.  That work is based on the dimensional reduction of dynamical axisymmetric spacetimes into 2+1 Einstein-wave map system. In the construction of the positive-definite energy, various dynamical terms, at the boundary of the orbit space, critically occur. 
	
	In this work, after setting up the initial value problem in harmonic coordinates, we prove that the positive energy for the axially symmetric linear perturbative theory of Kerr black holes is strictly conserved in time, by establishing that all the boundary terms dynamically vanish for all times. This result implies a form of dynamical linear stability of Kerr black holes.%, for the entire sub-extremal range.  
\end{abstract}

\maketitle

\section{Background and Motivation for this Article}
In any dynamical physical theory, the concept of energy plays a fundamental role in the analysis of stability of its solutions. For instance, in the Lyapunov theory of stability, the notion of energy, 
its positive-definiteness and its dynamical behaviour act as important basis for various notions of stability. Likewise, PDE techniques are typically based on a conserved and positive energy.  

As we discussed in our previous work \cite{NG_19_2}, the stability of Kerr black hole spacetimes is an important open problem in the mathematical studies of the Einstein general theory of relativity. 
A major obstacle in the Kerr black hole stability problem is that energy of fields propagating in the Kerr black hole spacetimes is typically not positive-definite, due to the ergo-region, that always surrounds the Kerr black hole with non-vanishing angular-momentum.

In that work (hereinafter referred to as `part I') we constructed a positive-definite `bulk' Hamiltonian energy functional for the axially symmetric linear perturbative theory of Kerr black hole spacetimes for the entire subextremal range $(\vert a \vert<M)$; and the associated Lagrangian and Hamiltonian variational principles. %The construction of the energy holds for any stationary spacetime that admits the dimensional reduction to the $2+1$ dimensional Einstein-wave maps system.  \\
This construction is relevant because it seeks to overcome the aforementioned obstacle in the Kerr black hole stability problem. 
%caused by the ergo-region that always surrounds a Kerr black hole with non vanishing angular-momentum. 
%This construction, which is based on Hamiltonian methods, is a special case of the stability results of the Kerr-Newman $(a, Q, M)$ spacetime. 

We used a special `Weyl-Papapetrou' gauge that provides additional structure in Einstein's equations for general relativity:
\begin{align}
R_{\mu \nu} =0, \quad \mu, \nu=0, 1,2, 3 \quad (\bar{M}, \bar{g})
\end{align}
when the $3+1$ Lorentzian spacetime $(\bar{M}, \bar{g})$ admits a rotational isometry: 
\begin{align} \label{WP-gauge}
\bar{g} = \vert \Phi \vert^{-1} g + \vert \Phi \vert (d \phi + A_\nu dx^\nu)^2, \quad \nu = 0 ,1, 2 \quad \text{(Weyl-Papapetrou gauge)}
\end{align}
where $g$ is the Lorentzian metric on the orbit space $M \fdg = \bar{M}/SO(2).$ $A, g$ are independent of the parameter $\phi$ corresponding to the rotationally symmetric Killing vector $\Phi \fdg=\ptl_\phi,$ whose spacetime norm squared is represented as $\vert \Phi \vert.$

The effects of the geometry of the Kerr black hole spacetime manifest themselves in the Kerr black hole stability problem in a fundamental way. In particular, as a result of the ergo-region the perturbations (scalar wave, Maxwell and linearized gravity) do not necessarily have a positive-definite and conserved energy. This significantly limits the immediate use of the standard PDE techniques in establishing the boundedness and asymptotic decay of the perturbations. 

A standard technique to construct a positive-definite energy functional is to consider an energy  current from the linear combination of the time-translational $\ptl_t$ and rotational $\ptl_\phi$ vector fields:
\begin{align}
\ptl_ t + \chi \ptl_\phi.
\end{align} 
 In view of the fact that the corresponding energy is not necessarily conserved, a separate, intricate Morawetz spacetime integral estimate is needed to control this energy in time. Moreover, these methods are suitable for small angular-momentum $\vert a \vert \ll M.$ %In this spirit, there  exist several remarkable works in the recent literature \cite{}. I
 
 In our previous work, we constructed a positive-definite  `regularized' Hamiltonian energy $H^{\text{Reg}}$ for the linear perturbative theory of Kerr black hole spacetimes, that holds for the entire subextremal range $\vert a \vert<M$ of the (background) Kerr black hole spacetimes. We refer the reader to the part I for the general context of our results. 
  
 %We would like to remark that our positive-definite `regularized' Hamiltonian energy $H^{\text{Reg}}$ \cite{} corresponds to the Hamilton-Jacobi flow (`symplectic current') of the $\ptl_t$ vector field (see below for details). 
 
 In our method, the construction of a positive-definite and conserved energy for the axially symmetric perturbations of the Kerr metric is based on the dimensional reduction of the $3+1$ spacetime $(\bar{M}, \bar{g})$ in the Weyl-Papapetrou gauge \eqref{WP-gauge}. In the Weyl-Papapetrou gauge the vacuum Einstein equations can be represented as 
 
 \begin{align}
 E_{\a \b} =& T_{\a \b}, \\
 \square_g U^A + \leftexp{(h)}{\Gamma}^A_{BC} g^{\a \b} \ptl_\a U^A \ptl_\b U^B=&0, \quad (M, g)
 \end{align} 
 where $T$ is the energy-momentum tensor 
 \begin{align}
 T_{\a \b} \fdg = \ip{\ptl_\a U}{ \ptl_\b U}_{h} - \halb g_{\a \b} \ip{\ptl^\sigma U}{ \ptl_\sigma U}_h
 \end{align}
 
 \noindent of the wave map $U \fdg (M, g) \to (N, h)$ the target $N$ is the negatively curved hyperbolic $2$- plane $\mathbb{H}^2$, $E_{\a \b} \fdg = R_{\a \b} - R_g g_{\a\b}, \a \b =0, 1, 2$ are the Einstein tensor, Ricci tensor and the scalar curvature of the orbit space $(M, g)$ and $\square_g \fdg = \leftexp{(g)}{\grad}^\a \leftexp{(g)}{\grad}_\a $ is the covariant wave operator. 
 In the construction of the energy, we use the linearization stability methods, whose ideas were originally developed in a different context. Using the `alternative' Hamiltonian constraints, 
 
 \begin{align}
 H'^{\text{Alt}} \fdg = \bar{\mu}_{q_0} (2 \Delta_0 \mbo{\nu}'  + h_{AB}(U) U'^B (\Delta_0 U^A + \leftexp{(h)}{\Gamma}^A_{BC} q^{ab}_0 \ptl_a  U^B \ptl_b U^C))
 \end{align}
 
 and identities of the form 
 
 \begin{align}
 &\halb \ptl_{U^C} h_{AB} q^{ab} \ptl_a U^A \ptl_b U^B U'^C + h_{AB}(U) q^{ab} \ptl_a U'^A \ptl_b U^B \notag\\
 &= h_{AB} U'^B (\Delta_q U'^A + \leftexp{(h)}{\Gamma}^A _{CD} q^{ab} \ptl_a U^C \ptl_b U^D ) 
 \end{align}

 and 
 
 \begin{align}
 H''^{\text{Alt}} \fdg =& \bar{\mu}^{-1}_{q_0} (2 e^{-2 \mbo{\nu}} \Vert \varrho' \Vert^2_{q_0} - \tau'^2 e^{2 \mbo{\nu}} \bar{\mu}^2_{q_0} + p'_A p'^A) \notag\\
 &- \bar{\mu}_{q_0} h_{AB} U'^B( \leftexp{(h)}{\Delta} U'^A +R^A_{BCD} q_0^{ab} \ptl_a U^B \ptl_b U^C U'^D)
 \end{align}
     
with further transformations like

\begin{align}
& \leftexp{(h)}{\grad}_a (h_{AB} U'^B \leftexp{(h)}{\grad} ^a U'^A) - h_{AB} U'^B  \leftexp{(h)}{\grad}_a  \leftexp{(h)}{\grad}^a U'^A \notag\\ 
&= h_{AB} q^{ab} \leftexp{(h)}{\grad}_a U'^A \leftexp{(h)}{\grad}_b U'^B , \quad (\Sigma, q),
\end{align}

we constructed a positive-definite  `regularized' Hamiltonian $H^{\text{Reg}}$
\begin{align} \label{e-reg}
H^{\text{Reg}} \fdg = \int_{\Sigma} \mathbf{e}^{\text{Reg}} \, d^2 x.
\end{align}

\begin{align} \label{e-reg-den}
\mathbf{e}^{\text{Reg}} \fdg=& N \bar{\mu}^{-1}_{q_0} e^{-2 \mbo{\nu}} \left( \Vert \varrho' \Vert^2_{q_0} + \halb p'_A p'^A  \right)  \notag\\
&+ \halb N \bar{\mu}_{q_0} q_0^{ab} h_{AB}(U) \leftexp{(h)}{\grad}_a U'^A \leftexp{(h)}{\grad}_b U'^B \notag\\
&- \halb N \bar{\mu}_{q_0} q^{ab}_0 h_{AE} (U) U'^A \leftexp{(h)}R^E _{BCD} \ptl_a U^B \ptl_b U^C U'^D 
\end{align} 

 where $  (U^A, p_A)$ are the elements of the wave map phase space $X'_{\text{WM}} \fdg= \{ (U^A, p_A)$, $A$ is the index on the Riemannian target $(\mathbb{H}^2, h) \}$ and $  (U'^A, p'_A )$ are the elements of the perturbed wave map phase space   $X'_{\text{WM}} \fdg= \{ (U'^A, p'_A)\}$, $U'^A$ are now the elements of the tangent bundle of the target $(\mathbb{H}^2, h).$ In the construction of the `regularized' energy, we fundamentally used the fact $(N, 0)^T$ is the kernel of the adjoint of the dimensionally reduced constraint map. In view of the fact that the aforementioned kernel is one dimensional, we demonstrated that this is the unique choice to construct an energy functional.  Our construction holds for any gauge on the target and also holds for the toroidal perturbations of higher-dimensional black holes \footnote{which admit a dimensional reduction to Einstein-wave system}. If we restrict to a special gauge on the target 
 
 \begin{align}
h = 4 d\gamma + e^{-4\gamma} d\omega^2
 \end{align}
  where the variables $(\gamma, \omega, p, \mbo{r})$ are the elements of the wave map phase space $X\fdg= \{ (\gamma, p), ( \omega, \mbo{r}) \}$ in the chosen gauge  of the target metric $h.$ Likewise,  $(\gamma', \omega', p', \mbo{r}')$
 are the elements of the perturbed phase space  $X' \fdg= \{ (\gamma', p'), ( \omega', \mbo{r}') \}$ for the background Kerr metric. Firstly, an energy functional can be constructed as follows 
 
 \begin{align} \label{Ham-orig}
 E^{\text{Alt}} =& \int_{\Sigma} N \bar{\mu}^{-1}_q (\varrho'^b_a \varrho'^a_b + \frac{1}{8} p'^2 + \halb e^{4\gamma} r'^2) - \halb \bar{\mu}_q \tau'^2 + \bar{\mu}_q q^{ab} (4 \ptl_a \gamma' \ptl_b \gamma' + e^{-4\gamma} \ptl_a \omega' \ptl_b \omega' \notag\\
 &+ 8 e^{-4\gamma} \gamma'^2 \ptl_a \omega \ptl_b \omega- 8 e^{-4\gamma} \gamma' \ptl_a \omega \ptl_b \omega')
 \end{align}
 
As it is evident from the above, this Hamiltonian has a indefinite sign and the effect of the ergo region is evident in this formula. We can construct a positive-definite energy functional for the perturbation theory of Kerr metric using a generalization of the Carter's identity:

\begin{align}
&\bar{\mu}_q q^{ab} \big( 4 \ptl_a \gamma'
\ptl_b \gamma'  + e^{-4\gamma} \ptl_a \omega' \ptl_b \omega'   + 8e^{-4\gamma} \gamma'^2 \ptl_a \omega \ptl_b \omega-8e^{-4\gamma} \gamma' \ptl_a \omega \ptl_b \omega'\big) \notag\\
&+ \ptl_b (N \bar{\mu}_q q^{ab} ( -2 e^{-4\gamma} \ptl_a \gamma \omega' + e^{-4\gamma} \omega' + 4 e^{-4\gamma} \gamma' \ptl_a \omega)) \notag\\
&+ \halb \bar{\mu}_q e^{-4\gamma} L_1 (e^{-2\gamma} \omega') + \bar{\mu}_q L_2 (-4\gamma' \omega')  \notag\\
&= N \bar{\mu}_q q^{ab} (\leftexp{(1)}{V}_a \leftexp{(1)}{V}_b + \leftexp{(2)}{V}_a \leftexp{(2)}{V}_b + \leftexp{(3)}{V}_a \leftexp{(3)}{V}_b),
\end{align}
where,
\begin{subequations}
	\begin{align}
	\leftexp{(1)}{V}_a =& 2 \ptl_a \gamma' + e^{-4\gamma} \omega' \ptl_a \omega, \\
	\leftexp{(2)}{V}_a=& -\ptl_a (e^{-2\gamma}\omega') + 2 e^{-2\gamma}\gamma' \ptl_a \omega, \\
	\leftexp{(3)}{V}_a =& 2 \ptl_a \gamma \omega' -2\gamma' \ptl_a \omega,
	\end{align}
\end{subequations}
and 
\begin{subequations}
	\begin{align}
	L_1 \fdg =& e^{-2 \gamma} ( \ptl_b (N \bar{\mu}_q q^{ab} \ptl_a \gamma) + N e^{-4\gamma} \bar{\mu}_q q^{ab} \ptl_a \omega \ptl_b \omega) \\
	L_2 \fdg =& - \ptl_b( N \bar{\mu}_q q^{ab} e^{-4\gamma} \ptl_a \omega) 
	\end{align}
\end{subequations}

It should be pointed out that a positive-definite energy functional for axially symmetric perturbations was first constructed by Dain-de Austria for the extremal $\vert a \vert =M$ Kerr black hole spacetimes. This work, which is based on the Brill mass formula, uses the original Carter's identity (with $N = `\rho'$). 

The work \cite{GM17_gentitle} deals with the Kerr-Newman (for the entire subextremal range $\vert a \vert, \vert Q \vert <M$) dynamical stability problem within Einstein-Maxwell system of equations and uses a generalization of the Robinson's identity. 
In addition to the construction of a positive-definite and conserved energy for the Kerr-Newman stability problem, this work comprehensively deals with several peripheral aspects of this problem. The current work on Kerr black holes was developed parallely and intended for mathematics audience. 

%which focuses specifically on strict conservation of the positive energy, was developed parallelly and intended for mathematics audience. 

% this work contains comprehensive details of our framework, methodology and several peripheral aspects, addressing this more general problem. %; and shall serve as a comprehensive treatise on our methodology. 

 The problem of construction of a positive-definite and conserved energy for axially symmetric Maxwell fields on Kerr black holes is a non-trivial one. Indeed, whether the Maxwell fields admit a positive-definite and conserved energy, analogous to an axially symmetric linear scalar field, was an open problem until it was first resolved in \cite{GM17_gentitle} (Section 2; see also \cite{NG_17_2, NG_19_1}). Interestingly, it appears that this energy functional follows as a special case of the Kerr-Newman black hole stability problem, rather than that of Kerr black holes.

 Following the use of the Carter-Robinson identities for the sub-extremal case, there was still the lingering question of why do these transformations magically solve the issues of the ergo-region and the positivity and conservation of energy in the stability theory of Kerr and Kerr-Newman black holes, even in the axially symmetric case, which as we just discussed is nontrivial because the energy density can in principle be locally negative. Our results in \cite{NG_19_2} demonstrate that the reason for these transformations holding in such a way that one can obtain positivity and yet preserve the 'symplectic structure', is not  `by fluke', but that there are well-defined geometric and variational underlying structures, namely, the covariant (in target) nature of the dimensionally reduced system, the negative curvature  of the target and the linearization stability methods.  In the context of the black hole uniqueness theorems, the associated generalizations of Carter-Robinson identities were constructed by Bunting and Mazur \footnote{It is indeed remarkable that the rather ingenious identity of Carter (Robinson for the Einstein-Maxwell case) that was seemingly constructed from `trail and error', would later have a natural geometric interpretation } \footnote{these results also hold for higher dimensional black holes and have been used for black hole uniqueness theorems in higher dimensions (see  \cite{Hollands-Ishibashi_12})}.  We adapted these results for our present problem of (dynamical) black hole stability in the aforementioned work \cite{NG_19_2}.

 In the process of construction of such a positive-definite Hamiltonian energy, we pick up several boundary terms, at various stages. From both conceptual and technical perspectives these boundary terms play an important role in our dynamical stability problem. These issues do not arise in the corresponding black hole uniqueness theorems of stationary black holes. For the convenience of the reader, let us elaborate on why it is fundamental to rigorously understand the behaviour of these boundary terms and provide a motivation for this article.

 As we already remarked, there is a possibility that the energy density in the original construction can be made to be locally negative. The positive-definite energy is constructed using the transformations that involve boundary terms from a Hamiltonian energy with an indefinite sign, while preserving the `bulk' Hamiltonian structure of the equations. Therefore, in our argument it is fundamental to rule out the `negativity' or the ambiguity of sign does not `secretly get hidden' in a plethora of boundary terms that occur in both the Carter-Robinson type identities and the linearization stability methods.
 
 In the usual PDE theory, we perform a variation with respect to smooth compactly supported $C^\infty_0$ `test functions' to evaluate the field equations and understand the critical points.  However, in the Einstein equations one cannot ignore the boundary terms, because they can have a physical and geometric interpretation. %Indeed, there is a school of thought within the Hamiltonian formulation of general relativity that these surface integrals (e.g., ADM mass and momenta) are `Hamiltonians' for the Einstein's equations. %In physics literature, a Hamiltonian formulation based on such surface integrals is referred to as `the covariant phase-space formulation'. %In this work, we partly reconcile our approach with the covariant phase space formulation. 
 
 From the perspective of calculus of variations, a positive-definite second variation mass-energy corresponds to the mass of the Kerr black hole being a minimizer, for fixed angular momentum, in the space of admissible metrics. This interpretation combines well with the mass-angular momentum inequalities for axisymmetric spacetimes \cite{D09}. However, as we already discussed, the boundary terms cannot be ignored in this interpretation.    
 
 Most PDE techniques in establishing the uniform boundedness and decay, work at their natural best if there exists a positive-definite and (strictly) conserved energy, that goes together with the evolution of the PDEs. In case there exists a residual, dynamically non-vanishing surface integral (especially with an indefinite sign), together with a positive-definite `bulk' Hamiltonian \eqref{e-reg} \eqref{e-reg-den}, then this could significantly impede the efficacy of our positive-definite   `bulk' Hamiltonian \eqref{e-reg} \eqref{e-reg-den} in PDE methods that establish uniform-boundedness and decay. 

The Weyl-Papapetrou gauge \eqref{WP-gauge}, although it plays a fundamental role in the construction of our bulk Hamiltonian energy \eqref{e-reg} \eqref{e-reg-den}, presents significant regularity issues at the axes and at the infinity. This is due to the behaviour of the rotationally symmetric Killing vector field i.e., $\vert \Phi \vert^{-1}$ blows up at the axes $\Gamma$ and $\vert  \Phi \vert$ blows up at the infinity $\iota^0$. These expressions routinely occur in our formulas due to the form of the Weyl-Papapetrou gauge \eqref{WP-gauge}. It may be noted that these regularity issues manifest themselves in the form of boundary terms that occur  precisely at these boundaries $\Gamma, \iota^0$ and $\mathcal{H}^+$. 

The Einstein equations in the Weyl-Papapetrou gauge are not purely hyperbolic in nature. In particular, the Weyl-Papapetrou gauge has coupled elliptic-hyperbolic PDE structure. This causes gauge-related causality issues. Even if we start with a compactly supported initial data, away from the boundaries $\Gamma, \iota^0, \mathcal{H}^+$, thereby ensuring the regularity and the vanishing of boundary flux integrals initially,  it is not necessary that the fields stay compactly supported at later times. In other words, `pure gauge' perturbations and the associated boundary terms, can `kick-in' in the asymptotic regions, away from the causal future of the support of initial data, at later times, thus affecting the boundary behaviour of the dynamics.

On account of the boundary related issues and complications mentioned above,  the question whether the Weyl-Papapetrou gauge is even compatible with the axially symmetric evolution problem of the Einstein equations can legitimately arise. In this work, we shall provide a favourable answer to this question. We explain how we overcome these complications and discuss our methods in the following. 

We formulate the initial value problem of the linearized Einstein equations in the harmonic gauge and transform the solutions to the Weyl-Papapetrou gauge. In the harmonic gauge, the constraints and the gauge condition propagate in time automatically, as long they are satisfied on the initial data. Indeed, in the harmonic gauge the linearized Einstein equations system is purely a hyperbolic system of equations, which in turn implies propagation of regularity and causality, from standard theory of hyperbolic PDE. We take advantage of the global regularity in harmonic gauge including at the axes and infinity in harmonic gauge. We would like to point out that, when we make such gauge transformations, we can lose regularity but in our construction we show  infact that there exists a $(C^\infty-)$diffeomorphism.  %from harmonic coordinates $(\bar{M}', \bgg')$ to Weyl-Papapetrou gauge $(\bar{M}', \bar{g}').$ 
 In our work, we also construct a variational formulation of the linearized Einstein equations in harmonic gauge, which may be of independent interest and fit well with the theme of our approach. 

%In view of such complications related to technique and interpretation, 
In view of such conceptual and technical subtleties, the `safest' way to prove that the positive-energy $H^{\text{Reg}}$ we constructed is strictly conserved is to use the argument that the time-derivative of the energy $H^{\text{Reg}}$ vanishes dynamically in time. In other words, as we have already shown that the time derivative of the energy density $\mathbf{e}^{\text{Reg}}$  is a pure spatial divergence (a fact that is associated to $(N, 0)^{\text{T}}$ being  the kernel of the adjoint of the dimensionally reduced constraint map):

\begin{align}
(J^ t)^ {\text{Reg}} \fdg =& \, \mathbf{e}^{\text{Reg}} \notag\\
(J^b)^{\text{Reg}} \fdg=& \, N^2 e^{-2 \mbo{\nu}}(q_{0}^{ab} p'_A \ptl_a U'^A) + U'^A \mathcal{L}_{N'} ( N\bar{\mu}_{q_0} q^{ab}_0 h_{AB} \ptl_b U^B)  ) \notag\\
& \mathcal{L}_{N'} (N) (2 \bar{\mu}_{q_0} q_{0}^{ab} \ptl_a \mbo{\nu}') + 2 \mathcal{L}_{N'} \mbo{\nu}' \bar{\mu}_q q^{ab} \ptl_a N - 2 N'^b \bar{\mu}_q q^{bc} \ptl_a \mbo{\nu}' \ptl_c N,
\end{align}
where $J^{\text{Reg}}$ is a divergence-free vector field density,
% $\big \{ (q'_{ab}, \mbo{\pi}'^{ab}), (U'^A, p'_A) \big \} \in \mathscr{C}_{H'} \cap \mathscr{C}_{H'_a} \cap \mathscr{C}_{\tau'}.$ 
we need to prove that the fluxes of $J^b$ at the boundaries vanish dynamically in time. 

%In dealing with the boundary terms, we take advantage of the special structure provided for our problem. The $2-$tensors for our problem can admit a splitting into transverse-traceless tensors and the trace part. This splitting combines well with the Fredholm theory. Subsequently, for the geometry and topology of the orbit space, the transverse-traceless tensors vanish. 
%Furthermore, we take advantage of the fact that the conformal Killing operator is invariant under conformal operator. As a result we obtain decoupled Poisson equations for tensor and vectors.  
 
Using Fredholm theory and that transverse-traceless $2-$tensors vanish \footnote{this is in contrast with the $2+1$ dimensional relativistic Teichm\"uller theory ( for e.g.,  $\Sigma$ is compact and of genus $>1$) , where transverse-traceless tensors play an important role (see \cite{Moncrief_2007}). } for our geometry and topology, we reduce the elliptic operators into conformal Killing operators. Furthermore, benefiting from conformal invariance of these operators, we reduce the elliptic operators into (tensorial) Poisson equations.  We are then able to obtain the desired decay rates of the fields using both the fundamental solution and Fourier methods.

In the fundamental solution approach, we construct regularity and decay rate  in our orbit space geometry, using the method of images in such a way  that the total `charge' of the source is zero. In the Fourier methods, the regularity conditions imply that the frequency corresponding to the logarithmic blow up of the solution does not occur. In both methods, we recover a faster decay rate than for the usual Poisson equation in two dimensions. This faster decay rate, in contrast to the translational symmetric dimensional reduction, plays a fundamental role in our work. 

 The diffeomorphism invariance of the Einstein equations allow us the gauge-freedom. In the $3+1$ (vis-\'a-vis $2+1$) picture, this gauge freedom is reflected in terms of the `lapse' and the `shift' vector field. After fixing the gauge, we estimate the behaviour of gauge dependent quantities in the asymptotic regions where the perturbations are assumed to be pure gauge, starting from the independent wave map phase space variables and then moving on to dependent phase space variables. 
 
 Following the estimates on the fundamental phase space variables, we estimate the fluxes of each of the quantities in $J^b$ term by term and we prove that regularity holds and that they dyanamically vanish at all the three boundaries $\Gamma, \iota^0$ and $\H^+,$ including at the corner $\Gamma \cap \H^+.$  
 
 In establishing the boundary behaviour of these terms, 
 we pay special attention to the quantity $\mbo{\nu}$ that is related to a conformal factor. We note that  $\mbo{\nu}'$ \emph{does not transform like a scalar}. %(in fact, its transformation behaviour is closer to that of a density). 
 This result may be of independent interest in conformal geometry. Secondly, we also establish that an integral quantity $a(R_+, t)$ that vanishes for all times. These two results are crucial in resolving the regularity issues on the axes and at the corners. 
 
%In view of the somewhat surprising and subtle overall nature of the problem and the results, we were cautious in releasing the results, starting with the Maxwell case \footnote{which is a locally gauge-invariant problem}, until all relevant aspects of the problem  are understood. As remarked in part I, our methods provide a robust mechanism to study the stability of a variety of black hole spacetimes. 

In the current work, in keeping with the spirit of part I, we pay special attention to the covariance 
of our analysis with respect to the target metric and not rely on a specific gauge on the target. Apart from the aesthetics, this can be used as a basis for studying the stability of higher $(n+1)$ dimensional black holes with toroidal symmetry $\mathbb{T}^{n-2}.$

\iffalse

We formulate the problem in harmonic coordinates by first computing the variational principles for the perturbative theory in harmonic coordinates. As far as we know this has not been done before, so it could be of independent interest.   We give a variational proof the gauge transformations are abelian for the linear perturbation theory. In other words, the Einstein metrics, orthogonality and Bianchi identity. The fact that the gauge transformations form an abelian group is relevant when use/need multiple gauge transformations. Constraints form an invariant set of the initial value problem of Einstein's equations due to the harmonic gauge condition, as a consequence the system is strictly hyperbolic. Subsequently we prove that there exists a $C^\infty$ diffeomorphism for the perturbative theory \\
**************\\
\fi 

%The problem of stability of Kerr black hole spacetimes is currently an area of very active research among several groups worldwide. 
The problem of stability of Kerr black hole spacetimes is currently an area of very active research among several groups worldwide. In this connection, there are several important and remarkable works (see e.g.,  \cite{HDR_17, GKS_20, HHV_19, DR_11, LB_15_1, LB_15_2, SMa_17_1, SMa_17_2, ABBM_19, Tato_11}). In \cite{HDR_16}, Dafermos-Holzegel-Rodnianski have established linear stability of Schwarzschild black hole spacetimes for gravitational perturbations, where a positive-definite energy played a fundamental role (see \cite{GH_16}). 
A positive-definite energy functional for both even and odd gravitational perturbations of Schwarzschild black holes was constructed by Moncrief \cite{Moncrief_74} using mode decomposition. We should point out that even in the case of Schwarzschild black holes, which do not contain the ergo-region, establishing the existence of a positive-definite energy functional for gravitational is non-trivial. 

In the case of axial symmetry, wave map behaviour in Kerr spaces has been studied in \cite{IK_15}. As we already pointed out, a positive-definite energy was first constructed by Dain-de Austria \cite{DA_14} for extremal Kerr black holes $(\vert a \vert =M)$ using Brill mass formula \cite{B59, D09}. 

 These works (except \cite{DA_14}) are dedicated to the stability of Kerr black hole spacetimes with `small' or `very small' angular momentum $(\vert a  \vert \ll M).$ Relatively less is understood about the stability of Kerr black hole spacetimes for large, but subextremal angular momentum $ \vert a \vert<M$.  This is mainly attributed due to the effects of the ergo-region that always surrounds a Kerr black hole spacetime with a non-vanishing angular momentum. 

%Arguably, the most difficult aspect of the mathematical problem of stability of Kerr black hole spacetimes ($\vert a \vert<M$; and in general rotating black hole spacetimes) is the ergo-region  and the associated lack of a positive-definite and conserved energy functional.  As a consequence, most of the important and remarkable works focus on the  focus on the small angular momentum $(\vert a \vert<M)$ of the Kerr black hole spacetimes. 

In the case of the large and sub-extremal angular-momentum of Kerr black hole spacetimes $ (\vert a \vert <M)$, the effects of the ergo-region become even more subtle and counterbalancing its effects to obtain uniform-boundedness and decay of propagating fields is even more difficult from a PDE perspective.
The decay of a linear wave equation on Kerr black hole spacetimes with $\vert a \vert$ is studied in the remarkable works \cite{FKSY_05,FKSY_06,FKSY_08,DRS_16}. 

  We expect that our results will be useful to fill this gap for Maxwell and gravitational (i.e., Einsteinian) perturbations. Even among the methods for large $\vert a \vert<M,$ our approach is different in the sense that the (Hamiltonian) flow of our phase space variables is restricted to positive-definite energy surfaces. Thus, we have  a relation (equality) between the energy at different time levels without the need for a spacetime `bulk' integral or Morawetz estimate. \\

The more general class of $3+1$ Lorentzian spacetimes is the Kerr-Newman family of spacetimes which is a solution of coupled  $3+1$ Einstein-Maxwell equations for general relativity. In a series of classical works \cite{Moncrief_74_1,Moncrief_74_2,Moncrief_74_3},  the stability of Reissner-Nordstrom spacetimes is studied for the entire sub-extremal range $\vert Q \vert <M$. Indeed, the fact that the Hamiltonian stability results of Reissner-Nordstrom spacetimes hold for the full sub-extremal range was an early indication that the current results and the Kerr-Newman results \cite{GM17_gentitle} are feasible. In view of the rigidity of the Reissner-Nordstrom spacetimes, `non-trivial' perturbations of the RN spacetimes belong to the Kerr-Newman family of black hole spacetimes. Apart from the Reissner-Nordstrom perturbations, as remarked in \cite{GM17_gentitle}, the mathematical problem of stability of perturbations of genuine Kerr-Newman black hole spacetimes $(a, Q \neq 0)$ is almost completely uncharted (see a recent work \cite{E_Giorgi_21}).  %\footnote{except for a brief discussion in the classical monograph of Chandrasekhar \cite{}}. 
We expect that the results in \cite{GM17_gentitle} shall allow us to venture in this direction. Equivalent results for the stability of Kerr-Newman-de Sitter spacetimes $(\vert a \vert, \vert Q \vert<M)$, which has different asymptotics and gauge issues compared to the Kerr-Newman problem in \cite{GM17_gentitle}, is a work in progress ( cf. \cite{NG_17_1, NG_17_2} for now).

Hollands and Wald \cite{WH_13} have developed a notion of canonical energy, which was later extended by Prabhu-Wald, where they showed that if the energy is not positive-definite for axisymmetric perturbations of Kerr black hole spacetimes, it would blow up at later times. 
This work (together with \cite{NG_19_2}), based on the $2+1$ Einstein-wave map formalism and dimensional reduction in Weyl-Papapetrou gauge, confirms their criterion for axisymmetric stability. 
%In other words, they conjecture that for Kerr black hole spacetimes to be stable for axisymmetric perturbations, there should be a positive-definite energy. 
%In our work (together with \cite{NG_19_2}), which originally began with the aim of understanding $2+1$ Einstein-wave map system in dynamical axisymmetric spacetimes (via dimensional reduction, in Weyl-Papapetrou gauge), we resolve their conjecture, for the full sub-extremal range $\vert a \vert <M.$

%Separately, the construction used in this work can also be useful to study the stability of the Kerr metric $(\vert a \vert<M)$ (vis-\'a-vis for  the Kerr-Newman metric $(\vert a \vert, \vert Q \vert <M) $, by building on []) in harmonic coordinates. %, in view of the fact that our regularized Hamiltonian energy $H^{\text{Reg}}$ is gauge-invariant. 
%Global existence and nonlinear stability of the Minkowski metric $(\vert a \vert, M =0)$ using harmonic coordiantes was established in []. 
%Recently, the linear stability of Schwarzschild black hole spacetime $(\vert a \vert=0)$ using harmonic gauge was established in \cite{}. 

%\section{Variational Formulation and Global Existence for Linearized Einstein Equations in Harmonic coordinates}
\section{Global Existence and Propagation of Regularity}

Suppose $(\bar{M}, \bar{g})$ is a Lorentzian spacetime, then consider the Einstein-Hilbert action on $(\bar{M}, \bar{g}):$
\begin{align}
S_{\text{H}}[\bar{\gg}] \fdg = \int  \bar{R}_{\bar{\gg}} \bar{\mu}_{\bar{\gg}}.
\end{align}
Consider a curve $\mbo{\gamma}_s \fdg [0,1] \to C(\bgg)$ where $C(\bgg)$ is a space of smooth Lorentzian metrics $\bgg$.  We shall use the following notation $ \bar{\gg}' \fdg = D_{\mbo{\gamma}_s} \cdot \bar{\gg}$ and define 

\begin{align} 
D_{\mbo{\gamma}_s} \cdot S_{\text{H}} [\bgg', \bgg] \fdg = \int \left(  -\text{Ric}(\gg)^{\mu \nu} \bar{\gg}' + \halb R_\gg \gg^{\mu \nu} \gg'_{\mu \nu} + \leftexp{(\bar{\gg})}{\grad}^\mu \leftexp{(\bar{\gg})}{\grad}^\nu \bar{\gg}'_{\mu \nu} - \leftexp{(\bar{\gg})}{\grad}^2 \text{tr}(\bgg')  \right) \bar{\mu}_\bgg 
\end{align} 

Now, consider another analogous curve $\mbo{\gamma}_\la \fdg [0, 1] \to C (\bgg)$ and let us define $ \bgg'' \fdg = D^2_{\gamma_s \gamma_\lambda} \cdot \bgg, $ we then have, for the Kerr background metric. The functional \eqref{EH-second-variation} can be simplified as follows,

\begin{align} \label{EH-second-variation}
D^2_{\mbo{\gamma}_\la \mbo{\gamma}_s} \cdot S_{\text{H}} [\bgg'', \bgg', \bgg] \fdg =&  \int 
\Big(  \bgg^{\mu \nu} \leftexp{(\bar{\gg})}{\grad}^\gm  \leftexp{(\bar{\gg})}{\grad}^\d \bgg'_{\gm \d}  - \bgg^{\mu \nu} \leftexp{(\bar{\gg})}{\grad}^\gm \leftexp{(\bar{\gg})}{\grad}_\gm \text{tr} (\bgg')
\notag\\
&- \leftexp{(\bar{\gg})}{\grad}_\gm \leftexp{(\bar{\gg})}{\grad}^\mu \bgg'^{\gm \nu} - \leftexp{(\bar{\gg})}{\grad}_\gm \leftexp{(\bar{\gg})}{\grad}^\nu \bgg'^{\gm \mu} 
+ \leftexp{(\bar{\gg})}{\grad}^\mu \leftexp{(\bar{\gg})}{\grad}^\nu \text{tr}(\bgg ') \notag\\
&+ \leftexp{(\bar{\gg})}{\grad}^\a \leftexp{(\bar{\gg})}{\grad}_\a \bgg'^{\mu \nu} \Big) \halb \bgg'_{\mu \nu} \bar{\mu}_\bgg +  \leftexp{(\bar{\gg'})}{\grad}^\mu( \leftexp{(\bar{\gg})}{\grad}^\nu \bgg'_{\mu \nu} - \bgg^{\gm \d} \leftexp{(\bar{\gg})}{\grad}_\mu \bgg'_{\gm \d} ) \bar{\mu}_{\bgg}
\notag\\
& + \leftexp{(\bar{\gg})}{\grad}^\mu ( \leftexp{(\bar{\gg})}{\grad}^\nu \bgg'_{ \mu \nu} - \bgg^{\gm \d} \leftexp{(\bar{\gg})}{\grad}_\mu \bgg'_{\gm \d} ) \bar{\mu}_\bgg \notag\\
 & + \leftexp{(\bar{\gg})}{\grad}^\mu ( \leftexp{(\bar{\gg})}{\grad}^\nu \bgg'_{\mu \nu}  - \bgg^{\gm \d} \leftexp{(\bar{\gg})}{\grad}_\mu \bgg'_{ \gm \d} ) ( \halb \bar{\mu}_\bgg \bgg^{\mu \nu} \bgg'_{\mu \nu} ).
\end{align}

We are interested in a variational principle, so that we get the Euler-Lagrangian field equations for the linearized Einstein equations around the Kerr black hole spacetimes. The functional \eqref{EH-second-variation} can be simplified as follows, 

\begin{align} \label{LG-Lagrangian}
S_{\text{LEE}}[\bar{\gg}', \bar{\gg}] =& \int \halb \Big( \leftexp{(\bar{\gg})}{\grad}_\a  \bar{\gg}'_{\mu \nu} \leftexp{(\bar{\gg})}{\grad}^\mu \gg'^{\a \nu}   + \leftexp{(\bar{\gg})}{\grad}_\a \bar{\gg}'_{\mu \nu} \leftexp{(\bar{\gg})}{\grad}^\nu \bar{\gg}'^{\a \mu} - \leftexp{(\bar{\gg})}{\grad}^\mu \gg'_{\mu \nu} \leftexp{(\bar{\gg})}{\grad}^\nu \text{tr} (\gg') \notag\\
&  - \leftexp{(\bar{\gg})}{\grad}_\a  \gg'_{\mu \nu} \leftexp{(\bar{\gg})}{\grad}^\a \gg'^{\mu \nu} 
- \leftexp{(\bar{\gg})}{\grad}^\a \text{tr} (\gg') \leftexp{(\bar{\gg})}{\grad}^\b \gg'_{\a \b} + 
\leftexp{(\bar{\gg})}{\grad} _\a \text{tr} (\gg')  \leftexp{(\bar{\gg})}{\grad}_\a \text{tr}(\gg') \Big) \bar{\mu}_{\bar{\gg}} 
\end{align}

Suppose $D$ is the local continuous group of diffeomorphisms generated by the vector field $\bar{Y}.$ It follows that the 2-tensor 
\begin{align}
(\bar{\gg}')_{\a \b}  =  (\mathcal{L}_{\bar{Y}} \bar{\gg})_{\a \b} =   \leftexp{(\gg)}{\grad}_\a \bar{Y}_\b +  \leftexp{(\gg)}{\grad}_\b \bar{Y}_\a, \quad  \a, \b = 0, 1, 2, 3
\end{align}

 is a critical point of the variational principle \eqref{LG-Lagrangian}, corresponding to the pure-gauge perturbations of the Kerr metric. Likewise, if $\bar{\gg'}_{\a \b}$ is a critical point, then 
 
 \begin{align}
 \bar{\gg}'_{\a \b} +   \leftexp{(\gg)}{\grad}_\a \bar{Y}_\b +  \leftexp{(\gg)}{\grad}_\b \bar{Y}_\a
 \end{align}
for the variational principle \eqref{LG-Lagrangian} for linearized gravity. Furthermore, it also follows that the gauge transformations are Abelian. In general, the  computation of the Euler-Lagrange field equations corresponding to the variation principle \eqref{LG-Lagrangian} involves the terms, 

\begin{itemize}

\item $\leftexp{(\bar{\gg})}{\grad} \bgg''_{\mu \nu} \leftexp{(\bar{\gg})}{\grad}^\mu \bgg'^{\a \nu} + 
\leftexp{(\bar{\gg})}{\grad}_\a \bgg'_{\mu \nu} \leftexp{(\bar{\gg})}{\grad}^\mu \bgg''^{\a \nu} $
which can transformed as 

\begin{subequations}
\begin{align}
\leftexp{(\bar{\gg})}{\grad}_\a \bgg''_{\mu \nu} \leftexp{(\bar{\gg})}{\grad}^\mu \bgg'^{\a \nu} = &
\leftexp{(\bar{\gg})}{\grad}_\a (\bgg''  \leftexp{(\bar{\gg})}{\grad}^\mu \bgg'^{\a \nu})  - \bgg''_{\mu \nu} \leftexp{(\bar{\gg})}{\grad}_\a \leftexp{(\bar{\gg})}{\grad}^\mu \bgg'^{\a \nu} \\
\leftexp{(\bar{\gg})}{\grad}_\a \bgg'_{\mu \nu} \leftexp{(\bar{\gg})}{\grad}^\mu \bgg''^{\a \nu} =& 
\leftexp{(\bar{\gg})}{\grad}^\mu ( \bgg''^{\a \nu} \leftexp{(\bar{\gg})}{\grad}_\a \bgg'_{\nu \nu})-
\leftexp{(\bar{\gg})}{\grad}^\mu \leftexp{(\bar{\gg})}{\grad}_a \bgg'_{\mu \nu} 
\end{align}
\end{subequations}

\item likewise, $ \leftexp{(\bar{\gg})}{\grad}_\a \bgg''_{\mu \nu} \leftexp{(\bar{\gg})}{\grad}^\nu \bgg'^{\a \mu}  + \leftexp{(\bar{\gg})}{\grad}_\a \bgg'_{ \mu \nu} \leftexp{(\bar{\gg})}{\grad}^\nu \bgg''^{\a \mu}$ can be rewritten as 

\begin{subequations}
\begin{align}
\leftexp{(\bar{\gg})}{\grad}_\a \bgg''_{\a \nu} \leftexp{(\bar{\gg})}{\grad}^\nu \bgg'^{ \a \nu} =& 
\leftexp{(\bar{\gg})}{\grad}_a (\leftexp{(\bar{\gg})}{\grad} \bgg'_{\mu \nu} \leftexp{(\bar{\gg})}{\grad} ^\nu \bgg'^{ \a \mu }) - \bgg''_{\mu \nu} \leftexp{(\bar{\gg})}{\grad}_\a \leftexp{(\bar{\gg})}{\grad}^\nu \bgg'^{\a \mu} \\
\leftexp{(\bar{\gg})}{\grad}_a \bgg'_{\mu \nu} \leftexp{(\bar{\gg})}{\grad}^\nu \bgg''^{\a \nu} =& 
\leftexp{(\bar{\gg})}{\grad}^\nu ( \bgg''^{\a \nu} \leftexp{(\bar{\gg})}{\grad}_\a \bgg'_{\mu \nu}) - 
\bgg'' \leftexp{(\bar{\gg})}{\grad}^\nu \leftexp{(\bar{\gg})}{\grad}_\a \bgg'_{\mu \nu}.
\end{align}
\end{subequations}
\item Now then, consider $ - \leftexp{(\bar{\gg})}{\grad}^\mu \bgg''_{\mu \nu} \leftexp{(\bar{\gg})}{\grad}^\nu \text{tr} \bgg' - \leftexp{(\bar{\gg})}{\grad}^\mu \bgg'_{\mu \nu} \leftexp{(\bar{\gg})}{\grad}^\nu \text{tr} \bgg'',$ which can be transformed conveniently as 

\begin{subequations}
\begin{align}
- \leftexp{(\bar{\gg})}{\grad}^\mu \bgg''_{\mu \nu} \leftexp{(\bar{\gg})}{\grad}^\nu \text{tr}\bgg' =& 
- \leftexp{(\bar{\gg})}{\grad}^\mu ( \bgg''_{\mu \nu} \leftexp{(\bar{\gg})}{\grad}^\nu \text{tr} \bgg') 
+ \leftexp{(\bar{\gg})}{\grad}''_{\mu \nu} \leftexp{(\bar{\gg})}{\grad}^\mu (  \leftexp{(\bar{\gg})}{\grad}^\nu \text{tr} \bgg') \\
- \leftexp{(\bar{\gg})}{\grad}^\mu \bgg'_{\mu \nu} \leftexp{(\bar{\gg})}{\grad}^\nu \text{tr} \bgg'' =& 
- \leftexp{(\bar{\gg})}{\grad}^\nu (\text{tr}\bgg'' \leftexp{(\bar{\gg})}{\grad}^\mu \bgg'_{\mu \nu} ) + 
\text{tr}\bgg'' \leftexp{(\bar{\gg})}{\grad}^\nu \leftexp{(\bar{\gg})}{\grad}^\mu \bgg'_{ \mu \nu}.
\end{align}
\end{subequations}
\item Finally, the terms $ - \leftexp{(\bar{\gg})}{\grad}_a \bgg''{\mu \nu} \leftexp{(\bar{\gg})}{\grad}^\a \bgg'^{ \mu \nu}$ and $ \leftexp{(\bar{\gg})}{\grad}^\a \text{tr} \bgg'' \leftexp{(\bar{\gg})}{\grad}_\a \text{tr} \bgg'$ can be transformed as 

\begin{align}
-\leftexp{(\bar{\gg})}{\grad}_\a \bgg''_{\mu \nu} \leftexp{(\bar{\gg})}{\grad}^\a \bgg'^{\mu \nu} =& 
-\leftexp{(\bar{\gg})}{\grad}_a ( \bgg'' \leftexp{(\bar{\gg})}{\grad}^\a \bgg'^{\mu \nu}) + \bgg''_{\mu \nu} \leftexp{(\bar{\gg})}{\grad}_\a \leftexp{(\bar{\gg})}{\grad}^\a \bgg'^{\mu \nu}
\intertext{and}
\leftexp{(\bar{\gg})}{\grad}^\a \text{tr}\bgg'' \leftexp{(\bar{\gg})}{\grad}_\a \text{tr} \bgg' =& \leftexp{(\bar{\gg})}{\grad}^\a (\text{tr}\bgg'' \leftexp{(\bar{\gg})}{\grad}_\a \text{tr}\bgg' ) - 
\text{tr} \bgg'' \leftexp{(\bar{\gg})}{\grad}^\a \leftexp{(\bar{\gg})}{\grad}_\a \text{tr} \bgg' 
\end{align}
respectively. 
\end{itemize}
Assembling the terms from above, we obtain the Euler-Lagrange equations for the variational principle \eqref{LG-Lagrangian} as 

\begin{align}
&\leftexp{(\bar{\gg})}{\grad}^\gm \leftexp{(\bar{\gg})}{\grad}_\mu \bgg'_{\gm \nu} + \leftexp{(\bar{\gg})}{\grad}^\gm \leftexp{(\bar{\gg})}{\grad}_\nu \bgg_{\gm \mu} - \leftexp{(\bar{\gg})}{\grad}_\mu \leftexp{(\bar{\gg})}{\grad}_\nu \text{tr}(\bgg') - \leftexp{(\bar{\gg})}{\grad}^\gm \leftexp{(\bar{\gg})}{\grad}_\gm \bgg'_{ \mu \nu}  \notag\\
&- \bgg_{\mu \nu} 
\leftexp{(\bar{\gg})}{\grad}^\gm \leftexp{(\bar{\gg})}{\grad}^\d \bgg'_{ \gm \d} + \bgg_{\mu \nu} 
\leftexp{(\bar{\gg})}{\grad}^\gm \leftexp{(\bar{\gg})}{\grad}_\gm \text{tr} \bgg' =0, \quad (\bar{M}, \bgg)
\end{align}

which are the field equations for the linearized Einstein equations around the Kerr background spacetime.

%Consider the Einstein's equations in vacuum: 

%\begin{align}
%\olin{R}_{\mu \nu} =0, \quad (\bar{M}, \bar{g})
%\end{align}
%The metric in harmonic coordinates $\bar{\gg}_{\mu \nu}$
%Consider the perturbations $D_{\mbo{\gamma}_s} \cdot \bar{g} \big \vert_{s=0} = \bar{g}'$  

\subsection{Harmonic Coordinates}
Consider a coordinate system $x^\a$ 
\begin{align}
\bar{x}^\a \fdg (\bar{M}, \bar{\gg}) \to (M, \bar{\gg}), \quad \a = 0, 1, 2, 3
\end{align}
It follows that the coordinate functions $x^\a$ satisfy 
\begin{align} \label{coordinate-wm}
\square_{\bar{\gg}} x^\a + \leftexp{(\bar{\gg})}{\Gamma}_{\b \gamma}^\a \bar{\gg}^{\b \gamma} =0, \quad (\bar{M}, \bar{\gg}), \quad  \a, \b, \gamma = 0,1, 2, 3
\end{align}
we point out that this equation \eqref{coordinate-wm} is reminiscent of the wave map equations. With the harmonic gauge condition $\square_\gg x^\a =0$, it follows from  \eqref{coordinate-wm} that $ \leftexp{(\bar{\gg})}{\Gamma}_{\b \gamma}^\a \bar{\gg}^{\b \gamma} =0. $ We define an equivalent gauge condition in the perturbative theory using $D_{\mbo{\gamma}_s} \cdot \leftexp{(\bar{\gg})}{\Gamma}_{\b \gamma}^\a \bar{\gg}^{\b \gamma}, $ i.e., 

\begin{align} \label{harmonic-gauge-condition}
\leftexp{(\bar{\gg})}{\grad}^ \mu \bgg'_{\mu \nu} = \halb \ptl_ \nu \text{tr} \bgg'
\end{align}
and if we define $\tilde{\gg}' \fdg = \bar{\gg}'- \halb \bar{\gg}^{\a \b} \bar{\gg}'_{\a \b}, $  the `trace-reversed' metric perturbation $\bar{\gg}',$ then this condition can be compactly represented as 

\begin{align} \label{harmonic-gauge-condition-trace-reversed}
\leftexp{(\bgg)}{\grad}_\a \tgg^{\a \b} =0.
\end{align}  
 If we consider the variational principle for the linearized gravity in the Einstein equations for general relativity %We can transform the expressions as follows: 
%\begin{align}
%S [\bgg'] =
%\end{align}
and transform using harmonic coordinates. 

\begin{align} \label{HG-LEE}
S_{\text{HG}}[\bar{\gg}', \bar{\gg}] =& \int \halb \Big( \leftexp{(\bar{\gg})}{\grad}_\a  \bar{\gg}'_{\mu \nu} \leftexp{(\bar{\gg})}{\grad}^\mu \gg'^{\a \nu}   + \leftexp{(\bar{\gg})}{\grad}_\a \bar{\gg}'_{\mu \nu} \leftexp{(\bar{\gg})}{\grad}^\nu \bar{\gg}'^{\a \mu} 
  - \leftexp{(\bar{\gg})}{\grad}_\a  \gg'_{\mu \nu} \leftexp{(\bar{\gg})}{\grad}^\a \gg'^{\mu \nu} 
\Big) \bar{\mu}_{\bar{\gg}} 
\end{align}

It may be noted that the structure of the variational functional, is closely related to the deformed wave map action that we considered previously e.g., in \cite{NG_19_2}.  The quantities in the variational principle can be transformed, using the identities 

\begin{align}
\leftexp{(\bar{\gg})}{\grad}_\gamma \leftexp{(\bar{\gg})}{\grad}_\mu \bgg'_{\a \nu} - \leftexp{(\bar{\gg})}{\grad}_\mu \leftexp{(\bar{\gg})}{\grad}_\gm \bgg'_{\a \nu} = -\text{Riem}(\bgg)^\sigma_{\,\,\,\,\a \gm \mu} \bgg'_{\sigma \nu} - \text{Riem}^\sigma_{\,\,\,\, \nu \gm \mu} \bgg'_{\a \sigma} \notag\\
\leftexp{(\bar{\gg})}{\grad}_\gamma \leftexp{(\bar{\gg})}{\grad}_\nu \bgg'_{\a \mu} - \leftexp{(\bar{\gg})}{\grad}_\nu \leftexp{(\bar{\gg})}{\grad}_\gm \bgg'_{\a \mu} = -\text{Riem}(\bgg)^\sigma_{\,\,\,\,\a \gm \nu} \bgg'_{\sigma \mu} - \text{Riem}^\sigma_{\,\,\,\, \mu \gm \nu} \bgg'_{\a \sigma}
\end{align}
%as 
% when the gauge-freedom is defined as 
%\begin{align}
%U_H \fdg M \to M, 
%\end{align}
then the Euler-Lagrange linearized Einstein equations in harmonic coordinates  are
\begin{align} \label{harmonic-linearized-einstein}
\leftexp{(\bgg)}{\grad}^\gm \leftexp{(\bar{\gg})}{\grad}_\gm \bgg'_{\mu \nu} + \text{Riem}(\bgg)^{\sigma \,\,\,\a}_{\,\,\,\mu\,\,\,\nu} \bgg'_{\a \sigma} + \text{Riem}(\bgg)^{\sigma \,\,\,\a}_{\,\,\,\nu\,\,\,\mu} \bgg'_{\a \sigma} =0, \quad (\bar{M}, \bgg)
\end{align}

a similar equation is satisfied by the trace-reversed metric $\tgg$, 

\begin{align} \label{harmonic-linearized-einstein-trace-reversed}
\leftexp{(\bgg)}{\grad}^\gm \leftexp{(\bar{\gg})}{\grad}_\gm \tgg'_{\mu \nu} + \text{Riem}(\bgg)^{\sigma \,\,\,\a}_{\,\,\,\mu\,\,\,\nu} \tgg'_{\a \sigma} + \text{Riem}(\tgg)^{\sigma \,\,\,\a}_{\,\,\,\nu\,\,\,\mu} \tgg'_{\a \sigma} =0, \quad (\bar{M}, \bgg).
\end{align}

%and the stress-energy tensor 
%\begin{align}
%T[\bar{\gg}_{\text{TR}}'] \fdg = 
%\end{align}
%It follows from the gauge-conditions that the stress-energy tensor is divergence-free.

In the initial value framework, we need to solve the field equations \eqref{harmonic-linearized-einstein} or \eqref{harmonic-linearized-einstein-trace-reversed} together with the gauge determining equations \eqref{harmonic-gauge-condition}. 
Now consider the divergence of the linearized Einstein tensor, 

\begin{align}
\halb & \leftexp{(\bgg)}{\grad}^ \mu \big(  \leftexp{(\bgg)}{\grad}^\gm  \leftexp{(\bgg)}{\grad}_\mu \bgg'_{ \gm \nu} +  \leftexp{(\bgg)}{\grad}^\gm  \leftexp{(\bgg)}{\grad}_\nu \bgg'_{\gm \mu} -  \leftexp{(\bgg)}{\grad}_\mu  \leftexp{(\bgg)}{\grad}_\nu \text{tr}\bgg' -  \leftexp{(\bgg)}{\grad}^\gm  \leftexp{(\bgg)}{\grad}_\gm \bgg'_{ \mu \nu}  \notag\\
& - \bgg_{\mu \nu}  \leftexp{(\bgg)}{\grad}^\gm  \leftexp{(\bgg)}{\grad}^\d \bgg'_{ \gm \d}  + \bgg_{\mu \nu} \leftexp{(\bgg)}{\grad}^\gm  \leftexp{(\bgg)}{\grad}_\gm \text{tr} \bgg' \big)
\end{align}

If we define the gauge-fixing quantity, $F_\nu = \leftexp{(\bgg)}{\grad}^\mu \tgg'_{ \mu \nu},$ we can construct a propagation equation for $F_\nu$ as 

\begin{align}
\leftexp{(\bgg)}{\grad}^\gm \leftexp{(\bgg)}{\grad}_\gm F_\nu \equiv &  \leftexp{(\bgg)}{\grad}^ \mu
( \leftexp{(\bgg)}{\grad}^\gm \leftexp{(\bar{\gg})}{\grad}_\gm \bgg'_{\mu \nu} + \text{Riem}(\bgg)^{\sigma \,\,\,\a}_{\,\,\,\mu\,\,\,\nu} \bgg'_{\a \sigma} + \text{Riem}(\bgg)^{\sigma \,\,\,\a}_{\,\,\,\nu\,\,\,\mu} \bgg'_{\a \sigma} ) \notag\\
=& 0, \quad (\bar{M}, \bgg). 
\end{align}
%using the transformation, 
%\begin{align}
 %\leftexp{(\bgg)}{\grad}
%\end{align}
If we consider the initial value problem 

\begin{subequations}
\begin{align}
\leftexp{(\bgg)}{\grad}^\gm \leftexp{(\bgg)}{\grad}_\gm F_\nu =& 0, \quad &( \bar{M}, \bgg) \\
F_\nu \vert_{\olin{\Sigma}_0} =0, \quad &\ptl_{\vec{t}} F_\nu \vert_{\olin{\Sigma}_0}  =0, \quad & (\olin{\Sigma}_0, \bar{q}_0)
\end{align}
\end{subequations}

It is straightforward to show that $F_\nu$ and $\ptl_{\vec{t}} F_\nu \equiv 0$ for all times in the domain of outer communications of the Kerr black hole spacetime. In other words, the gauge condition \eqref{harmonic-gauge-condition} is propagated for all times if it holds on the initial data and if the linearized Einstein equations in harmonic gauge hold. Analogously, consider the propagation of the constraint equations, defined as 

\begin{align}
  H' \fdg =& G' (n, n) \notag\\
  =& \, n^\mu n^\nu \Big( \leftexp{(\bgg)}{\grad}^\gm  \leftexp{(\bgg)}{\grad}_\mu \bgg'_{ \gm \nu} +  \leftexp{(\bgg)}{\grad}^\gm  \leftexp{(\bgg)}{\grad}_\nu \bgg'_{\gm \mu} -  \leftexp{(\bgg)}{\grad}_\mu  \leftexp{(\bgg)}{\grad}_\nu \text{tr}\bgg' -  \leftexp{(\bgg)}{\grad}^\gm  \leftexp{(\bgg)}{\grad}_\gm \bgg'_{ \mu \nu}   \notag\\
  & \quad - \bgg_{\mu \nu}  \leftexp{(\bgg)}{\grad}^\gm  \leftexp{(\bgg)}{\grad}^\d \bgg'_{ \gm \d}  + \bgg_{\mu \nu} \leftexp{(\bgg)}{\grad}^\gm  \leftexp{(\bgg)}{\grad}_\gm \text{tr} \bgg' \Big)
  \intertext{the expression for the propagation of the Hamiltonian constraint $H'$ follows from the transformations analogous to the above}
  =& n^\mu n^\nu \Big ( \leftexp{(\bgg)}{\grad}_\mu \Big(\leftexp{(\bgg)}{\grad}^\gm \bgg_{\gm \nu} - \halb \leftexp{(\bgg)}{\grad}_\nu \text{tr}\bgg' \Big) +  \leftexp{(\bgg)}{\grad}_\nu \Big(\leftexp{(\bgg)}{\grad}^\gm \bgg_{\gm \mu} - \halb \leftexp{(\bgg)}{\grad}_\mu \text{tr}\bgg' \Big) \Big) \notag\\
  &-  n^\mu n^\nu \Big(   \leftexp{(\bgg)}{\grad}^\gm \leftexp{(\bar{\gg})}{\grad}_\gm \bgg'_{\mu \nu} + \text{Riem}(\bgg)^{\sigma \,\,\,\a}_{\,\,\,\mu\,\,\,\nu} \bgg'_{\a \sigma} + \text{Riem}(\bgg)^{\sigma \,\,\,\a}_{\,\,\,\nu\,\,\,\mu} \bgg'_{\a \sigma} \Big)  \notag\\
  &+ n^ \nu n_\nu \leftexp{(\bgg)}{\grad}^\gm ( \leftexp{(\bgg)}{\grad}^\gm \leftexp{(\bgg)}{\grad}^\d - \leftexp{(\bgg)}{\grad}^\gm \leftexp{(\bgg)}{\grad}_\gm \text{tr}\bgg'  )  
  \intertext{after imposing the linearized Einstien field equations in the harmonic gauge }
   =& n^ \mu n^\nu \Big( \leftexp{(\bgg)}{\grad}_\mu F_\nu + \leftexp{(\bgg)}{\grad}_\nu F_\mu  \Big) + n^\nu n_\nu \leftexp{(\bgg)}{\grad}^\gm F_\gm =0 \\
   \intertext{for all times, due to the Harmonic gauge propagation. Likewise, for the momentum constraint, after plugging in the relevant formulas}
  H'_i \fdg =& G' (n, X), \quad X^i \in T(\olin{\Sigma}), i= 1, 2, 3. \notag\\
  =& n^\mu X^j \Big( \leftexp{(\bgg)}{\grad}^\gm  \leftexp{(\bgg)}{\grad}_\mu \bgg'_{ \gm j} +  \leftexp{(\bgg)}{\grad}^\gm  \leftexp{(\bgg)}{\grad}_j \bgg'_{\gm \mu} -  \leftexp{(\bgg)}{\grad}_\mu  \leftexp{(\bgg)}{\grad}_j \text{tr}\bgg' -  \leftexp{(\bgg)}{\grad}^\gm  \leftexp{(\bgg)}{\grad}_\gm \bgg'_{ \mu j}   \notag\\
  & \quad - \bgg_{\mu j}  \leftexp{(\bgg)}{\grad}^\gm  \leftexp{(\bgg)}{\grad}^\d \bgg'_{ \gm \d}  + \bgg_{\mu j} \leftexp{(\bgg)}{\grad}^\gm  \leftexp{(\bgg)}{\grad}_\gm \text{tr} \bgg' \Big) \notag\\
  =& n^\mu X^j \Big ( \leftexp{(\bgg)}{\grad}_\mu \Big(\leftexp{(\bgg)}{\grad}^\gm \bgg_{\gm j} - \halb \leftexp{(\bgg)}{\grad}_\nu \text{tr}\bgg' \Big) +  \leftexp{(\bgg)}{\grad}_j \Big(\leftexp{(\bgg)}{\grad}^\gm \bgg_{\gm \mu} - \halb \leftexp{(\bgg)}{\grad}_\mu \text{tr}\bgg' \Big) \Big) \notag\\
  &-  n^\mu X^j \Big(   \leftexp{(\bgg)}{\grad}^\gm \leftexp{(\bar{\gg})}{\grad}_\gm \bgg'_{\mu j} + \text{Riem}(\bgg)^{\sigma \,\,\,\a}_{\,\,\,\mu\,\,\,j} \bgg'_{\a \sigma} + \text{Riem}(\bgg)^{\sigma \,\,\,\a}_{\,\,\,j\,\,\,\mu} \bgg'_{\a \sigma} \Big)  \notag\\
  &+ n^ \mu X^j \bgg_{ \mu j} \leftexp{(\bgg)}{\grad}^\gm ( \leftexp{(\bgg)}{\grad}^\gm \leftexp{(\bgg)}{\grad}^\d - \leftexp{(\bgg)}{\grad}^\gm \leftexp{(\bgg)}{\grad}_\gm \text{tr}\bgg'  )
  \notag\\
  =&    n^ \mu X^j \Big( \leftexp{(\bgg)}{\grad}_\mu F_j + \leftexp{(\bgg)}{\grad}_j F_\mu  \Big) + n^\nu X^j \bgg_{ \nu j} \leftexp{(\bgg)}{\grad}^\gm F_\gm =0.
\end{align}

It follows that in harmonic gauge, the constraints are automatically propagated as long as they are satisfied on the initial data, analogous to the propagation of the harmonic gauge condition. This is an analogous statement for the nonlinear theory developed in the classic work of Choquet-Bruhat. 

Let us now discuss the degrees of freedom of the linearized gravity. The number of independent 
degrees of freedom, modulo the gauge degrees of freedom, is 6. We have shown that the constraint propagation in time follows from the harmonic gauge condition. As a consequence, it follows that these remaining degrees of freedom are all independent and unconstrained. 

It may be noted that the system of equations for the linearized Einstein equations is purely a hyperbolic partial differential equation system, from which it follows that the evolution of the data is causal.

\begin{proposition}
Suppose $(\olin{\Sigma}', \bar{\mbo{q}'} )$ is the linearized initial data for the linearized Einstein equations, satisfying the constraint equations, then
\begin{enumerate} 
\item the linearized Einstein equations in harmonic gauge is purely a hyperbolic differential equation system, with the independent degrees of freedom being dyanamically unconstrained 
\item The future (and past) development of the linearized initial data  in harmonic gauge is globally regular and globally hyperbolic
\end{enumerate}
\end{proposition}

Let us make a couple of comments about the aforementioned global existence theorem. As we already discussed, the Weyl-Papapetrou gauge offers significant benefits in terms of geometry and topology, that allows us to construct a positive-definite energy in the first place, but presents (gauge-related) causality and regularity issues at the boundaries. 

On the other hand, as we already discussed above, the global development of the linearized Einstein equations in the harmonic gauge is untroubled by the (gauge-related) causality and regularity issues at the axes $(\Gamma)$, infinity $(\bar{\iota}^0)$ and at the corners $\Gamma \cap \mathcal{H}^+$. In other words, the global development of the linearized Einstein equations from regular initial data is regular $(C^\infty)$ for all times, including at the pathologies like axes, infinity and at the corner of axes and horizon. For example, we have

\begin{align}
\ptl_{\vec{n}} A =0,  
\intertext{for a scalar, and }
A^{\perp} =0, \quad  \ptl_n A^{\parallel} =0
\end{align}
for a vector $A$, near the axes, globally in time. In our work, our approach is to combine these two gauges and take advantage of benefits offered in each. 

Recall that the metric $q$ in the orbit space can be expressed in harmonic gauge as $\mbo{q} = e^{2 \mbo{\nu}} \mbo{q}_0, $ where $q_0$ is the flat metric. 

The preservation of the flatness condition is

\begin{align} \label{flatness-condition}
D \cdot R'_{\mbo{q}_0} = \bar{\mu}_{\mbo{q}_0} ( \leftexp{(\mbo{q}_0)}{\grad}^a \leftexp{(\mbo{q}_0)}{\grad}^b (\mbo{q}'_0)_{ab} - \leftexp{(\mbo{q}_0)}{\grad}^a \leftexp{(\mbo{q}_0)}{\grad}_a \mbo{q}^{cd}_0 (\mbo{q}'_0)_{cd}  ) =0, \text{on each $\Sigma$}
\end{align}
The tensor $\mbo{q}'_0$ can be decomposed as $(\mbo{q}'_0)_{ab} = (\mbo{q}'^{\text{TT}}_0)_{ab} + 
\leftexp{(q_0)}{\grad}_a Y_b + \leftexp{(q_0)}{\grad}_a Y_a + \halb (\mbo{q}_0) \text{tr} \, \mbo{q}'_0.$ In particular, it may be noted that the pure gauge perturbations $(\mbo{q}'_0)_{ab} = \leftexp{(q_0)}{\grad}_a Y_b + \leftexp{(q_0)}{\grad}_b Y_a$ satisfy the condition \eqref{flatness-condition}. 
 The perturbed metric $(\mbo{q}'_0)$ has the following regularity conditions on the axes $\Gamma,$ expressed in $(R, \theta)$ coordinates \cite{O_Rinne_J_Stewart_2005}

\begin{align}
\ptl_\theta \mbo{q}'_0 (\ptl_R, \ptl_R) =0, \quad \ptl_\theta \mbo{q}'_0( \ptl_\theta, \ptl_\theta)  =0, \quad \mbo{q}'_0 (\ptl_\theta, \ptl_R) =0, \text{at the axes $\Gamma$}.
\end{align}

Our construction can also be used to study stability problem of Kerr black black hole spacetimes in harmonic gauge. However, in principle, we can use our energy to study the stability problem in any preferred gauge.  If we are given a solution $(M', \bgg')$ in any gauge, we can transform the solution to a harmonic gauge by solving the linearized wave map equations $x^\a \fdg \bar{M} \to \bar{M},$ for which it can be proven that they admit smooth solutions for all times. Subsequently, we can transform our the solution to the Weyl-Papapetrou gauge using a $(C^\infty-)$diffeomorphism, by making use of the abelian nature of gauge-transformations. 

\iffalse
Finally, we have few comments on the nonlinear problem. In the current work on the linear stability problem, we establish that the regularity is propagated from the initial data and that a positive-definite energy is conserved. In the nonlinear stability problem, establishing the propagation of regularity from the initial data (i.e., the existence problem) is highly non-trivial. The propagation of regularity of the harmonic gauge is itself quite delicate because $3+1$ wave maps are prone to develop singularities, especially for large data. In a separate work, we aim to pursue the existence problem of U(1) Einstein equations using curvature propagation estimates for curvature, using quasi-local approximate Killing fields. There is a general approach suggested by Moncrief to understand the curvature propagation.   
\fi 
\subsection{Remarks on the nonlinear problem} The ultimate purpose of studying the linear stability problem is that it paves the way for the resolution of  \emph{nonlinear} black hole stability. Therefore, it is pertinent to present a few remarks on the \emph{nonlinear} stability problem and discuss whether the dimensional reduction ($2+1$ Einstein-wave map) framework is useful in this regard. 

In the current work on the linear stability problem, we establish that the regularity is propagated from the initial data and that a positive-definite energy is conserved.  In the nonlinear black hole stability problem, establishing the propagation of regularity from the initial data is an important obstacle.

In previous works \cite{diss_13, AGS_15}, the special structure provided by `energy critical' $2+1$ wave maps was used to study the global behavior for large data of the Einstein equations with a translational symmetry. Even though the local dimensional reduction process to the $2+1$ Einstein-wave map system, for the $3+1$ Einstein equations with translational symmetry is locally similar to that of the axisymmetric case, our axisymmetric wave map problem has a few fundamental differences with the translational symmetric problem. In particular, $3+1$ axisymmetric spacetimes and the resulting $2+1$ Einstein-wave map system are mass-energy \emph{super-critical} in nature. %\footnote{even though, upon the dimensional reduction procedure, we get the same $2+1$ Einstein-wave map system locally as in the translational symmetry}.  
%The Brill mass-energy (equivalent to the ADM mass), which is the Hamiltonian for axisymmetric spacetimes,  is mass-energy super-critical\footnote{even though, upon the dimensional reduction we get the same $2+1$ Einstein-wave map system locally}.
%\footnote{even though we rely on the dimensional reduction for axisymmetric spacetimes to $2+1$ Einstein wave map system, on multiple instances we demonstrate that the behaviour of fields is actually $3+1$ dimensional in nature}, relative to the scaling of wave maps.  
Therefore, the program that is being pursued in those works for the Einstein equations with a translational symmetry differs from the current axisymmetric problem at hand. %This important distinction is explained clearly in \cite{NG17}.  

In the following, we shall  present approach that could be used to study the existence problem of axisymmetric nonlinear perturbations of Kerr black hole spacetimes (see Appendix L in \cite{GM17_gentitle} for some preliminary work).   
This approach is based on the \emph{`light cone'} estimates that were particularly powerful in  the proof of global existence of Yang-Mills fields \cite{Eardley_Moncrief_I, Eardley_Moncrief_II}. It is expected these methods are also useful for the Einstein equations. 

In view of the similarity of the Einstein equations and the Yang-Mills equations, %\cite{Moncrief_2005, KL_RO_SZ_2012}
\`a la Cartan formalism, we present an approach that is based on these methods. 
An essential ingredient in these methods is an integral formula for the propagation equation of linear and nonlinear equations. In the case of the linear wave equation this integral formula 
is a true representation formula for the unknown, but in the nonlinear case the integral formula, of course, involves the unknown itself. For the Yang-Mills fields \cite{Eardley_Moncrief_I, Eardley_Moncrief_II}, using Bianchi identities, one constructs a wave propagation equation for the curvature tensor  and uses the integral formula to represent the curvature tensor in the past light cone. 

It may be noted that the  nonlinear propagation equation contain commutator terms that are quadratic in curvature. The integral equation involves the terms in the past light cone and as well as the initial data. These integral terms are then estimated in terms of the initial data, in particular the energy, which is conserved in time, and the $\Vert \cdot \Vert_{L^\infty}$ norm of the Yang-Mills curvature itself. Consequently, this estimate provides an integral estimate for the 
$\Vert \cdot \Vert_{L^\infty}$ norm of the curvature, involving the initial data. 

In performing these estimates, the special structure offered by the Cronstr\"om gauge is particularly useful. It was shown separately that the Yang-Mills fields in any gauge can be transformed into the Cronstr\"om gauge. Subsequently,  using energy norms equivalent to  $(H^2 \times H^1),$ it is shown that they cannot blow up in finite time. In this step, the aforementioned $\Vert \cdot \Vert_{L^\infty}$ estimate of the curvature plays an important role, in controlling the energy norms. 

In the gravitational problem, the Einstein equations have a structure analogous  to the Yang-Mills equations, when expressed in the Cartan formalism. In Cartan orthonormal frame formalism, the spacetime curvature is expressed in terms of the Lorentz connection. 

In particular,  it can be shown that the spacetime curvature satisfies a non-linear wave equation that is analogous in structure to the Yang-Mills curvature equation. Even for the curved space wave operator, one can construct an integral formula for the curvature in terms of the source terms (which inevitably involve the curvature itself, thus not implying a true `representation formula' ), following Friedlander-Hadamard \cite{Friedlander_1975}. However, the hope is that one can derive estimates for these source terms, analogous to the Yang-Mills problem. 

In the case of the wave operator in a curved space, the integral formula, following the Friedlander-Hadamard analysis, involves integrals over the mantel of the cone, bulk interior of the light cone and the intersection of the initial data with the past light cone. With the aim of overcoming this issue, it was shown in  \cite{Moncrief_2005} that one can transform this integral formula in such a way that it only involves  boundary integrals over the mantel of the cone and the initial data. 

It may be noted that, in contrast with the Yang-Mills problem, the propagation equation for the curvature in the gravitational problem now involves the quadratic spacetime curvature terms, as well as gradient terms of frame fields and curvature. In our approach, the aim  is then to control the curvature in terms of the initial data. 

Quite analogous to the Yang-Mills problem, the terms quadratic in spacetime curvature can be controlled in terms of the Bel-Robinson fluxes, which in turn can be controlled in terms of the Bel-Robinson energy. In order to control the terms involving the gradient of curvature, we need a higher order energy for the curvature. 

However, it may be noted that, in the nonlinear gravitational problem, energies that are conserved or bounded in time are scarce. At a first glance, in order to achieve this control, we need Killing vector fields, conformal Killing vector fields or the like, for the spacetime. However, this is requirement is too strong and will almost certainly exclude interesting cases. 

In the case of translational and equivariant $U(1)$ problem, it was shown that there exists a conserved energy that effectively corresponds to a time-like vector field, which is not a Killing vector field \cite{diss_13, AGS_15}. This fortuitous behaviour was used in a fundamental manner in the aforementioned works. However, this outcome can be attributed to the special structure provided by the equivariance assumption and it is unlikely that this result holds without this assumption. In an attempt to overcome this issue, the  general notion of quasi-local approximate Killing vector fields  was developed in \cite{Moncrief_2005, Moncrief_2014}. The framefields in the Cartan formalism, when subjected to parallel propagation, can be shown to satisfy the Killing equation approximately, with an error term that is explicitly expressible in terms of the curvature and goes to zero as we approach the tip of the cone. We hope that this concept will be helpful in our approach. 

In the linear stability problem considered in the current work, we took advantage of the propagation of regularity in harmonic gauge and performed a gauge transform into the Weyl-Papapetrou gauge. In the nonlinear theory, especially for large data, the harmonic gauge is typically unfavorable for global regularity results. 

In the nonlinear gravitational problem, we need to understand the dynamical behavior of the spacetime curvature as well as the orthonormal framefields and connections. In order to benefit from the Einstein-wave map formalism of our axisymmetric problem, we need to represent the curvature propagation equation and the associated geometric structures in terms of the Weyl-Papapetrou gauge and  dimensionally reduced forms. In this context, the method of Hadamard descent, which allows us to find an integral formula for the wave operator in the dimensionally reduced setting, plays an important role.

In our axisymmetric problem, in consistency with the Weyl-Papapetrou gauge, we can choose a coordinate system such that one of the spatial coordinates coincides with the axial Killing vector field. In such a coordinate system the spacetime 4-metric, connections and the spacetime curvature are all independent of the axial coordinate. Furthermore, in order to be able to fully exploit the Cartan formalism, we can choose the orthonormal frame fields also to be independent of the axial coordinate and that they are invariant along the orbits of the axial symmetry group. 

As we already alluded to, it is at our discretion to choose a gauge condition on the spacetime metric. In particular, for our evolution problem, we have the possibility to also choose the time function. A suitable choice is a time function that results in a foliation of constant mean curvature level surfaces. This gauge condition, together with spatial harmonic (SH) gauge condition on $\Sigma$, is consistent with the gauge chosen for the current work. 

In our framework, based on the Cartan formalism of the Einstein equations and normal frame fields, an important role is played by the injectivity radius in the context of the null geometry of Lorentzian manifolds. In other words, in our approach, it is important to understand the injectivity radius of null cones in conjuction with the spacetime curvature. 

In \cite{LeFloch_08} (see also the references therein), a lower bound on the injectivity radius is obtained, based on conditions on the spacetime curvature ($\Vert\cdot\Vert_{L^\infty}$bound). This is achieved using a `reference' Riemannian metric. An especially relevant aspect of this work for us is that the metric with optimal regularity is constructed in the constant mean curvature spatial harmonic gauge (CMC-SH). We would like to point out that the Weyl-Papapetrou gauge (and in general the dimensional reduction framework) is naturally consistent with the CMC-SH gauge. 

An important feature in the axisymmetric problem is the behavior of fields at the axes and infinity. %We need to pay special attention to the behavior of the relevant fields at these boundaries.
We expect that some of the analysis carried out in the current work would be useful in this regard. 

The CMCSH gauge condition provides elliptic equations for the lapse and the shift. We would like to point out that the elliptic operators in the problem are simplified using the fact that the transverse-traceless tensors vanish on $\Sigma.$ Likewise, the dynamics of the two manifold  is captured by the conformal factor $e^{\mbo{2 \nu}}$ and the quantity  $\mbo{\nu}$ satisfies the Lichnerowicz equation. The elliptic analysis for this equation can be extended from the analysis for the linear theory (for $\mbo{\nu}'$), as discussed in the next section.

\section{Canonical phase space  variables and  Lagrange multipliers in the Weyl-Papapetrou Gauge}

Suppose, the group $SO(2)$ acts on the $3+1$ Lorentzian spacetime $(\bar{M}, \bar{g})$ such that the orbits of the group $SO(2)$ are closed and the group action has a nonempty fixed point set, denoted by $\Gamma.$ These conditions are satisfied by the Kerr metric $(\bar{M}, \bar{g})$

\begin{align}
\bar{g} \fdg =& - \Sigma^{-1} (\Delta -a^2 \sin^2 \theta) dt^2 - 4a\Sigma^{-1} \sin^{2} \theta mr dt d\phi \notag\\
&+ \Sigma^{-1} ((r^2 +a^2)^2-\Delta a^2 \sin^2 \theta) \sin^2 \theta d\phi^2 + \Delta^{-1} \Sigma dr^2 + \Sigma d\theta^2
\intertext{where}
\Sigma \fdg =& r^2 + a^2 \cos^2 \theta \quad \text{and} \quad  \Delta = r^2-2Mr + a^2
\end{align}
which can be represented as: 

\begin{align}
\bar{g} =& \vert \Phi \vert ( \Delta \sin^{-2} \theta dt^2 + R^{-2} \sin^2 \theta ((r^2 + a^2)^2- a^2 \Delta \sin^2 \theta) (d \rho^2 + dz^2) )  \\
&+ \vert \Phi \vert (d \phi^2 - 2 mar ((r^2+a^2)^2 - a^2 \Delta \sin^2 \theta)^{-1})^2
\intertext{where}
\rho =& R \sin \theta \quad z = R \cos \theta, \quad R \fdg = 2 (r-m + \sqrt{\Delta}), \quad \theta \in [0, \pi] 
\end{align}
 
\begin{subequations}
	\begin{align}
	\vert \Phi \vert =& \frac{\sin^2 \theta }{r^2 +a^2 \cos^2 \theta} \left(  (r^2+a^2)^2 -a^2 \Delta \sin^2 \theta \right) \\
	q_{ab} =& \sin^2 \theta  \frac{ (r^2+a^2)^2 -a^2 \Delta \sin^2 \theta }{R^2}
	\end{align}
	\end{subequations}

Now consider the conjugate harmonic functions $(\bar{\rho}, \bar{z})$ such that %$q= e^{2\Omega} (d\bar{\rho}^2 + d\bar{z}^2)$ and 

\begin{align}
\bar{\rho} \fdg = \rho (1- \frac{(m^2-a^2)}{4 (\rho^2 + z^2)} ) \quad \bar{z} \fdg = z (1+ \frac{m^2-a^2}{4 (\rho^2 + z^2)})
\end{align}
so that the Jacobian 

\begin{align}
J \fdg = \begin{pmatrix} 1 + \frac{m^2-a^2}{4 ( \r^2 + z^2)} \left(  \frac{2 \r^2}{ \r^2 +z^2} -1\right) & \r ( 1 + \frac{(m^2-a^2)z}{2 (\r^2 +z^2)^2} \\
z \left( 1- \frac{ ( m^2-a^2) \r}{ 2 (\r^2 + z^2)^2}\right) &  1+ \frac{m^2-a^2}{4 (\r^2 +z^2)} \left( 1- \frac{2 z^2}{\r^2 +z^2} \right) 
\end{pmatrix}
\end{align}
In these coordinates, the Kerr black hole horizon $\mathcal{H}^+$ corresponds to a `cut' on the $\{ \bar{\rho} =0\}$ curve and its complement on the $\{\bar{\rho} =0\}$ curve corresponds to the union of two axes, $\Gamma.$ This coordinate system and the $(\rho, z)$ coordinate system in the extremal case are the ones originally used by Carter \cite{Car_71}. 

In general, in the Weyl-Papapetrou gauge for the Einstein equations, we can reduce the Einstein-Hilbert action
into the reduced Einstein-wave map system: 

\begin{align}
\int (R_g - h_{AB}(U) g^{\a \b} \ptl_\a U^A \ptl_\b U^B) \bar{\mu}_g 
\end{align}

It is straightforward to verify that the Kerr metric is a critical point of the variational functional. In the Hamiltonian version of the dimensional reduction, we also encounter the intermediate phase space $X^{\text{Max}}:$
	
	\begin{align}
	X^{\text{Max}} \fdg = \{\mathcal{A}_i, \mathcal{E}^i \}
	\end{align}  

	so that the Hamiltonian and momentum constraints for the combined phase space $ \{ (q, \mbo{\pi}), (\mathcal{A}, \mathcal{E}),  (\vert \Phi \vert^{1/2}, p) \}$
	
	\begin{align}
	H  \fdg =& \bar{\mu}^{-1}_q ( \Vert \mbo{\pi} \Vert^2_q - \textnormal{tr}(\mbo{\pi})^2) + \frac{1}{8} p^2 + \frac{1}{2} \vert \Phi \vert^{-2} \mathcal{E}^a \mathcal{E}_a + \bar{\mu}_q (-R_q + 2 q^{ab} \ptl_a \log \vert \Phi \vert \ptl_b \log \vert \Phi \vert )
	\notag\\
	&+ \frac{1}{4} q^{ab} q^{bd} \ptl_{[b} \mathcal{A}_{a]} \ptl_{[d} \mathcal{A}_{c]} \\
	H_a=& -2 \leftexp{(q)}{\grad}_b \mbo{\pi}^b_a + p \ptl_a \log \vert  \Phi \vert + \mathcal{E}^b ( \ptl_{[a]} \mathcal{A}_{b]})
	\end{align}
	with the Lagrange multipliers 
	\begin{align}
	\{ N, N^a, A_0\}
	\end{align}
	Subsequently, the phase space $X$ was introduced 
	\begin{align}
	X \fdg = \{ (q, \mbo{\pi}), (\vert \Phi \vert, p), (\omega, \mbo{r}) \}
	\end{align}
	which resulted in the Hamiltonian and momentum constraint equations: 
	\begin{align}
	H=& \bar{\mu}^{-1}_q ( \vert \mbo{\pi} \vert^2_q - \textnormal{tr} (\mbo{\pi})^2 + \halb p_A p^A ) + \bar{\mu}_q (- R_q + \halb h_{AB} q^{ab} \ptl_a U^A \ptl_b U^B) \\
	H_a=& -2 \leftexp{(q)}{\grad}_b \mbo{\pi}^b_a + p_A \ptl_a U^A
	\end{align}
	where the Lagrange multipliers are now	
	\begin{align} \label{lapse-shift-3D}
	\{ N, N^a \}
	\end{align}
	The fact that the Lagrange multiplier set \eqref{lapse-shift-3D} is now simplified is due to the special topological 
	structure of the orbit space $M$ of the Kerr metric. This plays a convenient role in the properties of the adjoint of the dimensionally reduced constraint map.
	The main result of previous work was to obtain a positive-definite energy functional for the linear perturbative theory of Kerr black hole spacetimes within the assumption of axial symmetry, which allows the aforementioned dimensional reduction. The regularized Hamiltonian energy functional was 
	
	\begin{align}
	H^{\text{Reg}} \fdg = \int_{\Sigma} \mathbf{e}^{\text{Reg}} d^2x
	\end{align}
	
	where 
	\begin{align}
	\mathbf{e}^{\text{Reg}} \fdg = & N \bar{\mu}^{-1}_{q_0} e^{-2 \nu} ( \Vert \varrho' \Vert_{q_0}^2 + \halb p'_A p'^A ) - \halb N e^{2 \nu} \bar{\mu}_{q_0} \mbo{\tau}'^2 \notag\\
	&+ \halb N \bar{\mu}_{q_0} q^{ab}_0 h_{AB}(U) \leftexp{(h)}{\grad}_a U'^A \leftexp{(h)}{\grad}_b U'^B \notag\\
	&- \halb N \bar{\mu}_{q_0} q_0^{ab} h_{AE} U'^A \leftexp{(h)}{R}^E_{\,\,\,\, BCD} \ptl_a U^B \ptl_b U^C U'^D.
	\end{align}
	Let us now formally define the Weyl-Papapatrou gauge. 
\begin{definition}
Suppose $(\bar{M}, \bar{g})$ is a Lorentzian spacetime such that $\bar{M}$ admits the ADM decomposition $\bar{M} = \olin{\Sigma} \times \mathbb{R}$ and the group $SO(2)$ acts on $\bar{M}$ through isometries such that the fixed point set in nonempty and the orbits of its action on $\bar{M}$ are closed; suppose  $p$ belongs to the orbit of $SO(2)$ action on $\bar{M}$ then $p \in \olin{\Sigma}$ (i.e., $\text{Orb} (x) \in \olin{\Sigma}, x \in \bar{M}$) and $p \times \mathbb{R}$ is a timelike curve. Then we define $\bar{g}$ to be in Weyl-Papapetrou form if 
\begin{enumerate}
\item $\bar{g}$ admits the decomposition
\begin{align} \label{WP-def}
\bar{g} = \vert \Phi \vert^{-1} g + \vert \Phi \vert \mathcal{A}^2
\intertext{where the 1-form $\mathcal{A}$ in $\bar{M}$ is defined as} 
\mathcal{A} = d \phi + A_\nu dx^\nu, \quad \text{and}
\end{align}
$\ptl_\phi$ is the (Killing) vector field corresponding to the $SO(2)$ symmetry of $\bar{M}$ the scalar $\vert \Phi \vert = \bar{g}_{\a \b} (\ptl_\phi)^\a (\ptl_\phi)^\b$ is the spacetime norm of the Killing vector $\ptl_\phi.$  
$g, A, \vert \Phi \vert$ are independent of $\phi,$ i.e., $\mathcal{L}_{\ptl_\phi} g = \mathcal{\ptl_\phi} A =0.$ It may be noted that $g$ is a metric of Lorentzian signature in the orbit space $M \fdg = \bar{M}/SO(2)$ such that it further admits the ADM decomposition $M = \Sigma \times \mathbb{R},$ where now $\Sigma = \olin{\Sigma}/SO(2):$
\begin{align}
g = -N^2 dt^2 + q_{ab} (dx^a + N^a dt) \otimes (dx^b + N^b dt)
\end{align}
$q$ is the metric of $\Sigma.$
\item There exists a coordinate basis $e^1$ and $e^2$ in $(\Sigma, q)$ such that 
\begin{align} \label{conf-flat-nonlinear}
q(e^1, e^1) - q(e^2, e^2) =0,\quad \text{and} \quad q(e^1, e^2) =0, \quad \text{on each} \quad \Sigma
\intertext{and}
q = e^{2 \mbo{\nu}} q_0, \quad \text{where $q_0$ is the flat 2-metric}.
\end{align} 

\end{enumerate}
\end{definition}
The metric $\bar{g}$ represented in the above coordinate conditions is referred to as in `Weyl-Papapetrou' form. 
We would like to remark that the representation of the metric $\bar{g}$ in terms of a Weyl-Papapetrou form  is \emph{not unique}. We would also like to emphasize that on the fixed point set of the $SO(2)$ action on $\bar{M}$ we have $\vert \Phi \vert \to 0$ ($\vert  \Phi \vert^{-1}\to\infty$) and at the outer asymptotic end of $\bar{M}$ we have $\vert \Phi \vert \to \infty$ (and $\vert \Phi \vert^{-1} \to 0$). Counterbalancing these effects to obtain well-defined, convergent, gauge-independent quantities, in the context of the initial value problem, is one of the main aspects in our work. We shall apply this construction for the perturbative theory. 

\iffalse
\subsection*{2+1 Conformal invariant}
We formulate the gauge-condition in terms of a conformal invariant quantity $\bar{\mu}_q q^{ab}$. 
      
Suppose the gauge-transform from the harmonic coordinates, $\delta \bar{g}$ is given by

\begin{align}
\delta \bar{g}' = \bar{g}' + \mathcal{L}_{\xi} \bar{g}
\end{align}
where $\xi$ is the vector field that generates the gauge-transform. Now then, if we impose the condition equivalent to \eqref{conf-flat-nonlinear} for the linearized case, we have
\begin{align}
q'(e^1, e^1) - q(e^2, e^2) =0,\quad \text{and} \quad q' (e^1, e^2) =0, \forall t
\end{align}
\fi 
Then, we have the following conditions on the gauge-transformed perturbed metric $g$ as follows: 
\begin{subequations}
\begin{align}
\mathcal{L}_Y \bar{g} (e_1, e_1) -\halb q_{0} (e_1, e_1) q^{ab}_0 \mathcal{L}_{Y} \bar{g} _{ab} =& - (\bar{g}' (e_1, e_1) - \halb q_0 (e_1, e_1) q^{ab}_0 \bar{g}'_{ab} )  \\
\mathcal{L}_Y \bar{g} (e_1, e_2) =& - \bar{g}' (e_1, e_2)  \\
\mathcal{L}_Y \bar{g} (e_2, e_2) -\halb q_{0} (e_1, e_2) q^{ab}_0 \mathcal{L}_{Y} \bar{g} _{ab} =& - (\bar{g}' (e_1, e_2) - \halb q_0 (e_1, e_1) q^{ab}_0 \bar{g}'_{ab} ) 
\end{align}
\end{subequations}

which can further be expressed as, using \eqref{WP-def},

\begin{subequations}
\begin{align}
\mathcal{L}_Y q_0 (e_1, e_1) -\halb q_{0} (e_1, e_1) q^{ab}_0 \mathcal{L}_{Y} (q_0) _{ab} =& - \vert \Phi \vert e^{-2\nu}(\bar{g}' (e_1, e_1) - \halb q_0 (e_1, e_1) q^{ab}_0 \bar{g}'_{ab} )  \\
\mathcal{L}_Y q_0 (e_1, e_2) =& - \vert \Phi \vert e^{-2\nu} \bar{g}' (e_1, e_2)  \\
\mathcal{L}_Y q_0 (e_2, e_2) -\halb q_{0} (e_1, e_2) q^{ab}_0 \mathcal{L}_{Y} (q_0) _{ab} =& - \vert \Phi \vert e^{-2\nu}(\bar{g}' (e_1, e_2) - \halb q_0 (e_1, e_1) q^{ab}_0 \bar{g}'_{ab} ) 
\end{align}
\end{subequations}

 The above system can be expressed compactly  in a covariant form as follows
\begin{align}
(\mathcal{L}_{Y} q_0)_{ab} - \halb (q_0)_{ab}\, (q_0)^{cd} (\mathcal{L}_{Y} q_0)_{cd}=-\vert \Phi \vert e^{-2\mbo{\nu}} (\bar{g}_{ab} - \halb (q_0)_{ab} (q_0)^{cd} \bar{g}'_{cd})
\end{align} 
We have established the following:
\begin{lemma}
	Suppose the perturbations of Einstein's equations in axial symmetry are compactly supported in the harmonic gauge then
	\begin{enumerate}
		\item the gauge transformation vector field $Y,$ where $Y$ is the projection of the spacetime gauge transform using $P:$
		\begin{align}
		Y = P^* \bar{Y}, 
		\end{align} 
		then, 
		\begin{align} \label{2d-gauge-transform}
		\mathcal{L}_Y (\bar{\mu}_q q) ^{ab} = \mathcal{L}_Y (\bar{\mu}_{q_0} q_0) ^{ab}= \bar{\mu}_q q^{ac} q^{bd} \vert  \Phi \vert ( \bar{g}'_{cd} - \halb q_{cd} q^{ef} \bar{g}'_{ef}) 
		\end{align}
	\item Suppose $\bar{g}'$ is such that it is compactly supported away from the horizon $\mathcal{H^+}$ and the spatial infinity $\iota^0$,  then $Y$	is a conformal Killing vector field in $(\Sigma, q)$ i.e., $\textnormal{CK}(Y, q) = 0$ in the asymptotic regions (i.e., in the complement of the support of $\bar{g}'$, $\Sigma \setminus Supp(\bar{g}) \vert_{\Sigma}$ )

	\end{enumerate} 
\end{lemma}

Now let us choose a gauge for the dimensionally reduced Cauchy hypersurface $(\Sigma, q)).$ If we choose the polar coordinates $(R, \theta)$, the condition \eqref{2d-gauge-transform} can be expressed as follows: 

\begin{subequations}\label{Y-theta-transport}
\begin{align}
\ptl_ \theta  Y^\theta =&  R \ptl_R \frac{Y^R}{R} + \frac{1}{R} (\bar{\mu}_q q^{ac} q^{bd} \vert  \Phi \vert ( \bar{g}' - \halb q_{ab} q^{ef} \bar{g}'_{ef})  ) \\
\ptl_R Y^\theta =& - \frac{1}{R} \ptl_\theta \frac{Y^R}{R} - \frac{1}{R} (\bar{\mu}_q q^{ac} q^{bd} \vert  \Phi \vert ( \bar{g}' - \halb q_{ab} q^{ef} \bar{g}'_{ef}) )
\end{align}
\end{subequations}

It may be noted that above system of differential equations is an overdetermined system. It follows from the Picard theorem and the Frobenious theorem that the necessary and sufficient conditions for the existence of the solutions is the compatibility condition

\begin{align}
\frac{1}{R} \ptl_R \left(R \ptl_R \frac{Y^R}{R} \right) + \frac{1}{R^2} \ptl^2 _{\theta} \frac{Y^R}{R} =& -\frac{1}{R^2} \ptl_ \theta ( \bar{\mu}_q q^{aR} q^{b \theta} \vert  \Phi \vert ( \bar{g}'_{ab} - \halb q_{ab} q^{ef} \bar{g}'_{ef}) )  \notag\\
&- \frac{1}{R} \ptl_R ( \frac{1}{R} \bar{\mu}_q q^{ac} q^{bd} \vert  \Phi \vert ( \bar{g}' - \halb q_{ab} q^{ef} \bar{g}'_{ef}) ).
\end{align}
In formal terms, this corresponds to vanishing of the commutator of the vector fields corresponding to the differential equations \eqref{Y-theta-transport}. 
It may be noted that the compatibility condition fortuitously turns out to be a Poisson equation for $\frac{Y^R}{R}$, for the Laplacian $\Delta_0$ in the $R, \theta$ gauge, 

\begin{align}
\Delta_0 =  \frac{1}{R} \frac{\ptl}{ \ptl R} \left( R \frac{\ptl }{\ptl R} \right) + \frac{1}{R^2} \frac{\ptl^2}{ \ptl \theta^2}. 
\end{align}

Likewise, the equations, which are equivalent to the equation in the Lemma,  can be tansformed into overdetermined transport equations for $\frac{Y^R}{R}$

\begin{align}
\ptl_\theta \frac{Y^R}{R} =& - \ptl_R Y^\theta -  (\bar{\mu}_q q^{ac} q^{bd} \vert  \Phi \vert ( \bar{g}' - \halb q_{ab} q^{ef} \bar{g}'_{ef}) ) \\
\ptl_R \frac{Y^R}{R} =& \frac{1}{R} \ptl_\theta Y^\theta - \frac{1}{R^2} (\bar{\mu}_q q^{ac} q^{bd} \vert  \Phi \vert ( \bar{g}' - \halb q_{ab} q^{ef} \bar{g}'_{ef}) )
\end{align} 
for which the compatibility condition is 
\begin{align}
\frac{1}{R} \ptl_R ( R \ptl_R Y^\theta) + \frac{1}{R^2} \ptl^2_\theta Y^\theta = & - \frac{1}{R} \ptl_R ( \frac{1}{R} \bar{\mu}_q q^{Rc} q^{\theta d} \vert  \Phi \vert ( \bar{g}' - \halb q_{ab} q^{ef} \bar{g}'_{ef}) )
 \notag\\
 &+ \frac{1}{R^3} \ptl_\theta  ( \frac{1}{R} \bar{\mu}_q q^{Rc} q^{Rd} \vert  \Phi \vert ( \bar{g}' - \halb q_{ab} q^{ef} \bar{g}'_{ef}) )
\end{align}
which is again a Poisson equation for $Y^\theta.$ For the reasons of regularity on the axes, we impose the conditions \cite{O_Rinne_J_Stewart_2005}
\begin{align}
Y^ \theta =0, \quad \ptl_\theta Y^R =0, \quad \text{on the axes} \quad \Gamma.
\end{align}

In particular, we assume that the behaviour of $Y^\theta \sim \sin \theta$ and $\ptl_\theta Y^R \sim \sin \theta$ close to the axes $\Gamma$

\begin{equation}
\left.\begin{aligned}
\Delta_0 \frac{Y_R}{R}  =&-\frac{1}{R^2} \ptl_ \theta ( \bar{\mu}_q q^{aR} q^{b \theta} \vert  \Phi \vert ( \bar{g}'_{ab} - \halb q_{ab} q^{ef} \bar{g}'_{ef}) )  \notag\\
&- \frac{1}{R} \ptl_R ( \frac{1}{R} \bar{\mu}_q q^{ac} q^{bd} \vert  \Phi \vert ( \bar{g}' - \halb q_{ab} q^{ef} \bar{g}'_{ef}) ) \quad \text{on} \quad (\Sigma)\\
\frac{Y^R}{R} &= Y^R_{\mathcal{H}^+} \quad \text{on} \quad (\mathcal{H}^+).
\end{aligned}
\qquad \right\}
\qquad \quad \text{(D-BVP)}_{Y^R}
\end{equation}
We solve the Dirichlet problem  $\text{(D-BVP)}_{Y^R}$ above with the method of images. Let us first consider the Poisson equation:
\begin{align}
\Delta_0 u= F, \quad \text{in any Lipschitz domain $(\Sigma)$}
\end{align}
From elliptic theory, it follows that if $F \in C^{\infty} (\Sigma)$ for regular Dirichlet or Neumann boundary data, then $u \in C^{\infty} (\Sigma).$
Suppose $K_u$ is the fundamental solution such that 
\begin{align}
\Delta_0 K_u = \delta (x-x')
\end{align}
where $\delta$ is a Dirac-delta function with a Eucliean metric on $\Sigma.$
Then consider the quantity, 
\begin{align}
\ptl_a (u \ptl^a K_u - K_u \ptl^a u) =& (\ptl_a u \ptl^a K_u+ u \Delta K_u) - (\ptl_a K_u \ptl^a u + K_u \Delta u) \notag\\
=& u \Delta K_u - K_u \Delta u
\end{align}
and upon integration over the domain $\Sigma$, we get 
\begin{align}
&\int_{\Sigma} \ptl_a (u \ptl^a K_{u} - K_{u} \ptl^a u) = \int_{\ptl \Sigma} n \cdot (u\ptl^a K_{u} - K_{u} \ptl^a u) \notag \\
&=  \int_{\ptl \Sigma} u \ptl_n K_u - K_u \ptl_n u
= \int_{\Sigma} u\delta (x-x') - K_u F = u - \int_{\Sigma} K_u F
\end{align}
Therefore, the general representation formula for $u$ is

\begin{align} \label{gen-Poisson-representation-formula}
u = \int_{\Sigma} K_u F + \int_{\ptl \Sigma} u \ptl_n K_{u} - K_{u} \ptl_n u.
\end{align}
The formula \eqref{gen-Poisson-representation-formula} will be useful for us through out our work, in different contexts. We can tailor this general formula for both Dirichlet and Neumann boundary value problems. In the following, we shall discuss two configurations that would be particularly relevant for us. 
   
The orbit space $\Sigma$ geometry of Kerr black hole spaetime resembles that of the complement of  a half disk ( boundary representing the horizon) in a half plane. We solve the Dirichlet problem with the method of images. The regularity and compatibilty conditions for our problem imply that the `image' is reflection antisymmetric. This applies for the image charge  as well as the Dirichlet data. Thus, with this picture, we have the complement of a full disk in a full plane, with reflection (with respect to the axes) antisymmetric data at the disk. It follows that the asymptotic decay rate for this problem is $\mathcal{O} (\frac{1}{R})$ for this Dirichlet problem. This decay rate can be independtly verified using the separation of variables. It may be noted that this decay rate is faster than that of the Poisson equation in a plane ($\frac{1}{2 \pi} \log R$ asymptotic behaviour). This faster decay rate plays a fundamental role in our problem. 

Likewise, consider the Poisson equation with Neumann boundary conditions in the orbit space $\Sigma$. The regularity at the axes implies that the Neumann data in the extended picture is reflection anti-symmetric. We thus recover the $\mathcal{O} (\frac{1}{R})$ decay rate  for the solution (with appropriate decay conditions for the source function $f$).  
 
%which has the decay rate of $\frac{1}{R}$ as opposed to the $\log R$ asymptotic behaviour  of the fundamental solution in the usual 2-plane.

\begin{align}
-\frac{1}{R^2} \ptl_ \theta ( \bar{\mu}_q q^{aR} q^{b \theta} \vert  \Phi \vert ( \bar{g}'_{ab} - \halb q_{ab} q^{ef} \bar{g}'_{ef}) ) &  \notag\\
- \frac{1}{R} \ptl_R ( \frac{1}{R} \bar{\mu}_q q^{ac} q^{bd} \vert  \Phi \vert ( \bar{g}' - \halb q_{ab} q^{ef} \bar{g}'_{ef}) )  &
\end{align}
are compactly supported (or with appropriate decay rate, consistent with asymptotically flat conditions) for our problem, 
it follows that the solution of the Dirichlet problem $(\text{D-BVP})_{Y^R}$ decays as 
$\frac{Y^R}{R} \sim \frac{1}{R}$ asymptotically, for large $R.$ It follows from analogous arguments that the solutions for the boundary value problem

\begin{equation}
\left.\begin{aligned}
\Delta_0 Y^\theta =& - \frac{1}{R} \ptl_R ( \frac{1}{R} \bar{\mu}_q q^{Rc} q^{\theta d} \vert  \Phi \vert ( \bar{g}' - \halb q_{ab} q^{ef} \bar{g}'_{ef}) )
\notag\\
&+ \frac{1}{R^3} \ptl_\theta  ( \frac{1}{R} \bar{\mu}_q q^{Rc} q^{Rd} \vert  \Phi \vert ( \bar{g}' - \halb q_{ab} q^{ef} \bar{g}'_{ef}) ), \quad \text{on} \quad (\Sigma)\\
Y^\theta &= Y^\theta_{\mathcal{H}^+} \quad \text{on} \quad (\mathcal{H}^+)
\end{aligned}
\qquad \right\}
\qquad \quad \text{(D-BVP)}_{Y^\theta}
\end{equation}

are unique, regular (well-posed) and decay $Y^\theta \sim \frac{1}{R}$ asymptotically for large $R$.  
On the other hand, due to the regularity conditions $Y$ admits the expansion: 

\begin{subequations}
	\begin{align}
	Y^{\theta}  =& \sum^{\infty}_{n=1} Y^{\theta}_n \sin (n \theta) \\
	Y^{R} =& \sum^{\infty}_{n=0} Y^R_n \cos(n \theta)
	\end{align}
\end{subequations}
for the solutions of the conformal Killing vector $Y$. Likewise, for regularity reasons, the behaviour of homogeneities in the boundary value problems mentioned above is restricted on the axes. In particular, they behave as follows

\begin{align}
 \bar{\mu}_q q^{RR} q^{\theta \theta} \vert \Phi \vert ( \bar{g}' ( \ptl_R, \ptl_\theta ) )\sim& \sin \theta \\
\ptl_R (\bar{\mu}_q q^{RR} q^{RR} \vert  \Phi \vert ( \bar{g}' ( \ptl_R, \ptl_R) - \halb q_{RR} tr_q ( \bar{g}')) ) \sim & \sin  \theta 
\end{align}
close to the axes $\Gamma$. Now that we clarified the structure of the source terms, let us introduce the notation, for $(R, \theta)$ coordinates

\begin{align}
\cal{M}^{ R \theta } \fdg =& \bar{\mu}_q q^{RR} q^{\theta \theta} \vert \Phi \vert \bgg'(\ptl_R , \ptl_\theta) \notag\\
\cal{M}^{R R} \fdg=& \bar{\mu}_q q^{R R} q^{RR} \vert \Phi \vert ( \bgg' ( \ptl_R, \ptl_R) - \halb q_{RR} \text{tr}_q \bgg')
\end{align}

 As a consequence of the above arguments, they must admit a Fourier decomposition of the form: 

\begin{align}
- \frac{1}{R} \mathcal{M}^{RR} =& \sum^\infty_{n = 0} I_n (R, t) \cos n \theta  \notag\\
-R \mathcal{M}^{R \theta} =&  \sum^\infty_{n=1} J_n (R, t) \sin n \theta
\end{align}

Now, plugging in these decompositions in the first order equations \eqref{Y-theta-transport}, we 
get 
\begin{align}
\ptl_R Y^R_0 (R, t) - \frac{1}{R} Y^R_0 (R, t) = I_0 (R, t)
\end{align}
which admits the solution, integrable equations 
\begin{align}
Y^R_0 (R, t) = \frac{R}{R_+} Y^R (R_+, t) + \int^R_{R_+} \frac{I (R', t)}{R'} dR'  
\end{align}
for the lowest frequency quantity. Now for higher frequencies, 

\begin{subequations}
\begin{align}
\ptl_R Y_n^R (R, t) - \frac{1}{R} Y^R_0 (R, t) - n Y^\theta_n (R, t) =&  I_n (R, t) \notag\\
R^2 \ptl_R Y^\theta _n (R, t) - n Y^R_n (R, t) = J_n (R, t)
\end{align}
\end{subequations}
which follow from the first-order equations respectively. These equations can be decoupled as 

\begin{align} \label{Y-first-order-modes}
\frac{1}{R} \ptl_R (R \ptl_R Y^\theta_n (R, t)) - \frac{n^2 Y^\theta_n(R, t)}{R^2} = \frac{n I_n (R, t)}{R^2} + \frac{1}{R} 
\ptl_R \frac{J_n (R, t)}{R}
\end{align}

The characteristic equation admits two real roots and it may be noted that the fundamental set of solutions is given by 
\begin{align}
 \text{fundamental solutions set for $Y^\theta$  in} \, \eqref{Y-first-order-modes} = \{ R^n, R^{-n} \}.
\end{align}
The corresponding Wronskian is $  \frac{2n}{R} \neq 0$ $\forall R \in (R_+, \infty)$,  $\to 0$ as $R \to \infty$ and $\to \frac{2n}{R_+}$ for $R \to R_+.$ It follows that $Y^\theta = A(t) R^{-n} + B(t) R^{-n}.$

We get the following asymptotic behaviour of the Wronskian near the horizon $\mathcal{H}^{+}:$

\begin{align}
W_{\mathcal{H}^+} (Y_n^\theta)=& \frac{2n}{R_+} \notag\\
W_{\mathcal{H}^+} (\frac{Y_n^R}{R})=&  \frac{2n}{R_+} \quad \text{as} \quad R \to R_+ 
\end{align}
and near the outer asymptotic region,

\begin{align}
W_{\bar{\iota}^0} (Y_n^\theta) =& \frac{2n}{R}, \notag\\
W_{\bar{\iota}^0} (\frac{Y_n^R}{R}) =& \frac{2n}{R} \to 0 \quad \text{as} \quad  R \to \infty
\end{align} 
so that, for the  conformal Killing vector $Y,$ the asymptotic behaviour is
\begin{subequations} \label{Y-theta-asym}
\begin{align}
	Y_n^ \theta (R, t) =& Y_n^\theta (t) (R^n - R^{2n}_+ R^{-n}), \, R \quad \text{near} \quad R_+, \quad n \geq 1 \notag \\ 
	& \to 0 \quad \text{as} \quad R \to R_+ \\ 
	Y_n^\theta (R, t) =& Y_n^\theta (t) R^{-n} \quad \text{for large} \quad R, \quad n \geq 1  \notag\\
	& \to 0, \quad \text{as} \quad R \to \infty 
\end{align}
\end{subequations}

and 

\begin{subequations}  \label{Y-R-asym}
	\begin{align}
	Y_n^ R (R, t) =& Y_n^R (t) (R^{n+1} - R^{2n}_+ R^{-n+1}), \, R \quad \text{near} \quad R_+, \quad n \geq 1 \notag \\ 
	& \to 0 \quad \text{as} \quad R \to R_+ \\ 
	Y_n^R (R, t) =& Y_n^R (t) R^{-n+1} \quad \text{for large} \quad R, \quad n \geq 1  \notag\\
\frac{Y_n^R}{R}	 \to & \, 0, \quad \text{as} \quad R \to \infty
	\end{align}
\end{subequations}

For the estimates in this work, the quantities $Y^\theta (t), Y^R (t) $ in the right hand sides of \eqref{Y-theta-asym} \eqref{Y-R-asym} are treated as constants (in each $\Sigma_t$) and thus there is a slight abuse of notation. The behaviour of $Y^R_0 (R, t)$ is a bit subtle and it is directly related to the regularity issues of our problem. This will be studied separately later.

%\subsection*{Expressions for the phase space $X$ in the asymptotic region/ Pure Gauge Perturbations}
\subsection*{Wave map phase space $X$}
The general gauge transforms of the quantitites
look like
\begin{subequations}
\begin{align}
\vert \Phi \vert' =& \vert \mbo{\Phi} \vert' + \mathcal{L}_{\bar{\text{Y}}}\vert \Phi \vert \\
\intertext{subsequently the wave map canonical pairs}
U'^A =& \mbo{U}'^A + \mathcal{L}_{\bar{\text{Y}}} U^A, \\
p'_A =& \mbo{p}'_A + \mathcal{L}_{\bar{\text{Y}}} p_A, \quad \forall A
\intertext{likewise, the (spacetime) gauge transform of the metric on the target} 
h'_{AB} (U)=& \mbo{h}'_{AB} (\mbo{U}) + \mathcal{L}_{\bar{\text{Y}}} h_{AB}(U), \quad \forall A, B
\end{align}
\end{subequations}
which is analogous to the transformation of a scalar. 
In the case of the axially symmetric and stationary Kerr black hole spacetime the operator $\mathcal{L}_{\bar{Y}} = \mathcal{L}_{\bar{Y}} \vert_{\Sigma}.$ As a consequence, the formulas above reduce to 
\begin{align}
 U'^A =& \mbo{U}'^A +\mathcal{L}_{\bar{\text{Y}}} U^A, \\
 p'_A =& \mbo{p}'_A + \mathcal{L}_{\bar{\text{Y}}} p_A
 \end{align}
In the asymptotic regions we have the above formulas reduce to 
\begin{align} U'^A =&  \mathcal{L}_{\bar{\text{Y}}}\vert_{\Sigma} U^A, \\
p'_A =& \mathcal{L}_{\bar{\text{Y}}}\vert_{\Sigma} p_A
\end{align}
After noting that, in the $(R, \theta)$ coordinates for $(\Sigma, q)$
\begin{align}
&\ptl_R  \vert \Phi \vert \sim \mathcal{O}(R) , \quad \ptl_\theta \vert \Phi  \vert \sim \mathcal{O}(R^2), \quad \text{for large $R$}  \\
& \ptl_R \vert  \Phi \vert \sim \mathcal{O}(1), \quad \ptl_\theta \vert \Phi \vert \sim \mathcal{O}(1), \quad \text{for $R$ close to $R_+$} 
\end{align}

%In the special gauge of the target, $h = 4 d\gm^2 + e^{4 \gm} d \omega^2, $  we have 
As a consequence, we have 
\begin{align}
\vert  \Phi \vert' \sim \mathcal{O} (\frac{1}{R}) \quad \text{for large $R$} \notag\\
\vert  \Phi \vert' \sim \mathcal{O}(1) \quad \text{for $R$ close to $R_+$}
\end{align}

Next, the other component of the wave map $U \fdg (M, g) \to (N, h)$ is constituted by the twist potential. It follows from background Kerr geometry that
\begin{subequations}
\begin{align}
	&\ptl_R  \omega \sim \mathcal{O}(\frac{1}{R^3} ) , \quad \ptl_\theta \omega \sim \mathcal{O}(1), \quad \text{for large $R$}  \\
	& \ptl_R  \omega\sim \mathcal{O}(1), \quad \ptl_\theta  \omega \sim \mathcal{O}(1), \quad \text{for $R$ close to $R_+$} 
\end{align}
\end{subequations}
and 
\begin{subequations}
\begin{align}
\vert  \omega \vert' \sim \mathcal{O} (\frac{1}{R}) \quad \text{for large $R$} \\
\vert  \omega \vert' \sim \mathcal{O}(1) \quad \text{for $R$ close to $R_+$}.
\end{align}
\end{subequations}

Now, then let us turn to the conjugate momenta, we have, from the Hamiltonian equation, 

\begin{align}
\frac{N}{ \bar{\mu}_q}p_A  = h_{AB} (U) \ptl_t U'^B - h_{AB}(U) \mathcal{L}_{N'} U^B
\intertext{in the asymptotic regions}
  = h_{AB} (U) \ptl_t \mathcal{L}_Y U^B - h_{AB}(U) \mathcal{L}_{N'} U^B. 
\end{align}

%$Y^t$ can be expressed as, $Y^\phi$ can be expressed as... This is connected to the geodesic completeness for our problem. 

%The gauge transformation aspects of the terms  conformal factor $\nu$ and $\nu'$ are bit more subtle and we shall discuss it later. 

 \subsection*{Lagrange multipliers}
 Now let us turn to the remaning quantities that occur in the ADM formalism, Lagrange multipliers $\{ N', N'^a \}$ in our Weyl-Papapetrou gauge, as constructed using a gauge transform from harmonic coordinates. It may noted that, for our background Kerr metric, $\bar{\gg}' (\ptl_\phi, \ptl_\phi) = \vert \Phi \vert'.$
 
 \begin{align}
 	N =& \vert \Phi \vert^\halb  (- \bar{g} (dt, dt))^{-\halb}
 	\intertext{noting that}
 	\vert \Phi \vert '=& \bgg' (\ptl_\phi, \ptl_\phi) + Y^\a \ptl_\a \vert  \Phi \vert
 	\intertext{and}
\bar{g}' (\ptl_t, \ptl_t)=&  	\bgg( d t, d x^\a) \bgg (d t, d x^\b)  \Big(\bgg' ( \ptl_a, \ptl_\b) + \mathcal{L}_Y \bgg(\ptl_\a, \ptl_\b ))
 	\intertext{we get}
 	N'= & \halb \vert  \Phi \vert^{-1} N (\bgg ( \ptl_\phi, \ptl_\phi) + Y^\a \ptl_\a \vert  \Phi \vert) \notag\\
 	&- \halb \vert  \Phi \vert^\halb N^3  \Big( \bgg( d t, d x^\a) \bgg (d t, d x^\b)  \Big(\bgg' ( \ptl_a, \ptl_\b) + \mathcal{L}_Y \bgg(\ptl_\a, \ptl_\b )\Big) \Big). \label{pert-lapse}
 	\end{align}
 %In view of, 

%\begin{align}
%	\bgg^{0\a} \bgg^{0 \b} \mathcal{L}_{Y} \bgg_{\a \b} =  ( \frac{}) 
	%\end{align}
%we get, 
In the asymptotic regions,

\begin{align}
	N' = Y^a \ptl_a N + N \ptl_t Y^t.
	\end{align}
 Likewise, for the shift vector, we have
 
 \begin{align}
  ( \vert  \Phi \vert^{-1} q_{ab}+ \vert \Phi \vert \mathcal{A}_a \mathcal{A}_b)\bar{N}^b =& \vert  \Phi \vert q_{ab} N^b + \vert  \Phi \vert \mathcal{A}_t \mathcal{A}_a
  \intertext{consequently}
  \vert  \Phi \vert^{-1} \bar{N}'^b =& \vert \Phi  \vert^{-1} N'^b + q^{ab} \vert \Phi\vert \mathcal{A}_t \mathcal{A}'_a .
  \intertext{Recall}
\bar{N}^b = &\bar{q}^{ab} \bar{g}_{0a} = \vert \Phi \vert q^{ab} \bar{g}_{0a}
  \intertext{then, }
   N'^b=& \bar{N}'^b + q^{ab} \vert \Phi \vert^2 \mathcal{A}_t \mathcal{A}'_a  \notag\\
   =& \vert  \Phi\vert q^{ab} \bar{g}'_{ta} + q^{ab} \vert \Phi \vert \mathcal{A}_t \bar{g}' _{a \phi }\notag\\
   =&  \vert  \Phi\vert q^{ab} (\bgg'_{ta} + (\mathcal{L}_Y \bar{g})_{ta}) + q^{ab} \vert \Phi \vert \mathcal{A}_t (\bgg'_{a \phi} + (\mathcal{L}_Y \bar{g})_{a \phi} ) 
 	\end{align} 
 Now then, using, 
 
 \begin{align}
 	(\mathcal{L}_Y \bar{g})_{0a} = &\ptl_0 Y^b \bar{g}_{ba} + \ptl_a Y^t  (-N^2 + N_\phi N^\phi ) + \ptl_a Y^\phi N_\phi \notag\\
 	=& \ptl_0 Y^b \bar{g}_{ba} + \ptl_a Y^t (- N^2 + \vert  \Phi\vert\mathcal{A}^2_0 ) + \vert \Phi \vert \ptl_a Y^\phi \mathcal{A}_0 
 	\end{align}
 
 \begin{align}
 	(\mathcal{L}_{Y} \bar{g} )_{a \phi} =&   \vert \Phi \vert \ptl_a Y^\phi +  \ptl_a Y^t \bar{N}_\phi \notag\\
 	=& \vert \Phi\vert ( \ptl_a Y^\phi + \ptl_a Y^t \mathcal{A}_0)
 	\end{align}
 in the asymptotic regions,  we have
 \begin{align} \label{pert-shift}
 	N'^b = \ptl_t Y^b - N^2 q^{ab} \ptl_a Y^t.
 \end{align}
The above results can be summarized in the following lemma. 
 
 \begin{lemma} [Lagrange multipliers $\{ N', N'^a \}$]
 	Suppose $\bar{Y}$ is a gauge transform from the harmonic coordinates $(\bar{M}', \bar{\gg}')$ to the Weyl Papapetrou gauge $(\bar{M}', \bar{g}' )$ 
 	\begin{itemize}
 \item In the Weyl-Papapetrou gauge, the Lagrange multipliers $\{ N, N'^a\}$ in the asymptotic regions are given by \eqref{pert-lapse}and \eqref{pert-shift} respectively 
 \iffalse
 \begin{subequations}
 \begin{align}
 N'=& \halb N \vert \Phi \vert^{-1} \big( \bar{\gg}'_{\phi \phi} - N^2 \bar{g}^{t \a} \bar{g}^{t \b} \bar{g}'_{\a \b}  \big) + N \ptl_t \bar{Y}^t + Y^a \ptl_a N \\  
 N'^a =& q^{ab} ( \vert \Phi \vert (\bar{\gg}'_{0a} -\mathcal{A}_0 \bar{\gg}'_{0 \phi})) + \ptl_a \bar{Y}^c - N^2 q^{ac} \ptl_a \bar{Y}^0
 \end{align}
 \end{subequations}
\fi 
\item The vector field $N'^a$ so constructed is regular at the axes and behaves as, 
\begin{align}
N'^R = \mathcal{O}(1) , \quad N'^\theta = \mathcal{O}(\frac{1}{R}) \quad \text{near the spatial infinity}, 
\intertext{and}
N'^R \sim \mathcal{O}(1), \quad N'^\theta =0 \quad  \text{at the horizon}
\end{align}
\end{itemize}
\end{lemma}
\begin{proof}

In our gauge, we have 
\begin{align}
\ptl_t Y^t = - \frac{1}{N} Y^a\ptl_a N, \quad \forall t
\end{align}
in the asymptotic regions. 
Thus, 
\begin{align}
\ptl_b Y^t = - \ptl_b \left(\frac{1}{N} \int^t_{0} Y^a \ptl_a N dt' \right).
\end{align}
Now define a quantity $\mathcal{D}^a $ occuring in \eqref{pert-shift} as, 

\begin{align}
\mathcal{D}^a \fdg= N^2 q^{ab} \ptl_a Y^t
\intertext{then}
\mathcal{D}^a = - N^2 q^{ab}   \ptl_b \left(\frac{1}{N} \int^t_{0} Y^a \ptl_a N dt'\right). 
\end{align}

\iffalse
Recall that we have 

\begin{align}
	Y^t =& Y^t \Big \vert_{t=0} + \int^t_0 \frac{N'}{N} - \mathcal{Y}^a \frac{\ptl_a N}{N}
	\intertext{where}
\end{align}
\fi 
Recall that, 
\begin{subequations}
\begin{align}
	e^{2 \nu} =& \sin^2 \theta  \frac{(r^2 +a^2)^2 - a^2 \Delta \sin^2 \theta }{R^2} \\
	\frac{\ptl_R N}{N} =& \frac{1+ \frac{R^2_+}{R^2}}{R \left(1- \frac{R^2_+}{R^2} \right) } \sim (1-\frac{R_+}{R})^{-1} \quad \text{as} \quad R \to R_+ \notag\\ 
	&\to 0 \quad \text{as} \quad \mathcal{O} (\frac{1}{R}) \quad \text{as} \quad R  \to \infty,  \label{NR-asym}\\
	\frac{\ptl_\theta N}{N} =& \cot \theta  \label{Ntheta-asym}
		\intertext{and denote}
		\mathcal{Y}^a = \int^t_0 Y^a dt  .
\end{align}
\end{subequations}
Let us now compute the behaviour of $\mathcal{D}$ at various boundaries. We have the following expressions for the components of $\mathcal{D}.$ 
\begin{align}
	D^R  =& N^2 e^{-2 \mbo{\nu}} (q_0)^{RR} \ptl_R Y^t \notag \\
	%=& R^2 \sin^2 \theta \left( 1- \frac{R^2_+}{R^2}\right)^2 \cdot e^{-2 \mbo{\nu}} \cdot 1 \cdot \ptl_R Y^t \notag\\
	%=&  R^2 \sin^2 \theta \left( 1- \frac{R^2_+}{R^2}\right)^2 \cdot e^{-2 \mbo{\nu}}  \ptl_R \left( Y^t(t=0) - \mathcal{Y}^a \frac{\ptl_a N}{N} \right) \notag\\
	=&  \left( 1 - \frac{R^2_+}{R^2} \right)^2 \cdot \frac{R^4}{ ( (r^2+a^2)^2 -a^2 \Delta \sin^2 \theta)} \notag\\
	& \quad \ptl_R \left(Y^t(t=0) - \mathcal{Y}^R \frac{1+ \frac{R^2_+}{R^2}}{1-\frac{R^2_+}{R^2}} - \mathcal{Y}^\theta \cot \theta  \right) \notag\\
	& \to 0 \quad {as} \quad R \to \infty
	\intertext{at the rate of $\mathcal{O}(\frac{1}{R^2})$; and }
D^R	=& \mathcal{O}(1) \quad \text{as} \quad R \to R_+
\end{align}

Likewise, we have 
\begin{align}
	D^\theta =& N^2 e^{-2 \mbo{\nu}} (q_0)^{\theta \theta} \ptl_\theta Y^t \notag\\
	%=& N^2 e^{-2 \mbo{\nu}} \frac{1}{R^2} \ptl_\theta \left( Y^t (t=0) - \mathcal{Y}^R \frac{1 + \frac{R^2_+}{R^2}} {{1-\frac{R^2_+}{R^2}} }  - \mathcal{Y}^ \theta \cot \theta \right) \notag\\
	=&  (1- \frac{R^2_+}{R^2})^2 \cdot \frac{R^4}{ \left( (r^2+a^2)^2 -a^2 \Delta
		\sin^2 \theta  \right)}   \notag\\
	& \quad \ptl_\theta \left( Y^t (t=0) - \mathcal{Y}^R \frac{1 + \frac{R^2_+}{R^2}} {{1-\frac{R^2_+}{R^2}} }  - \mathcal{Y}^ \theta \cot \theta \right) 
	& \to 0 \quad \text{as} \quad R \to \infty 
	\intertext{at the rate of  $\mathcal{O}(\frac{1}{R^3})$}
	D^\theta   =& \mathcal{O} (1-\frac{R_+}{R}) \to 0 \quad \text{as} \quad R \to R_+.  
\end{align}
	\end{proof}
It now follows that the conjugate momenta, are given by, as follows,
\begin{align}
N \frac{p'_A}{\bar{\mu}_q} =& h_{AB} (U) \ptl_t \mathcal{L}_Y U'^A - h_{AB} (U) \mathcal{L}_{N'} \ptl_a U^b
\intertext{in the asymptotic regions}
=& h_{AB}(U) N^2 q^{ab} \ptl_b Y^t \ptl_a U^B \notag\\
=& h_{AB} (U) R^2 \sin^2 \theta (1-\frac{R^2_+}{R^2})^2 \cdot \frac{R^2}{ \sin^2 \theta ( ( r^2+a^2)^2-a^2 \Delta \sin^2 \theta)} \notag\\
& \cdot ( -\ptl_R U^B \ptl_R ( \mathcal{Y}^a \frac{\ptl_a N}{N})  - \frac{1}{R^2} \ptl_\theta U^B \ptl_\theta (\mathcal{Y}^a \frac{\ptl_a N}{N})) 
\end{align}
As a consequence, we can estimate the values of the conjugate momenta at various boundaries (cf. eqs \eqref{NR-asym}, \eqref{Ntheta-asym}). %If we choose the gauge \eqref{} on the target, we get explicitly

\begin{subequations}
\begin{align}
N \frac{p'^A}{ \bar{\mu}_q} =& \mathcal{O} (\frac{1}{R^4}) \quad \text{for large $R$}, \\
N \frac{p'^A}{\bar{\mu}_q} =& \mathcal{O} (1-\frac{R^2_+}{R^2} )\quad \text{for $R$ close to $R_+$}, \quad A= 1, 2. 
\end{align}
\end{subequations}
It can be verified that the decay rates and the boundary behaviour of the shift vector field agree  from  separate analysis using the momentum constraint. From the momentum constraint, we have, 

\begin{align}
- \leftexp{(q_0)}{\grad}_b \varrho^b_a + \halb p'_A \ptl_a U^A =0
\end{align}
after taking into account that the transverse-traceless tensors vanish for our geometry,  where, 
\begin{align}
\varrho^a_c = \bar{\mu}_{q_0} (\leftexp{(q_0)}{\grad}_c Y^a +  \leftexp{(q_0)}{\grad}^a Y_c - \delta^a_c \leftexp{(q_0)}{\grad}_b Y^b  )
\end{align}
which can be reduced to an elliptic equation for the shift vector.

Now, let us turn  our attention to the Hamiltonian constraint:  
\begin{align}
H =& \bar{\mu}^{-1}_{q_0} (e^{-2 \mbo{\nu}} \Vert \varrho \Vert^2_{q_0} - \halb \mbo{\tau}^2 e^{2 \mbo{\nu}} \bar{\mu}^2_{q_0} + \halb  p_A p^A) \notag\\
&+ \bar{\mu}_{q_0} (2  \Delta_0 \mbo{\nu} + \halb h_{AB} q_0^{ab} \ptl_a U^A \ptl_b U^B), \quad (\Sigma, q_0), 
\intertext{where}
\Delta_0 \mbo{\nu} \fdg=& \frac{1}{\bar{\mu}_{q_0}} \ptl_b(q^{ab}_0 \bar{\mu}_{q_0} \ptl_b \mbo{\nu}).
\end{align}
There are a few delicate aspects of the Hamiltonian constraint in our dynamical axisymmetric problem.  For the special case of the Kerr metric, 

\begin{align}
H =&  \bar{\mu}_{q_0} (2  \Delta_0 \mbo{\nu} + \halb h_{AB} q_0^{ab} \ptl_a U^A \ptl_b U^B), \quad (\Sigma, q_0).
\end{align}

If we consider  the quantity $\int H =0,$ it may be noted that the inner boundary term involves the $\ptl_R \mbo{\nu}$ at the horizon $\mathcal{H}^+$. This quantity vanishes for the Kerr black hole metric and we recover the positive mass theorem of Schoen-Yau \cite{schoen-yau-1, schoen-yau-2}. This is consistent with the (inner boundary) horizon being the minimal surface, which is the case with the Kerr black hole spacetime. 
Let us now analyze the linearized Hamiltonian constraint: 

\begin{align}
H' =&   \bar{\mu}_{q_0} (2 \Delta_0 \mbo{\nu}') + \halb \bar{\mu}_{q_0} \ptl_{U^C} h_{AB} q_0^{ab} \ptl_a U^A \ptl_b U^B U'^C \notag\\
&+ \bar{\mu}_{q_0} q_0^{ab} h_{AB} \ptl_a U'^A \ptl_b U^B , \quad (\Sigma)
\end{align}
which is an elliptic PDE. 
We would like to construct $\mbo{\nu}'$ such that:

\begin{proposition} \label{well-posedness-nu'}
	A boundary value problem for $\nu'$ admits a unique, regular and bounded solution that decays at the rate of $\mathcal{O}(\frac{1}{R})$ for large $R$. 
\end{proposition}

In order to set up a well-posed boundary value problem for this elliptic PDE, we need to specify appropriate boundary conditions. 

Suppose, $\Sigma$ is a 2-surface  that embeds into the  Cauchy hypersurface of the Kerr metric $(\Sigma \hookrightarrow) \olin{\Sigma}$ , the Gauss curvature of $\Sigma$  is given by: 

\begin{align}
K (\Sigma) = - \frac{1}{e^{2 \mbo{\nu}}} \Delta_0 \mbo{\nu}
\intertext{and the mean curvature $H$ of $\Sigma \hookrightarrow \olin{\Sigma}$ }
H_\Sigma = 0
\end{align}
on account of the fact that the inner boundary of $\Sigma$ is a minimal surface. In the perturbative theory, we need to find an expression of the mean curvature. 
\begin{align}
\vec{H}_{\Sigma}  = H_{\Sigma} \mbo{n}, \quad \text{mean curvature vector, $\Sigma \hookrightarrow \olin{\Sigma}$ }.
\end{align}
For the conformal transformation $q= \Omega^2 q_0,$ we have 

\begin{align}
H_\Sigma = \frac{1}{\Omega}  \left(H_0 + \frac{2}{\Omega} \frac{ \ptl \Omega}{\ptl \mbo{n}} \right).
\end{align}
It then follows that, 
\begin{align}
H_{\Sigma} =   \vert \Phi  \vert^{1/2} e^{-\mbo{\nu}} \left( \frac{1}{R} + \ptl_R \mbo{\nu}  
\right).
\end{align}
It may be noted that the preservation of the minimal surface condition at the inner boundary horizon $\mathcal{H}^+$ implies a Neumann boundary condition for $\mbo{\nu}'$ at the horizon $\mathcal{H}^+.$
Let us now turn to the issue of regularity of $\mbo{\nu}'$ at the axes $\Gamma.$ Firstly note that, from the form of the Kerr metric, the quantity $\mbo{\nu}'$ has the form 

\begin{align}
\mbo{\nu} = 2\gamma + c - 2 \log \rho  
\end{align} 

in the Weyl-Papapetrou gauge, near the axes. For the reasons of regularity of the Kerr metric at the axes, we have $c \equiv 0.$ Thus, 

\begin{align}
\mbo{\nu} = 2 \gamma - 2 \log \rho. 
\end{align}

Consequently, we have the condition 

\begin{align} \label{gamma-nu condition}
\mbo{\nu}' = 2 \gamma'
\end{align}
for the linear perturbation theory, thus suggesting the Dirichlet boundary conditions for $\mbo{\nu}'$ at the axes.

%where $\Delta_0$ is the flat-space Laplacian
%\begin{align}
%\Delta_0 = \frac{1}{R} \ptl_R (R \ptl_R) + \frac{1}{R^2}\ptl^2_{\theta} =  \ptl^2 _{\r \r} +  \ptl^2_{zz}
%\end{align} 
\begin{equation}
\left.\begin{aligned}
\Delta_0 \mbo{\nu}' =&\frac{1}{4} \ptl_{U^C} h_{AB}(U) q^{ab}_{0} \ptl_a U^A U'^C\\
& \quad + q_{ab} h_{AB} \ptl_a U^a \ptl_b U^B \quad \text{on}
& \quad (\Sigma)\\
\mbo{\nu}' =& 2\gamma', \quad \text{on}& \quad (\Gamma) \\
\ptl_n \mbo{\nu}' =&0 \quad \text{on}& \quad (\mathcal{H}^{+})
\end{aligned}
\qquad \right\}
\qquad \text{(M-BVP)}_{\mbo{\nu'}}
\end{equation}

It must be pointed out that for the regularity of $\mbo{\nu'}$ itself, at the axes, we need to impose the condition $\ptl_{\mbo{n}} \mbo{\nu}'=0$ at the axes. A priori, the apparent over-determined boundary data $-$both Dirichlet and Neumann$-$ at the axes is a significant issue for the well-posedness of the boundary value problem for $\mbo{\nu}'$. It may be recalled that an elliptic problem with both Dirichet and Neumann data, specified at the same boundary, is typically ill posed. Importantly, we can prove that the boundary conditions $\mbo{\nu}' = 2 \gamma'$ and $\ptl_{\mbo{n}} \mbo{\nu}'=0$ are equivalent. In other words, imposition of one condition automatically satisfies the other and vice versa.  This saves us from the aforementioned issue.

However, we should ask a more fundamental question: Can the set-up of our problem result in a well-posed boundary value problem for $\mbo{\nu}'$ (as proposed in Proposition \ref{well-posedness-nu'}) in the first place?  and what about the regularity of $\mbo{\nu}'$ at the corners $\Gamma \cap \mathcal{H}^+$? 

Suppose, $\mbo{\nu}' = 2 \gamma',$ then a computation shows that, 
\begin{align}
\ptl_{\mbo{n}} \mbo{\nu}' =  \ptl_R \mbo{\nu}'=  2 \ptl_R \gamma' =0
\end{align} 
at the corners, in a limiting sense. On the other hand, suppose $\mbo{\nu}'$ were a scalar, then $\mbo{\nu}' = (\mbo{\nu}')_{\text{HG}} + Y^a \ptl_a \nu = Y^a \ptl_a \mbo{\nu}$ in the asymptotic regions. Then a computation shows that 

\begin{align}
\ptl_R \mbo{\nu}' = \ptl_R Y^a \ptl_a \mbo{\nu} \neq 0.
\end{align}
In the above, there is an inconsistency for two reasons. Firstly, the quantity $\mbo{\nu}'$ is not regular at the corners. Secondly, the Neuman boundary condition at the horizon is not satisfied at the end points, the corners. The following lemma is crucial and comes to our rescue. In this lemma we show that $\mbo{\nu}'$ does not transform like a scalar.  %This outcome is perhaps attributed to the fact that $\mbo{\nu}$ occurs as both the metric function of $(\Sigma, q)$ and also in the Jacobian.  

\begin{lemma}
	Suppose the quantity $\mbo{\nu}'$ is compactly supported in the harmonic gauge, then
	\begin{enumerate}
	  \item $\mbo{\nu}'$ has the following structure 
	\begin{align}
	\mbo{\nu}' = Y^a \ptl_a \mbo{\nu} + \halb \leftexp{(q_0)}{\grad}_a Y^a \label{gauge-transform-nu'}
	\end{align}
	in the asymptotic regions. 
	\item  Furthermore, $\mbo{\nu}'$ is regular at the corners $\Gamma \cap \mathcal{H}^+$
	\end{enumerate}
\end{lemma}
\begin{proof}
	Consider $D \cdot \bar{\mu}_q$ in a conformally flat form, we have 
	\begin{align}
	D \cdot \bar{\mu}_q = e^{2 \mbo{\nu}} D \cdot \bar{\mu}_{q_0} + 2 e^{2 \mbo{\nu}} \mbo{\nu}' \bar{\mu}_{q_0} 
	\end{align}
	Now, then for our choice of the gauge condition, where we hold the flat metric $q_0$ fixed, 
	\begin{subequations}
	\begin{align}
	2 e^{2 \mbo{\nu}} \mbo{\nu}' \bar{\mu}_{q_0} =& (\mbo{\nu}')_{\text{HG}} + \mathcal{L}_Y \bar{\mu}_q  \notag\\
\intertext{in the asymptotic regions}
	=& 2 e^{2 \mbo{\nu}} Y^a \ptl_a \mbo{\nu} \bar{\mu}_{q_0} + e^{2 \mbo{\nu}} \ptl_a \bar{\mu}_{q_0} + e^{2 \mbo{\nu}} \ptl_a Y^a \bar{\mu}_{q_0} 
	\intertext{the right hand side can be expressed equivalently as}
	=& 2 e^{2 \mbo{\nu}} Y^a \ptl_a \mbo{\nu} \bar{\mu}_{q_0} + e^{2 \mbo{\nu}} \mathcal{L}_Y \bar{\mu}_{q_0} \\
	=&  2 e^{2 \mbo{\nu}} Y^a \ptl_a \mbo{\nu} \bar{\mu}_{q_0} + e^{2 \mbo{\nu}} \ptl_a (Y^a \bar{\mu}_{q_0}) \\
	=& 2 e^{2 \mbo{\nu}} Y^a \ptl_a \mbo{\nu} \bar{\mu}_{q_0} + e^{2 \mbo{\nu}} \bar{\mu}_{q_0} \leftexp{(q_0)}{\grad}_a Y^a
	\end{align} 
	\end{subequations}
	The result \eqref{gauge-transform-nu'} follows. 
	Now let us show that the `correction term' is harmonic in the asymptotic regions. We have, 
	\begin{align}
	\leftexp{(q_0)}{\grad}_a  \leftexp{(q_0)}{\grad}_b Y_c - \leftexp{(q_0)}{\grad}_b \leftexp{(q_0)}{\grad}_a Y_c = \leftexp{(q_0)} R^{\quad d}_{abc} Y_d 
	\end{align}
	on account of the fact that $Y$ is a conformal Killing vector field in the asymptotic regions, $\text{CK}(Y, q_0)=0$, it follows that 
	\begin{align}
	\leftexp{(q_0)}{\grad}^a \leftexp{(q_0)}{\grad}_a ( \leftexp{(q_0)}{\grad}_c Y^c ) =0.
	\end{align}
	
	With the correction term in \eqref{gauge-transform-nu'}, the assertion (2) can be also be verified explicitly.  It can also be verified that the Neumann boundary condition for $\mbo{\nu}'$ at the horizon $\mathcal{H}^+$ is also satisfied. 
\end{proof}

Firstly, we note that the mixed boundary value problem, with regular boundary data, is well-posed. 
\iffalse
\begin{claim}
	the solution is unique
\end{claim}
\fi 
We would like to remark that `uniqueness upto a constant' which is usually the case for a Neumann boundary value problem is not sufficient for our problem because we need the decay of $\mbo{\nu}'$ (consistent with the regularity on the axes) to establish the needed boundary behaviour of our fields. 
In view of the uniqueness, we claim that the solution to the mixed boundary value problem above can be decomposed into the following two parts: solution for a inhomogeneous Dirichlet boundary value problem $\mbo{\nu'}_D$ and a homogeneous Neumann problem $\mbo{\nu}'_N$: 

\begin{align}
\mbo{\nu}' = \mbo{\nu}'_D + \mbo{\nu}'_N
\end{align} 
represented in the same domain. In the following, we shall elaborate on our construction that ensures that the needed boundary conditions for $\mbo{\nu}'$ are satisfied. We specify the (regular) data for Dirichlet problem for $\mbo{\nu}'_D$ in a half plane such that it is consistent with the data for $\mbo{\nu}'$. This problem is well-posed on account of the regularity conditions of the source (involves $U'$, $h(U)$, $h'(U)$ ) and admits a regular solution. We shall use this to obtain a suitable choice of data for the Neumann problem for $\mbo{\nu}'_N$ at the horizon, so that the total sum $\mbo{\nu}'_D + \mbo{\nu}'_N$ satisfies the needed Neumann boundary condition for $\mbo{\nu}'.$ The well-posedness and regularity for $\mbo{\nu}'$ at the axes implies that it should vanish on the axes, thus implying that $\mbo{\nu}'$ satisfies the needed Dirichlet condition on the axes. This construction also ensures that the regularity of $\mbo{\nu}'$ at the corners is also satisfied. Consider, 

\begin{equation} \label{Dirichlet-nu'}
\left.\begin{aligned}
\Delta_0 \mbo{\nu}_D' =&\frac{1}{4} \ptl_{U^C} h_{AB}(U) q^{ab}_{0} \ptl_a U^A U'^C\\
& \quad + q_{ab} h_{AB} \ptl_a U^a \ptl_b U^B \quad \text{on}
& \quad (\Sigma)\\
\mbo{\nu}' =&2 \gamma', \quad \text{on}& \quad (\Gamma) \\
\mbo{\nu}' =&2 \gamma' \quad \text{on}& \quad (\{ \rho=0\} \setminus \Gamma)
\end{aligned}
\qquad \right\}
\qquad \text{(D-BVP)}_{\mbo{\nu'}}
\end{equation}
It may be noted that, for our construction, the data at the `cut' between the axes $\{ \rho =0\} \setminus \Gamma$  can be specified arbitrarily so that its smooth but for convenience we choose $ \mbo{\nu}'_D = 2 \gamma'.$ It follows that well posed Dirichlet boundary value problem \eqref{Dirichlet-nu'} is well posed and that there exists a regular solution. We shall use this to construct Neumann data for $\mbo{\nu}'_N$. Note that we can compute $\ptl_{\mbo{n}} \mbo{\nu}'_D$ at the horizon (boundary) using the 'Dirichlet to Neumann map' 

   %whose  explicit representation can be deduced from \eqref{gen-Poisson-representation-formula}.  It may be noted that, effectively, `Dirichlet data' 

\begin{align}
\Lambda \fdg \mbo{\nu'} \to \ptl_{\mbo{n}} \mbo{\nu}'
\end{align}
 and the conformal inversion map discussed previously. We then formulate the homogeneous Neumann boundary value problem
\begin{equation}
\left.\begin{aligned}
\Delta_0 \mbo{\nu}_N' =&0 \quad \text{on}
& \quad (\Sigma)\\
\ptl_n \mbo{\nu}'_N =&0, \quad \text{on}& \quad (\Gamma) \\
\ptl_n \mbo{\nu}'_N =& \ptl_{\mbo{n}} \mbo{\nu}'_D \big\vert_{\mathcal{H}^+} \quad \text{on}& \quad (\mathcal{H}^{+})
\end{aligned}
\qquad \right\}
\qquad \text{(N-BVP)}_{\mbo{\nu'}}
\end{equation}
It follows from the use of the representation formulas \eqref{gen-Poisson-representation-formula} for both Dirichlet and Neumann problems, that $\mbo{\nu}'$ decays at the rate of $\mathcal{O}(\frac{1}{R})$ for large $R.$

\begin{proposition}
	The integral invariant quantity $Y_0(\mathcal{H}^+)$   at the horizon is finite  for all times if and only if it vanishes
\end{proposition}

\begin{proof}
	Recall the expression for $Y_0^R (R_+)$
	\begin{align}
	Y^R_0 (R_+) =&  \frac{R_+}{ 2 \pi} \int^\infty_{R_+}  \frac{1}{R'^2} \int^{2 \pi}_{0} \frac{R'}{2} e^{2 \gamma - 2 \mbo{\nu}} \left( \bgg' ( \ptl_R, \ptl_R ) - \frac{1}{R'^2}  \bgg' (\ptl_\theta, \ptl_\theta) \right) d \theta dR'
	\intertext{which can be reexpresed as}
	=& \frac{R_+}{ 2 \pi} \int^\infty_{R_+}  \frac{1}{R'^2} \int^{2 \pi}_{0} \frac{R'}{2} \left( \mbo{q}'_0 ( \ptl_R, \ptl_R ) - \frac{1}{R'^2}  \mbo{q}'_0 (\ptl_\theta, \ptl_\theta) \right) d \theta dR' \label{integral-invariant-expression}
		\end{align}
		where $q_0$ is the metric perturbation in harmonic gauge.

The quantities, $\mbo{q}'_0 ( \ptl_R, \ptl_R), \mbo{q}'_0 ( \ptl_R, \ptl_\theta)$ and $ \mbo{q}'_0 (\ptl_\theta, \ptl_\theta)$ admit the decomposition, 

\begin{align}
\mbo{q}'_0 ( \ptl_R, \ptl_R) =& \sum^{\infty}_{n=0} \Big \{ \mbo{q}'_0 ( \ptl_R, \ptl_R)  \Big\}_n  \cos n \theta \\
\mbo{q}'_0 ( \ptl_R, \ptl_\theta) =&  \sum^{\infty}_{n=1} \Big \{ \mbo{q}'_0 ( \ptl_R, \ptl_\theta)  \Big\}_n  \sin n \theta  \\
\mbo{q}'_0 ( \ptl_\theta, \ptl_\theta)=& \sum^{\infty}_{n=0} \Big \{ \mbo{q}'_0 ( \ptl_\theta, \ptl_\theta)  \Big\}_n  \cos n \theta 
\end{align}
Now then \eqref{integral-invariant-expression} can be simplified as 

\begin{align}
Y^R_0 (R_+) =& \frac{R_+}{2} \int^\infty_{R_+} \frac{1}{R'} \left(  \Big\{ \mbo{q}'_0 (\ptl_R, \ptl_R ) \Big\}_0 (R') - \frac{1}{R'^2} \Big \{ \mbo{q}'_0 (\ptl_\theta, \ptl_\theta ) \Big\}_0 (R') \right) dR'
\notag\\ 
=&  \frac{R_+}{2} \int^\infty_{R_+} \frac{1}{R'} \left(  \mbo{q}'_0 (\ptl_R, \ptl_R ) \Big\}_0 (R') - \ptl_{R'} \left( \frac{1}{R'} \Big\{  \mbo{q}'_0 (\ptl_\theta, \ptl_\theta ) \Big\}_0 (R')  \right)   \right)
dR' \notag\\
&\quad +\frac{R_+}{2} \int^\infty_{R_+} \ptl_{R'}  \left( \frac{1}{R'^2}  \Big \{ \mbo{q}'_0 (\ptl_\theta, \ptl_\theta ) \Big\}_0 (R') \right) dR' \label{integral-invariant-II}
\end{align}

Now consider \eqref{flatness-condition} and note that the operator on the LHS is a linear differential operator. In polar coordinates, 

\begin{align}
& -\frac{1}{R} \ptl^2_\theta \mbo{q}_0 (\ptl_R, \ptl_ R) + \ptl_R \mbo{q}_0 ( \ptl_R, \ptl_R) + \frac{2}{R} \ptl^2 R \theta \mbo{q}_0 (\ptl_R, \ptl_ \theta ) - \frac{2}{R^3} \mbo{q}_0 ( \ptl_\theta, \ptl_\theta)  \notag\\
& + \frac{2}{R^2} \ptl_R \mbo{q}'_0 ( \ptl_\theta, \ptl_\theta) - \frac{1}{R} \ptl^2_{RR} \mbo{q}_0 ( \ptl_\theta, \ptl_\theta) =0.
\end{align}

It follows that, 

\begin{align}
&\ptl_R  \Big \{ \mbo{q}'_0 (\ptl_R, \ptl_R) \Big\}_0 - \frac{2}{R^3}  \Big \{ \mbo{q}'_0 (\ptl_\theta, \ptl_\theta ) \Big\}_0  \notag\\ 
& \quad + \frac{2}{R^2} \ptl_R  \Big \{ \mbo{q}'_0 (\ptl_\theta, \ptl_\theta ) \Big\}_0 - 
\frac{1}{R} \ptl^2_R  \Big \{ \mbo{q}'_0 (\ptl_\theta, \ptl_\theta ) \Big\}_0  =0
\end{align}

which is a pure divergence, 

\begin{align}
\ptl_R \left(   \Big \{ \mbo{q}'_0 (\ptl_R, \ptl_R) \Big\}_0 - \ptl_R \left(  \frac{1}{R}  \Big \{ \mbo{q}'_0 (\ptl_\theta, \ptl_\theta ) \Big\}_0 \right)    \right) = 0, 
\end{align}
uniformly in $\Sigma$ for all times. We get the condition that the $R$ derivative of the quantity in the first integral of \eqref{integral-invariant-II} vanishes. As a result this quantity can make a finite contribution only if it vanishes. The result follows. 
\end{proof}

\section{Boundary Behaviour of the Constrained Dynamics in the Orbit Space and Strict Conservation of the Regularized Hamiltonian}
%We systematically do the analysis for the independent canonical variables and then do the analysis for the dependent canonical variables. We demonstrate that our construction and the boundary and regularity conditions are such that the dynamical flux terms vanish globally in time.   
%Based on the framework developed above, in the following, we shall systematically and carefully estimate all the boundary terms that occur in our problem, for the following gauge on the target: 
In this section, we shall establish that the regularized Hamiltonian energy $H^{\text{Reg}}$ is strictly conserved in time. This is the main result of the current article. In particular, we shall establish that: 
\begin{theorem}
	Suppose we have the initial value problem of Einstein's equations for general relativity, 
	\begin{enumerate}
		\item There exists a ($C^\infty-$)diffeomorphism from harmonic coordinates to the Weyl-Papapetrou gauge of the maximal development $(M', g')$ of the perturbative theory of Kerr black hole spacetimes, under axial symmetry. 
		
		\item The positive-definite Hamiltonian $H^{\textnormal{Reg}}$ is strictly-conserved forwards and backwards in time.  
	\end{enumerate}
\end{theorem}
 
 Consider the vector field density, 
 \begin{align}
 	(J^b )^{\text{Reg}}=& N^2 e^{-2 \mbo{\nu}} q^{ab}_0 p'_A \ptl_a U'^A + U'^A \mathcal{L}_{N'} ( N \bar{\mu}_{q_0} q^{ab}_0 h_{AB} \ptl_b U^B) \notag\\
 	&+ \mathcal{L}_{N'} (N) ( 2 \bar{\mu}_{q_0} q^{ab}_0 \ptl_a \mbo{\nu}')  + 2 N'^a \ptl_a \mbo{\nu}'  \bar{\mu}_{q_0} q^{ab}_0 \ptl_a N \notag\\
 	&- 2 N'^b \bar{\mu}_{q_0} q^{bc}_0 \ptl_a \mbo{\nu}' \ptl_c N.
 	\end{align}
 We shall classify the terms in $(J^b)^{\text{Reg}}$ into kinematic, dynamical and conformal terms: 
 \begin{subequations}
 \begin{align}
 	J_1^b \fdg = & N^2 e^{-2 \mbo{\nu}} q^{ab}_0 p'_A \ptl_a U'^A && \text{Dynamical terms} \\
 	J_2^b \fdg= & U'^A \mathcal{L}_{N'} ( N \bar{\mu}_{q_0} q^{ab}_0 h_{AB} \ptl_b U^B) && \text{Kinematic terms} \\
 	J_3^b \fdg= & \mathcal{L}_{N'} (N) ( 2 \bar{\mu}_{q_0} q^{ab}_0 \ptl_a \mbo{\nu}'),  \\
 	J_4^b \fdg= & 2 N'^a \ptl_a \mbo{\nu}'  \bar{\mu}_{q_0} q^{ab}_0 \ptl_a N, \\
 	J_5^b \fdg= & - 2 N'^b \bar{\mu}_{q_0} q^{bc}_0 \ptl_a \mbo{\nu}' \ptl_c N && \text{Conformal terms}
 	\end{align}
\end{subequations}

\iffalse 
Let us start with the flux of the vector field $J_1$, in the region bounded by $\Sigma_0$ and $\Sigma_t.$ We have at the horizon,

\begin{align}
	
	\end{align} 

\begin{align}
\text{Flux} (J_1, \Gamma) = \int^t_0 
\end{align}
\fi 
The asymptotic and decay rates of the fluxes and the associated integrands can be explicitly evaluated with  a specific choice of gauge on the target. In the following, we choose,  $h = 4 d \,\gamma ^2 + e^{-4 \gamma } d\,  \omega^2.$ Analogous computations can be performed for other gauges on the target manifold. An important aspect of these flux estimates is that most of the integrands of these flux terms vanish pointwise at the boundaries of the orbit space. In this work, we shall neglect signs of total flux quantities (but not for the individual constituting terms), because we shall eventually prove that they converge to $0.$

\subsection{`Dynamical' Boundary Terms} 
Consider the dynamical flux terms contributed by $$ (J_1)^b \fdg = N^2 e^{-2 \mbo{\nu}} q^{ab}_0 p'_A \ptl_a U'^A $$ 
\begin{align}
\text{Flux} (J_1, \Gamma)
\end{align}
Let us start by considering the terms $$ \text{Flux} \left(N^2 e^{-2 \mbo{\nu}} (q_0)^{ab} (4 \mbo{p} \ptl_a \gamma' \right), \Gamma)$$ and $$ \text{Flux} \left(N^2 e^{-2 \mbo{\nu}} (q_0)^{ab} ( e^{-4 \gamma} \mbo{r} \ptl_a \gamma' \right), \Gamma).$$ We have

\begin{align}
 &\text{Flux} \left(N^2 e^{-2 \mbo{\nu}} (q_0)^{ab} (4p \ptl_a \gamma' \right), \Gamma) \\
  =& \int^t_0 \int_{ (-\infty, -R_+) \, \cup\, (R_+, \infty) }  \lim_{\theta \to 0, \pi}   \Big( 4 N^2 e^{-2 \mbo{\nu}} (q_0)^{ab} \mbo{p} \ptl_a \gamma' \Big)^\theta dR dt  \notag\\
 =& \int^t_0   \int_{ (-\infty, -R_+) \, \cup\, (R_+, \infty) }    \lim_{\theta \to 0, \pi} \Big ( 4 N \left( \frac{ N\mbo{p}'}{\bar{\mu}_q} \right) \bar{\mu}_{q_0} (q_0)^{ab}  \ptl_a \gamma'   \Big)^\theta  dR dt \notag\\
 \intertext{if we consider the integrand, within the domain of integration, we have}
 & \lim_{\theta \to 0, \pi}4 \frac{N}{R^2}  \cdot \frac{N \mbo{p}'}{\bar{\mu}_q} \bar{\mu}_{q_0} \ptl_\theta \gamma' \notag\\
 =& 4 \frac{\Delta^\halb \sin \theta}{R^2} \frac{N \mbo{p}'}{\bar{\mu}_q} \bar{\mu}_{q_0}  \ptl_\theta \gamma'
 \to 0, \quad \text{as} \quad \theta \to 0, \pi . 
\end{align}
and 
\begin{align}
  & \text{Flux} \left(N^2 e^{-2 \mbo{\nu}} (q_0)^{ab} ( e^{-4 \gamma} \mbo{r} \ptl_a \omega' \right), \Gamma) \\
  =& \int^t_0  \int_{ (-\infty, -R_+) \, \cup\, (R_+, \infty) }  \lim_{\theta \to 0, \pi}  \Big(    e^{-4\gamma } N^2 e^{-2 \mbo{\nu}} (q_0)^{ab} \mbo{r} \ptl_a \omega'   \Big)^\theta  dR dt \notag\\
  =& \int^t_0  \int_{ (-\infty, -R_+) \, \cup\, (R_+, \infty) } \lim_{\theta \to 0, \pi}  \Big(     N^2 (q_0)^{ab} \frac{e^{-4\gamma }\mbo{r}}{\bar{\mu}_q} \bar{\mu}_{q_0}\ptl_a \omega'   \Big)^\theta dR dt \notag \\
  \intertext{the integrand}
  & \lim_{\theta \to 0, \pi} \frac{N^2}{R^2} e^{-4\gamma} \frac{\mbo{r}'}{\bar{\mu}_q} \bar{\mu}_{q_0} \ptl_\theta \omega' \notag\\
  =& \lim_{\theta \to 0, \pi} \frac{\Delta}{R} \sin^2 \theta\frac{\Sigma^2}{\sin^4 \theta ((r^2+a^2)^2 -a^2 \Delta \sin^2 \theta)} \frac{\mbo{r}'}{\bar{\mu}_q} \ptl_\theta \omega' \notag\\
  & \to 0, \quad \text{as} \quad  \theta \to 0, \pi
\end{align}
due to the rapid decay of $\ptl_\theta \omega'$ as $\theta \to 0$ and $\pi.$
\begin{align}
\text{Flux} (J_1, \mathcal{H}^+) 
\end{align}
We have the term, $\text{Flux} (N^2 e^{-2 \mbo{\nu}}q^{ab}_0 \mbo{p} \ptl_a \gamma', \mathcal{H}^+)$, which can be estimated near the future horizon $\mathcal{H}^+$as 

\begin{align}
&\text{Flux} (N^2 e^{-2 \mbo{\nu}}q^{ab}_0 \mbo{p} \ptl_a \gamma', \mathcal{H}^+) \notag\\
=&
\int^t_0 \int^\pi_0  \lim_{R \to R_+} \Big(
N^2 e^{-2 \mbo{\nu}}q^{ab}_0 \mbo{p}' \ptl_a \gamma' \Big)^R  d \theta dt \notag\\
=& \int^t_0 \int^\pi_0   \lim_{R \to R_+} \Big(  N \cdot  q^{ab}_0 \frac{N \mbo{p}'}{\bar{\mu}_q} \bar{\mu}_{q_0} \ptl_a \gamma'  \Big)^R  d\theta dt \notag\\
=& \int^t_0 \int^\pi_0 \lim_{R \to R_+} \Big(  R \sin \theta \left(1- \frac{R^2_+}{R^2} \right) \frac{N \mbo{p}' }{\bar{\mu}_q} \bar{\mu}_{q_0} \ptl_R \gamma' \Big)  d\theta dt .
\end{align}
It may be recalled that the behaviour of the canonical pair  $(\gamma' ,\mbo{p}') $ near the future horizon is
\begin{align}
4 \ptl_R \gamma' =&  \left(1- \frac{R^2_+}{R^2} \right) \quad \text{and} \quad \frac{N \mbo{p}'}{\bar{\mu}_q} = \mathcal{O} \left( 1- \frac{R^2_+}{R^2} \right)  \notag\\
\text{Flux} (N^2 e^{-2 \mbo{\nu}}q^{ab}_0 \mbo{p} \ptl_a \gamma', \mathcal{H}^+) & \to 0 \quad \text{as} \quad R \to R_+.
\end{align}
Likewise, 
\begin{align}
&\text{Flux} (N^2 e^{-2 \mbo{\nu}}q^{ab}_0 \mbo{r}' \ptl_a \omega', \mathcal{H}^+) \notag\\
=&
\int^t_0 \int^\pi_0 \lim_{R \to R_+} \Big(  (N^2 e^{-2 \mbo{\nu}}q^{ab}_0 \mbo{r}' \ptl_a \omega' \Big)^R  d\theta d t \notag\\
=& \int^t_0 \int^\pi_0  \lim_{R \to R_+}   \Big( (N e^{-4\gamma} \cdot  q^{ab}_0 \frac{N e^{4\gamma} \mbo{r}'}{\bar{\mu}_q} \bar{\mu}_{q_0} \ptl_a \omega' \Big)^R  d\theta dt  \notag\\
%=& \int^t_0 \int^\pi_0   \lim_{R \to R_+} \Big( R \sin \theta \left(1- \frac{R^2_+}{R^2} \right) \frac{N \mbo{p}' }{\bar{\mu}_q} \bar{\mu}_{q_0} \ptl_R \omega' \Big)  d\theta d t 
\end{align}
again recall that the behaviour of the canonical pair $ (\omega', \mbo{r}' )$ near the future horizon is
\begin{align}
\ptl_R \omega' =& \left(1- \frac{R^2_+}{R^2} \right) \quad \text{and} \quad \frac{N e^{4 \gamma} \mbo{r}'}{\bar{\mu}_q} = \mathcal{O}\left( 1- \frac{R^2_+}{R^2} \right) \notag\\
\intertext{the integrand in
$\text{Flux} (N^2 e^{-2 \mbo{\nu}}q^{ab}_0 \mbo{p} \ptl_a \gamma', \mathcal{H}^+)$ is  }
=& \lim_{R \to R_+} R^2 \sin \theta \left(1- \frac{R^2_+}{R^2} \right) \frac{\Sigma^2}{ \sin^4 \theta \left((r^2+a^2)^2-a^2 \Delta \sin^2 \theta \right)^2} \notag\\
& \quad \quad  c \left(1-\frac{R^2_+}{R^2} \right) \cdot \left(1-\frac{R^2_+}{R^2} \right)   \notag\\
& \to 0 \quad \text{as} \quad R \to R_+.
\end{align}

\subsubsection*{Flux of $J_1$ at the outer boundary, $\text{Flux} (J_1, \bar{\iota}^0)$}

\begin{align}
\text{Flux} (J_1, \bar{\iota}^0) 
\end{align}
 Consider the flux of the term: 
\begin{align}
\text{Flux} \left(\frac{N^2}{ \bar{\mu}_q} \mbo{p}' \bar{\mu}_q q^{ab} \ptl_a \gamma', \bar{\iota}^0 \right) = \int^t_0 \int^\pi_0 \Big( \lim_{R \to \infty} \frac{N^2}{ \bar{\mu}_q} \mbo{p}' \bar{\mu}_{q_0}  \ptl_R \gamma'  \Big)  d\theta dt 
\end{align}
we have 
\begin{align}
\frac{N p}{\bar{\mu}_q}' = \mathcal{O} \left(\frac{1}{R^3} \right), \quad \gamma' = \frac{1}{R} \quad \ptl_R \gamma' = \mathcal{O} \left( \frac{1}{R^2} \right) 
\end{align}
and recall that 
\begin{align}
N = (r(R)^2 - 2Mr(R) + a^2)^{1/2} \sin^2 \theta = \mathcal{O}(R), \quad \text{for large} \quad R . 
\end{align}
Therefore, in the region under consideration, we have for the integrand of $\text{Flux} \left(\frac{N^2}{ \bar{\mu}_q} p' \bar{\mu}_q q^{ab} \ptl_a \gamma', \bar{\iota}^0 \right)   \sim \mathcal{O}(R) \cdot \mathcal{O} (\frac{1}{R^3}) \cdot \mathcal{O} (\frac{1}{R^2}) \cdot R \to 0 \quad \text{as} \quad R \to \infty.$
Likewise,  consider the following flux terms at the outer boundary: 
$$\text{Flux} \left( -\frac{4 N^2}{\bar{\mu}_q} \mbo{r}' \gamma' \bar{\mu}_{q_0} q_0^{ab} \ptl_a \omega, \bar{\iota}^0  \right), \quad \text{Flux} \left( \frac{N^2}{\bar{\mu}_q} \mbo{r}' \bar{\mu}_{q_0} q_0^{ab} \ptl_a \omega', \bar{\iota}^0  \right) $$

\begin{align}
\text{Flux} \left( \frac{N^2}{\bar{\mu}_q} \mbo{r}' \bar{\mu}_{q_0} q^{ab} \ptl_a \omega', \bar{\iota}^0  \right) = \int^t_0 \int^\pi_0 \Big(  \lim_{R \to \infty} \frac{N^2}{\bar{\mu}_q} \mbo{r}' \bar{\mu}_{q_0} \ptl_R \omega' \Big)  d \theta dt 
\end{align}
It may be recalled that 
\begin{align}
\frac{N e^{4 \gamma} \mbo{r}'}{\bar{\mu}_q} =& \mathcal{O} \left(\frac{1}{R^3} \right),  
\intertext{and}
 e^{-2\gamma} =& \left( \frac{r^2 + a^2 \cos^2 \theta}{ \sin^2 \theta} \right) \frac{1}{ ((r^2 +a^2)^2 -a^2 \Delta \sin^2 \theta)} \sim \mathcal{O} (\frac{1}{R^2}) \quad \text{for large} \quad R 
\end{align}
Therefore, for the integrand within

$\text{Flux} \left( \frac{N^2}{\bar{\mu}_q} \mbo{r}' \bar{\mu}_{q_0} q^{ab} \ptl_a \omega', \bar{\iota}^0  \right) = \mathcal{O}(R) \cdot \mathcal{O}\left(\frac{1}{R^3} \right) \cdot \mathcal{O}\left(\frac{1}{R^4} \right) R^2 \cdot \frac{1}{R^2} \to 0 \,\text{as} \, R \to \infty.$

Next, consider 
\begin{align}
&\text{Flux} \left( -\frac{4 N^2}{\bar{\mu}_q} \mbo{r}' \gamma' \bar{\mu}_{q_0} q_0^{ab} \ptl_a \omega, \bar{\iota}^0  \right) \notag\\
=& \int^t_0 \int^\pi_0 \Big( \lim_{R \to \infty}  -\frac{4 N^2}{\bar{\mu}_q} \mbo{r}' \gamma' \bar{\mu}_{q_0}  \ptl_R \omega \Big)  d\theta dt    \notag\\
=& \int^t_0 \int^\pi_0 \Big( \lim_{R \to \infty}  - 4N \left( \frac{ N e^{4 \gamma}\mbo{r'}}{\bar{\mu}_q} \right) e^{-4\gamma} \gamma' \bar{\mu}_{q_0}  \ptl_R \omega \Big)  d\theta dt 
\intertext{ Using the above estimates again and noting that $\ptl_R \omega = \mathcal{O}(\frac{1}{R^3})$}
\frac{4 N^2}{\bar{\mu}_q} \mbo{r}' \gamma' \bar{\mu}_{q_0}  \ptl_R \omega =& \mathcal{O}(R) \cdot \mathcal{O}\left(\frac{1}{R^3} \right) \cdot  \mathcal{O}\left(\frac{1}{R^4} \right) \cdot 
\mathcal{O}\left(\frac{1}{R} \right)  R \mathcal{O}\left(\frac{1}{R^3} \right)  \notag\\
=& 0 \quad \text{as} \quad R \to \infty,
\end{align}
from which it follows that $\text{Flux} \left( -\frac{4 N^2}{\bar{\mu}_q} \mbo{r}' \gamma' \bar{\mu}_{q_0} q_0^{ab} \ptl_a \omega, \bar{\iota}^0  \right) =0. $
\subsection{`Kinematic' Boundary  Terms}
\begin{align}
(J_2)^b \fdg =& U'^A \mathcal{L}_{N'} (N \bar{\mu}_q q^{ab} h_{AB} (U) \ptl_b U^B)
\intertext{in a special gauge}
=& \gamma' \mathcal{L}_{N'} (4N \bar{\mu}_q q^{ab} \ptl_a \gamma) + \omega' \mathcal{L}_{N'} (N \bar{\mu}_q q^{ab} e^{-4 \gamma} \ptl_a \omega). 
\end{align}
\begin{align}
\text{Flux} (J_2, \Gamma)
\end{align}
We have, 
\begin{align}
&\text{Flux}(J_2, \Gamma) \notag \\
=&  \int^t_0 \int_{ (-\infty, R_+) \cup (R_+, \infty) }  \lim_{\theta \to 0, \pi} \left(  U'^A \mathcal{L}_{N'} (N \bar{\mu}_{q_0} q^{ab}_0 h_{AB}(U) \ptl_b U^B ) \right)^\theta dR dt \notag\\
%=&  \int^t_0 \int_{ (-\infty, R_+) \cup (R_+, \infty) } \left( \lim_{\theta \to 0, \pi} \frac{1}{R} U'^A \mathcal{L}_{N'} (\frac{1}{R^2}\bar{\mu}_{q_0} h_{AB}(U) \ptl_\theta U^B ) \right) dR dt 
\end{align}
We will expand the expression $U'^A \mathcal{L}_{N'} (N \bar{\mu}_{q_0} q_0^{ab} h_{AB}(U) \ptl_a U^B )$ as follows: 

\begin{align} \label{Kinematic-expansion}
&U'^A \mathcal{L}_{N'} \left( N q^{ab}_0\bar{\mu}_{q_0} h_{AB}(U) \ptl_\theta U^B \right) \notag\\
=& U'^A ( N \ptl_c N'^c \bar{\mu}_{q_0} q^{ab}_0 h_{AB}(U) \ptl_a U^B + N'^c \ptl_c (N \bar{\mu}_{q_0} q_0^{ab} h_{AB}(U) \ptl_a U^B) \notag\\
& \quad - \ptl_c N'^b N \bar{\mu}_{q_0} q^{ab}_0 h_{AB} \ptl_a U^B)
\end{align}
We would like to point out that $\omega'$ vanishes on the axes whereas $\gamma'$ does not. Likewise, $\ptl_\theta \omega$ decays rapidly on the axes whereas $\ptl_\theta \gamma$ does not. This causes a few subtleties for terms involving $\gamma'$ but as will see later, we will have few fortuitous cancellations involving this quantity.  
Each of the terms in \eqref{Kinematic-expansion} can in turn be decomposed as 
\begin{align}
U'^A ( N \ptl_c N'^c \bar{\mu}_{q_0} q^{ab}_0 h_{AB}(U) \ptl_a U^B) =& \gamma' (4 N \ptl_c N'^c \bar{\mu}_{q_0} q^{ab} _0  \ptl_a \gamma ) \notag\\
&+ \omega' (N  \ptl_c N'^c \bar{\mu}_{q_0} q^{ab} _0 e^{-4 \gamma} \ptl_a \omega) 
\end{align}
As a consequence, we have 
\begin{align} \label{kin-flux-axes-1}
&\text{Flux} ( \gamma' (4 N \ptl_c N'^c \bar{\mu}_{q_0} q^{ab} _0  \ptl_a \gamma ), \Gamma) \notag\\
 =& \int^t_0 \int_{ (-\infty, -R_+) \cup (R_+, \infty) }  \lim_{\theta \to 0, \pi}   \left(   \gamma' (4 N \ptl_c N'^c \bar{\mu}_{q_0} q^{ab} _0  \ptl_a \gamma )  \right)^\theta  dR dt \notag\\
=& \int^t_0 \int_{ (-\infty, R_+) \cup (R_+, \infty) } \left(  \lim_{\theta \to 0, \pi}  4 \frac{N}{R} \gamma' \ptl_\theta \gamma (\ptl_\theta N'^\theta + \ptl_R N'^R) \right) dR dt.
\end{align}
Let us note that $\displaystyle \lim_{\theta \to 0, \pi} \ptl_\theta \gamma' =  \cot \theta.$ Likewise, 
\begin{align}
& \text{Flux} ( \omega' ( \ptl_c N'^c N \bar{\mu}_{q_0} q^{ab} _0 e^{-4 \gamma} \ptl_a \omega), \Gamma) \notag\\
=& \int^t_0 \int_{ (-\infty, R_+) \cup (R_+, \infty) }   \lim_{\theta \to 0, \pi} \left( \omega' ( \ptl_c N'^c N \bar{\mu}_{q_0} q^{ab} _0  e^{-4\gamma} \ptl_a \omega \ptl_b ) \right)^\theta dR dt \notag\\
=&  \int^t_0 \int_{ (-\infty, R_+) \cup (R_+, \infty) } \left( \lim_{\theta \to 0, \pi}  \frac{N}{R} e^{-4\gamma} \omega' \ptl_\theta \omega \ptl_c ( N'^c) \right) dR dt 
=0.
\end{align}
Where, we took into account the fact that the $\ptl_\theta \omega$ term has rapid decay as $\theta \to 0, \pi.$
Next, 
\begin{align}
U'^A N'^c \ptl_c (N \bar{\mu}_{q_0} q^{ab}_0 h_{AB} (U) \ptl_a U^A) =& 4\gamma' N'^c\ptl_c ( N \bar{\mu}_{q_0} q^{ab}_0 \ptl_a \gamma) \notag\\
&+ \omega' N'^c \ptl_c (N \bar{\mu}_{q_0} e^{-4\gamma} \ptl_a \omega)
\end{align}

\begin{align}
&\text{Flux} (4\gamma' N'^c\ptl_c (N \bar{\mu}_{q_0} q^{ab}_0 \ptl_a \gamma) , \Gamma)  \notag\\
=& \int^t_0 \int_{ (-\infty, R_+) \cup (R_+, \infty) }  \lim_{\theta \to 0, \pi} \left( 
    4\gamma' N'^c\ptl_c ( N \bar{\mu}_{q_0} q^{ab}_0 \ptl_a \gamma) , \frac{1}{R} \ptl_\theta  \right)^\theta dR dt \\
=& \int^t_0 \int_{ (-\infty, R_+) \cup (R_+, \infty) } \left( \lim_{\theta \to 0, \pi}  4  \gamma' N'^c \ptl_c (N\frac{\bar{\mu}_{q_0}}{R^2} \ptl_\theta \gamma) \right) dR dt 
\end{align}
Consider the integrand: 
\begin{align} \label{N'^R terms}
 R\gamma' N'^c \ptl_c (\frac{N}{R} \ptl_\theta \gamma) = 4  \gamma' \left( N'^ \theta \ptl_\theta  (\frac{N}{R} \ptl_\theta \gamma) + N'^R \ptl_R ( \frac{N}{R} \ptl_\theta \gamma) \right) 
\end{align}
The $N'^R$ term remains while the $N'^\theta$ term vanishes at the axes (as $\theta \to 0$ or $\pi$). For the sake of brevity, let us combine the $N'^R$ terms in \eqref{N'^R terms} and \eqref{kin-flux-axes-1}. We shall revisit this later. 

\begin{align} \label{kin-flux-comb-nu'}
\int^t_0 \int_{ (-\infty,-R_+) \cup (\infty, R_+ )}  \left( 4 \gamma'  \ptl_R (N'^R \frac{N}{R} \ptl_\theta \gamma ) \right) dR dt.
\end{align}
Similarly, 
\begin{align}
& \text{Flux} ( \omega' N'^c \ptl_c ( N \bar{\mu}_{q_0} q^{ab}_0 e^{-4\gamma} \ptl_a \omega), \Gamma)  \notag\\ 
=&
\int^t_0 \int_{ (-\infty, R_+) \cup (R_+, \infty) }  \lim_{\theta \to 0, \pi} \left(      \omega' N'^c \ptl_c ( N \bar{\mu}_{q_0} q^{ab}_0 e^{-4\gamma} \ptl_a \omega) \ptl_b  \right)^\theta dR dt  \\
=& \int^t_0 \int_{ (-\infty, R_+) \cup (R_+, \infty) } \left(  \lim_{\theta \to 0, \pi}   \omega' N'^c \ptl_c (\frac{N \bar{\mu}_{q_0}}{R^2} e^{-4\gamma} \ptl_\theta \omega) \right) dR dt 
\end{align}
The integrand
\begin{align}
 \omega' N'^c \ptl_c (\frac{\bar{\mu}_{q_0}}{R^2} e^{-4\gamma} \ptl_\theta \omega)= \omega' (N'^\theta  \ptl_\theta  (N \frac{\bar{\mu}_{q_0}}{R^2} e^{-4\gamma} \ptl_\theta \omega) + N'^ R \ptl_R (N \frac{\bar{\mu}_{q_0}}{R^2} e^{-4\gamma} \ptl_\theta \omega ) ) 
\end{align}
vanishes at the axes due to $\omega' =0$ on $\Gamma$ and the rapid decay of $\ptl_\theta \omega$ at the axes. 

Finally, 
\begin{align}
-U'^A \ptl_c N'^b \bar{\mu}_{q_0} q^{ac}_0 h_{AB} \ptl_a U^B =&  -4\gamma' \ptl_c N'^b \bar{\mu}_{q_0} \bar{q_0}^{ac} \ptl_a \gamma  \notag\\
&-\omega' \ptl_c N'^b \bar{\mu}_{q_0} \bar{q_0}^{ac} e^{-4\gamma} \ptl_a \omega   
\end{align}
so that the fluxes, 

\begin{align}
& \text{Flux} ( -4\gamma' \ptl_c N'^b \bar{q_0}^{ac} \ptl_a \gamma, \Gamma)  \notag\\
=& \int^t_0 \int_{ (-\infty, R_+) \cup (R_+, \infty) } \lim_{\theta \to 0, \pi}  \left(    -4N \gamma' \ptl_c N'^b \bar{\mu}_{q_0} \bar{q_0}^{ac} \ptl_a \gamma  \right)^\theta dR dt \notag\\
=& \int^t_0 \int_{ (-\infty, R_+) \cup (R_+, \infty) } \left( \lim_{\theta \to 0, \pi}   ( -4 \gamma' N (R \ptl_R N'^\theta \ptl_R \gamma + \frac{1}{R} \ptl_\theta N'^\theta \ptl_\theta \gamma ) ) \right) dR dt 
\end{align}
If we consider the integrand 
\begin{align}
\lim_{\theta \to 0, \pi}  ( 4 \gamma' N  (N\ptl_R N'^\theta \ptl_R \gamma + \frac{1}{R} \ptl_\theta N'^\theta \ptl_\theta \gamma) ) 
\end{align}
we note that the term involving $ \ptl_R N'^R \ptl_R \gamma $vanishes because $\ptl_R N'^\theta$ vanishes on the axes ($N'^\theta $ is a constant along the axes)
and in the second term involving $\ptl_\theta N'^\theta \ptl_\theta \gamma$ cancels with a flux term in \eqref{kin-flux-axes-1}.
\begin{align}
& \text{Flux} ( -\omega' \ptl_c N'b \bar{q_0}^{ac} e^{-4\gamma} \ptl_a \omega, \Gamma ) \notag\\
=& \int^t_0 \int_{ (-\infty, R_+) \cup (R_+, \infty) }   \lim_{\theta \to 0, \pi}  \left(  
-N \omega' \ptl_c N'^b \bar{q_0}^{ac} e^{-4\gamma} \ptl_a \omega  \right)^\theta  dR dt  \notag\\
=& \int^t_0 \int_{ (-\infty, R_+) \cup (R_+, \infty) } \left( \lim_{\theta \to 0, \pi}   ( - N \omega' ( \ptl_R  N'^ \theta e^{-4\gamma} \ptl_R \omega + \ptl_\theta N'^ R e^{-4 \gamma} \ptl_\theta \omega) \right) dR dt 
\notag\\
=&0
\end{align}
on account of rapid decay of $\ptl_\theta \omega$ at the axes, vanishing of $N \sim \sin \theta $ and $\omega'$ at the axes.
\iffalse
\begin{align}
N\omega' \ptl_c N'b \bar{q_0}^{ac} e^{-4\gamma} \ptl_a \omega (\frac{1}{R} \ptl_\theta)^d
\end{align}
\fi 
Now then, for
\begin{align}
\text{Flux} (J_2, \mathcal{H}^+),
\end{align}
we have 
\begin{align}
\text{Flux} (\gamma' (4 \ptl_c N'^c \bar{\mu}_{q_0} q^{ab} _0  \ptl_a \gamma, \mathcal{H}^+) 
&= \int^t_0 \int^\pi_0  (\lim_{R \to R_+}  4 R N \gamma' \ptl_R \gamma \ptl_c N'^c)  d \theta dt 
\end{align}

Using, 

\begin{align}
\lim_{R \to R_+} \ptl_R \gamma  \leq c, \quad \lim_{R \to R_+} \ptl_c N'^c \leq c 
\end{align}

 we estimate the integrand, 

\begin{align}
\lim_{R \to R_+} 4 R N \gamma' \ptl_R \gamma \ptl_c N'^c \leq c ( 1- \frac{R^2_+}{R^2})^2 
\end{align}

So we have flux, 

\begin{align}
\text{Flux} (\gamma' (4 \ptl_c N'^c \bar{\mu}_{q_0} q^{ab} _0  \ptl_a \gamma, \mathcal{H}^+) =0.
\end{align}

\begin{align}
& \text{Flux} (\omega' ( \ptl_c N'^c \bar{\mu}_{q_0} q^{ab} _0 e^{-4 \gamma} \ptl_a \omega) , \mathcal{H}^+)  \notag \\
= &  \int^t_0 \int^\pi_0 ( \lim_{R \to R_+}  \omega' \ptl_c N'^c N \bar{\mu}_{q_0} e^{-4\gamma} \ptl_R \omega )  d \theta dt   \notag\\
=& \int^t_0 \int^\pi_0 (\lim_{R \to R_+}  R N e^{-4\gamma} \omega' \ptl_R \omega \ptl_c N'^c)R d\theta dt 
\end{align}
Noting that, 

\begin{align}
\lim_{R \to R_+} \ptl_R \omega \leq c, \quad \lim_{R \to R_+} e^{-4\gamma} \leq c 
\end{align}
the integrand 

\begin{align}
R N e^{-4\gamma} \omega' \ptl_R \omega \ptl_c N'^c \leq c (1- \frac{R_+}{R^2})^2 
\end{align}
thus, again 

\begin{align}
\text{Flux} (\omega' ( \ptl_c N'^c \bar{\mu}_{q_0} q^{ab} _0 e^{-4 \gamma} \ptl_a \omega) =0.
\end{align}

\begin{align}
& \text{Flux} (4\gamma' N'^c\ptl_c (N \bar{\mu}_{q_0} q^{ab}_0 \ptl_a \gamma), \mathcal{H}^+) \notag \\ 
=& \int^t_0 \int^\pi_0  \left(\lim_{R \to R_+}  4 \gamma' N'^c \left(\ptl_c N R \ptl_R \gamma + N \ptl_c (R \ptl_R \gamma) \right) \right) dR dt 
\end{align}
the integrand can be estimated close to the horizon as, 

\begin{align}
4 \gamma' N'^c \left(\ptl_c N R \ptl_R \gamma + N \ptl_c (R \ptl_R \gamma) \right) \leq c (1- \frac{R^2_+}{R^2}) + c (1- \frac{R^2_+}{R^2})^2
\end{align}
so we have, 

\begin{align}
\text{Flux} (4\gamma' N'^c\ptl_c (N \bar{\mu}_{q_0} q^{ab}_0 \ptl_a \gamma), \mathcal{H}^+) =0. 
\end{align}
\begin{align}
&\text{Flux} (\omega' N'^c \ptl_c ( \bar{\mu}_{q_0} e^{-4\gamma} \ptl_a \omega), \mathcal{H}^+)  \notag\\
=& 
\int^t_0 \int^\pi_0 \left( \lim_{R \to R_+}  \omega' N'^c \ptl_c( N R e^{-4\gamma} \ptl_R \omega) \right)  d\theta dt  \notag\\
=&  \int^t_0 \int^\pi_0  \left( \lim_{R \to R_+}  \omega' N'^c ( \ptl_c N R e^{-4\gamma} \ptl_R \omega + N \ptl_c (e^{-4\gamma} R \ptl_R \omega)) \right)  d \theta dt 
\end{align}
the integrand can be estimated as 

\begin{align}
\omega' N'^c ( \ptl_c N R e^{-4\gamma} \ptl_R \omega + N \ptl_c (e^{-4\gamma} R \ptl_R \omega)) \leq c (1- \frac{R^2_+}{R^2}) + c ( 1- \frac{R^2_+}{R^2})^2
\end{align}
so we have 
\begin{align}
\text{Flux} (\omega' N'^c \ptl_c ( \bar{\mu}_{q_0} e^{-4\gamma} \ptl_a \omega), \mathcal{H}^+) =0.
\end{align}
\begin{align}
& \text{Flux} (-4 N \gamma' \ptl_c N'^b \bar{q_0}^{ac} \ptl_a \gamma, \mathcal{H}^+)  \notag\\
&=  \int^t_0 \int^\pi_0  \left( \lim_{R \to R_+} -4 N \gamma' (\ptl_R N'^R \ptl_R \gamma + \frac{1}{R^2} \ptl_\theta N'^ \theta \ptl_ \theta \gamma) \right)  d\theta dt \\
& =0
\end{align}
in view of the behaviour of the integrand 
\begin{align}
 4 N \gamma' (\ptl_R N'^R \ptl_R \gamma + \frac{1}{R^2} \ptl_\theta N'^ \theta \ptl_ \theta \gamma ) \quad \text{behaves as }  \quad  c(1- \frac{R^2_+}{R^2})^2 \quad \text{as} \quad R \to R_+
\end{align}
finally, 
\begin{align}
& \text{Flux} (-\omega' \ptl_c N'^b \bar{q_0}^{ac} N e^{-4\gamma} \ptl_a \omega, \mathcal{H}^+) \notag \\
 =&  \int^t_0 \int^\pi_0  \left( \lim_{R \to R_+}  - N e^{-4 \gamma} \omega' ( \ptl_R N'^R \ptl_R \omega + \frac{1}{R^2} \ptl_\theta N'^\theta \ptl_\theta \omega) \right)  d\theta dt  \notag\\
=& 0
\end{align}
where, again the integrand can be estimated such that 
\begin{align}
N e^{-4 \gamma} \omega' ( \ptl_R N'^R \ptl_R \omega + \frac{1}{R^2} \ptl_\theta N'^\theta \ptl_\theta \omega)
\intertext{behaves as }  
 c(1- \frac{R^2_+}{R^2})^2 \quad \text{for $R$ close to $R_+$ }.
\end{align}
Next, for the kinematic fluxes at the spatial infinity
\begin{align}
\text{Flux}(J_2, \iota^0)
\end{align}
let us start with estimating the divergence term $\displaystyle \ptl_c N'^c$ near the spatial infinity. We have
\begin{align}
& \ptl_c N'^c = \ptl_\theta N'^ \theta + \ptl_R N'^R \notag\\
&=  \ptl_R \Big( \ptl_t Y^R - N^2 e^{-2 \nu} \ptl_R (Y^t(t=0) - \mathcal{Y}^ \theta \cot \theta - \mathcal{Y}^R \frac{ 1 + \frac{R^2_+}{R^2}}{ R ( 1- \frac{R^2_+}{R^2})}  ) \Big) \notag \\
& + \ptl_\theta \Big( \ptl_t Y^\theta - \frac{N^2 e^{-2\mbo{\nu}'}}{R^2} (Y^t(t=0) - \mathcal{Y}^\theta \cot \theta -   \mathcal{Y}^R \frac{ 1 + \frac{R^2_+}{R^2}}{ R ( 1- \frac{R^2_+}{R^2})} )\Big)
\end{align}
We have the following behaviour for large $R$
\begin{align}
& \ptl_R (N^2 e^{-2 \mbo{\nu}}) \sim \mathcal{O}(\frac{1}{R}), 
\quad \ptl_\theta  \frac{N^2 e^{-2 \mbo{\nu}'}}{R^2}  \sim \mathcal{O} (1)
\intertext{and}
&  \ptl_R  \frac{ 1 + \frac{R^2_+}{R^2}}{ R ( 1- \frac{R^2_+}{R^2})} \sim \mathcal{O} (\frac{1}{R^2}), \quad \ptl^2_R  \frac{ 1 + \frac{R^2_+}{R^2}}{ R ( 1- \frac{R^2_+}{R^2})} \sim \mathcal{O} (\frac{1}{R^3})
\end{align} 
from somewhat lengthy but standard computations. It now follows that the coordinate divergence of $N'$ is then 
\begin{align}
\ptl_c N'^c \sim \mathcal{O} (\frac{1}{R})
\end{align}

\begin{align}
\text{Flux} (\gamma' (4 \ptl_c N'^c \bar{\mu}_{q_0} q^{ab} _0  \ptl_a \gamma, \iota^0) = 
\int^t_0 \int^\pi_0  (\lim_{R \to \infty} 4 R \gamma' \ptl_R \gamma \ptl_c N'^c \bar{\mu}_{q_0} )  d \theta dt \notag\\
\end{align}

\begin{align}
4 R \gamma' \ptl_R \gamma \ptl_c N'^c \bar{\mu}_{q_0} =& 4 R \cdot \mathcal{O} (\frac{1}{R}) \cdot \mathcal{O} (\frac{1}{R}) \cdot \mathcal{O} (\frac{1}{R})
\intertext{for large $R$}
%&= \lim_{R \to \infty} 4R \gamma' \ptl_R \gamma \ptl_c ( \ptl_t  Y^c - N^2 e^{-2 \mbo{\nu}} q^{ac}_0 \ptl_a Y^t)
 \to &\, 0 \quad \text{as} \quad R \to \infty.
\end{align}
Thus, 
\begin{align}
\text{Flux} (\gamma' (4 \ptl_c N'^c \bar{\mu}_{q_0} q^{ab} _0  \ptl_a \gamma, \iota^0) =0.
\end{align}

\begin{align}
& \text{Flux} (\omega' ( \ptl_c N'^c \bar{\mu}_{q_0} q^{ab} _0 e^{-4 \gamma} \ptl_a \omega) , \iota^0)  \notag\\
& = \int^t \int^\pi_0 (\lim_{R \to \infty}  R e^{-4\gamma} \omega' \ptl_R \omega \ptl_c( \ptl_t Y^c - N^2 e^{-2 \mbo{\nu}} q^{ac}_0 \ptl_a Y^t) )  d \theta dt 
\end{align}
\begin{align}
R e^{-4\gamma} \omega' \ptl_R \omega \ptl_c( \ptl_t Y^c - N^2 e^{-2 \mbo{\nu}} q^{ac}_0 \ptl_a Y^t) = & R \cdot \mathcal{O} (\frac{1}{R^4}) \cdot \mathcal{O} (\frac{1}{R}) \cdot \mathcal{O} (\frac{1}{R^3}) \cdot \mathcal{O} (\frac{1}{R}) \notag\\
 \to & \, 0 \quad \text{as} \quad R \to \infty
\end{align}
\begin{align}
\text{Flux} (\omega' ( \ptl_c N'^c \bar{\mu}_{q_0} q^{ab} _0 e^{-4 \gamma} \ptl_a \omega) , \iota^0) =0.
\end{align}
Analogously, noting that $\gamma' = \mathcal{O}(\frac{1}{R})$ and $\omega' = \mathcal{O} (\frac{1}{R})$ for large $R$ we have the following: 
\begin{align}
& \text{Flux} (4\gamma' N'^c\ptl_c ( \bar{\mu}_{q_0} q^{ab}_0 \ptl_a \gamma), \iota^0) \notag\\
&= \int^t_0 \int^\pi_0  (\lim_{R \to \infty}  4 \gamma' N'^c \ptl_c (R \ptl_R \gamma) )  d\theta dt \notag\\
& = \int^t_0 \int^\pi_0 ( 4 \gamma' ( \ptl_t Y^c - N^2 e^{-2 \mbo{\nu}} q^{ac} \ptl_a Y^t) \ptl_c (R \ptl_R \gamma)  )  d\theta dt \notag\\
&=0.
\end{align}
\begin{align}
&\text{Flux} (\omega' N'^c \ptl_c ( \bar{\mu}_{q_0} e^{-4\gamma} \ptl_a \omega), \iota^0) \notag\\
& =  \int^t_0 \int^\pi_0 ( \lim_{R \to \infty} \omega' ( \ptl_t Y^c - N^2 e^{-2 \mbo{\nu}} q^{ac} \ptl_a Y^t) \ptl_c (R e^{-4\gamma} \ptl_R \omega))   d\theta dt  \notag\\
&=0 
\end{align}
\begin{align}
&\text{Flux} (-4\gamma' \ptl_c N'^b \bar{q_0}^{ac} \ptl_a \gamma, \iota^0)  \notag\\
&=  \int^t_0 \int^\pi_0 ( \lim_{R \to \infty} -4 \gamma' \left(\ptl_R N'^R \ptl_R \gamma  + \frac{1}{R^2} \ptl_\theta N'^R \ptl_\theta \gamma \right) )  d\theta dt 
\notag\\
& = 0
\end{align}
\begin{align}
&\text{Flux} (-\omega' \ptl_c N'^b \bar{q_0}^{ac} e^{-4\gamma} \ptl_a \omega, \iota^0)  \notag\\
&=  \int^t_0 \int^\pi_0 ( \lim_{R \to \infty} \omega' e^{-4\gamma} \left(\ptl_R N'^R \ptl_R \omega  + \frac{1}{R^2} \ptl_\theta N'^R \ptl_\theta \omega \right) )  d\theta dt 
&=0.
\end{align}

\iffalse
It may be noted that $\mathcal{L}_{N'} (4 N \bar{\mu}_{q_0}\ptl_R \gamma)$ and $\mathcal{L}_{N'} (4N \bar{\mu}_{q_0} e^{-4\gamma} \ptl_R \omega)$ are bounded as $R \to \infty$ 
\begin{align}
\text{Flux} (J_2, \bar{\iota}^0) =& \gamma' \mathcal{L}_{N'} (4N R \ptl_R \gamma) + \omega' \mathcal{L}_{N'} (N R e^{-4 \gamma} \ptl_R \omega) \\
&\leq \mathcal{O}\left(\frac{1}{R} \right) \cdot c \\
&\to 0 \quad \text{as} \quad R \to \infty 
\end{align}
\fi 
\subsection{`Conformal' Boundary Terms}
\begin{align}
(J_3)^b \fdg = \mathcal{L}_{N'} (N) (2 \bar{\mu}_q q^{ab} \ptl_a \mbo{\nu}')
\end{align}
Let us start with the flux of this conformal term at the axes: 
\begin{align}
\text{Flux} (J_3, \Gamma) 
\end{align}

\begin{align}
&\text{Flux} (J_3, \Gamma) \\
&= \int^t_0 \int_{(-\infty, -R_+) \cup ( R_+, \infty)}   \lim_{\theta \to 0, \pi}  \left( \mathcal{L}_{N'} (N) (2 \bar{\mu}_{q_0}q^{ab}_0 \ptl_a \mbo{\nu})  \right)^\theta dR dt \notag\\
%=& \lim_{\theta \to 0, \pi} 2 \mathcal{L}_{N'} \ptl_\theta \mbo{\nu}' 
& = \int^t_0 \int_{(-\infty, -R_+) \cup ( R_+, \infty)} \left(  2 \ptl_\theta \mbo{\nu}' ( N'^R \ptl_R N + N'^\theta \ptl_\theta N) \right) dR dt \notag\\
=& \int^t_0 \int_{(-\infty, -R_+) \cup ( R_+, \infty)} \left( \lim_{\theta \to 0, \pi} 2 \ptl_\theta \mbo{\nu}' \left( N'^R \frac{\sin \theta R^2_+}{R^2} + N'^\theta R \cos \theta(1-\frac{R^2_+}{R^2} \right)  \right) dR dt
\end{align}
expanding out the integrand within $\text{Flux} (J_3, \Gamma) $ , near the axes $\Gamma$, we get
\begin{align}
& \lim_{\theta \to 0, \pi} 2 \ptl_\theta \mbo{\nu}' \Bigg \{  \Big( \ptl_t Y^R -\left( 1 - \frac{R^2_+}{R^2} \right)^2 \cdot \frac{R^4}{ ( (r^2+a^2)^2 -a^2 \Delta \sin^2 \theta)} \notag\\
& \quad \ptl_R \Big(Y^t(t=0) - \mathcal{Y}^R \frac{1+ \frac{R^2_+}{R^2}}{1-\frac{R^2_+}{R^2}} - \mathcal{Y}^\theta \cot \theta  \Big)   \Big )\frac{\sin \theta R^2_+}{R^2} \notag\\
& + \Big( \ptl_t Y^\theta -    (1- \frac{R^2_+}{R^2})^2 \cdot \frac{R^4}{ \left( (r^2+a^2)^2 -a^2 \Delta
	\sin^2 \theta  \right)}   \notag\\
& \quad \ptl_\theta \Big( Y^t (t=0) - \mathcal{Y}^R \frac{1 + \frac{R^2_+}{R^2}} {{1-\frac{R^2_+}{R^2}} }  - \mathcal{Y}^ \theta \cot \theta \Big)  \Big)  R \cos \theta (1-\frac{R^2_+}{R^2}) \Bigg\}
\intertext{from which it  follows that}
\text{Flux} (J_3, \Gamma)=&  0
\end{align}
due to the decay rate of the terms in the parenthesis.

%=& 2 \ptl_\theta \mbo{\nu}' \left( (\ptl_t Y^R - N^2 e^{-2 \mbo{\nu}} \ptl_R Y^t) \sin \theta \frac{R^2_+}{R^2} + (\ptl_t Y^\theta -  \frac{N^2}{R^2} e^{-2 \mbo{\nu}} \ptl_\theta Y^t) R \cos \theta (1-\frac{R^2_+}{R^2})  \right)  

Now then, consider the flux of $J_3$ at the horizon, $H^+:$
\begin{align}
& \text{Flux} (J_3, \mathcal{H}^+) =  \int^t_0 \int^\pi _0  \left( \lim_{R \to R_+}  \mathcal{L}_{N'} (N) (2 \bar{\mu}_q q^{ab} \ptl_a \mbo{\nu}') \right)  d\theta dt \notag\\
& \int^t_0 \int^\pi_0  \left( \lim_{R \to R_+} 2 R  \mathcal{L}_{N'} N \ptl_R \mbo{\nu}' \right)  d\theta dt  \notag\\
&=  \int^t_0 \int^\pi _0 \Big( \lim_{R \to R_+}  2 R \ptl_R \mbo{\nu}' \Big( (\ptl_t Y^R - N^2 e^{-2 \mbo{\nu}} \ptl_R Y^t) \sin \theta \frac{R^2_+}{R^2} + (\ptl_t Y^\theta  \notag\\
 & \quad -  \frac{N^2}{R^2} e^{-2 \mbo{\nu}} \ptl_\theta Y^t) R \cos \theta (1-\frac{R^2_+}{R^2})  \Big) \Big)  d\theta dt 
\end{align}
expanding out the integrand within $ \text{Flux} (J_3, \mathcal{H}^+)$ close to the horizon, $\mathcal{H}^+$ we have, 
\begin{align}
%&\text{Flux} (J_3, \mathcal{H}^+) \notag\\
&  \lim_{R \to R_+} 2R \ptl_R \mbo{\nu}' \Bigg\{ \bigg(\ptl_t Y^R -    \left( 1 - \frac{R^2_+}{R^2} \right)^2 \cdot \frac{R^4}{  ( (r^2+a^2)^2 -a^2 \Delta \sin^2 \theta)} \\
& \quad \ptl_R \bigg(Y^t(t=0) - \mathcal{Y}^R \frac{1+ \frac{R^2_+}{R^2}}{1-\frac{R^2_+}{R^2}} - \mathcal{Y}^\theta \cot \theta  \bigg) \bigg)\sin \theta \frac{R^2_+}{R^2} \notag\\
& + \bigg(  \ptl_t Y^\theta - (1- \frac{R^2_+}{R^2})^2 \cdot \frac{R^4}{ \left( (r^2+a^2)^2 -a^2 \Delta
	\sin^2 \theta  \right)}   \notag\\
& \quad \ptl_\theta \bigg( Y^t (t=0) - \mathcal{Y}^R \frac{1 + \frac{R^2_+}{R^2}} {{1-\frac{R^2_+}{R^2}} }  - \mathcal{Y}^ \theta \cot \theta \bigg)  \bigg)    \Bigg\} \notag\\
& \to 0 \quad  \text{as} \quad R \to R_+
\end{align}
after plugging in the expression of $Y$ and $\mathcal{Y}$ for $R$ near $R_+.$
Likewise, at spatial infinity, we have 
\begin{align}
\text{Flux} (J_3, \bar{\iota}^0) =&  \int^t_0 \int^\pi _0 \Big ( \lim_{R \to \infty}  2 R \ptl_R \mbo{\nu}' \Big( (\ptl_t Y^R - N^2 e^{-2 \mbo{\nu}} \ptl_R Y^t) \sin \theta \frac{R^2_+}{R^2} + (\ptl_t Y^\theta  \notag\\
& \quad -  \frac{N^2}{R^2} e^{-2 \mbo{\nu}} \ptl_\theta Y^t) R \cos \theta (1-\frac{R^2_+}{R^2})  \Big)   \Big)  d\theta dt 
\end{align}
again the integrand occuring above can be estimated, after plugging in the relevant quantities:
\begin{align}
=&  \lim_{R \to \infty} 2R \ptl_R \mbo{\nu}' \Bigg\{ \bigg(\ptl_t \sum^\infty_{n=1} Y^R_n R^{-n+1} \cos n \theta -    \left( 1 - \frac{R^2_+}{R^2} \right)^2 \cdot \frac{R^4}{  ( (r^2+a^2)^2 -a^2 \Delta \sin^2 \theta)} \\
& \quad \ptl_R \bigg(Y^t(t=0) + \sum^\infty_{n =1} \mathcal{Y}^R_n  R^{-n+1} \frac{1+ \frac{R^2_+}{R^2}}{1-\frac{R^2_+}{R^2}} - \sum^\infty_{n=1} \mathcal{Y}^\theta_n R^{-n+1} \sin n \theta  \cot \theta  \bigg) \bigg)\sin \theta \frac{R^2_+}{R^2} \notag\\
& + \bigg(  \ptl_t \sum^{\infty}_{n=1} Y^\theta_n R^{-n} \sin n \theta  - (1- \frac{R^2_+}{R^2})^2 \cdot \frac{R^4}{ \left( (r^2+a^2)^2 -a^2 \Delta
	\sin^2 \theta  \right)}   \notag\\
& \quad \ptl_\theta \bigg( Y^t (t=0) + \sum^\infty_{n=1} \mathcal{Y}^R_n R^{-n+1} \cos n \theta  \frac{1 + \frac{R^2_+}{R^2}} {{1-\frac{R^2_+}{R^2}} }  - \sum^\infty_{n=1} \mathcal{Y}^\theta_n R^{-n} \sin n \theta  \cot \theta \bigg)  \bigg)    \Bigg\} \notag\\
& \to 0 \quad \text{as} \quad R \to \infty. 
\end{align}
We would like to point out that the choice of $Y^t(t=0)$ is at our discretion. %Likewise, we would like to remind the reader that it is sufficient to estimate the leader order terms in the series above (see Appendix H in \cite{GM17_gentitle}).  
\iffalse 
\begin{align}
(J_4)^b \fdg = 2 \mathcal{L}_{N'} \mbo{\nu}' \bar{\mu}_q q^{ab} \ptl_a N
\end{align}

\fi 

\begin{align}
\text{Flux} (J_4, \Gamma)
\end{align}
We have the flux expression at the axes: 

\begin{align}
&\text{Flux} (J_4, \Gamma) =  \int^t_0  \int_{ (-\infty,-R_
	+)  \cup (R_+, \infty)} \lim_{\theta \to 0, \pi}  2 R \left( \mathcal{L}_{N'} \mbo{\nu}' \frac{1}{R^2} \ptl_\theta N \right) dR dt 
\notag\\ 
&=   \int^t_0  \int_{ (-\infty,-R_
	+)  \cup (R_+, \infty)} \left( \lim_{\theta \to 0, \pi} 2  \cos \theta \left(1 - \frac{R^2_+}{R^2} \right)  \left(\ptl_t Y^c - N^2 e^{-2 \mbo{\nu}} q^{ac} \ptl_a Y^t \right) \ptl_c \mbo{\nu}' \right) dR dt 
\end{align}
expanding out the integrand above, we have, 

\begin{align}
%& \text{Flux} (J_4, \Gamma)  \notag\\
& = \lim_{\theta \to 0, \pi} 2  \cos \theta \left(1 - \frac{R^2_+}{R^2} \right) 
\Bigg \{ \ptl_t Y^R \ptl_R \mbo{\nu}' -   \left (1- \frac{R^2_+}{R^2} \right)^2 \cdot \frac{R^4 \ptl_R \mbo{\nu}'}{ \left( (r^2+a^2)^2 -a^2 \Delta
	\sin^2 \theta  \right)}   \notag\\
& \quad \ptl_R \left( Y^t (t=0) - \mathcal{Y}^R \frac{1 + \frac{R^2_+}{R^2}} {{1-\frac{R^2_+}{R^2}} }  - \mathcal{Y}^ \theta \cot \theta \right) \notag\\
&+ \ptl_t Y^\theta \ptl_\theta \mbo{\nu}'  -   \left( 1 - \frac{R^2_+}{R^2} \right)^2 \cdot \frac{R^4 \ptl_\theta  \mbo{\nu}'}{ ( (r^2+a^2)^2 -a^2 \Delta \sin^2 \theta)} \notag\\
& \quad \ptl_\theta  \left(Y^t(t=0) - \mathcal{Y}^R \frac{1+ \frac{R^2_+}{R^2}}{1-\frac{R^2_+}{R^2}} - \mathcal{Y}^\theta \cot \theta  \right)  \Bigg\}
\end{align}
It may be noted that $N'^\theta$ now vanishes using the expansion of $\sin n \theta $. 
A term related to $\mathcal{Y}^R$ and thus $N'^R$ remains. Let us re-compress this term and represent the integrand as 

\begin{align}
2N'^R \ptl_R \mbo{\nu}' \cos \theta \left( 1- \frac{R^2_+}{R^2} \right) 
\intertext{which can be re-expressed as}
\ptl_R \left( 2 N'^R \ptl_R \mbo{\nu}' \cos \theta \left( 1- \frac{R^2_+}{R^2} \right) \right) - 2 \mbo{\nu'}\ptl_R \left(N'^R \cos \theta \left(1-\frac{R^2_+}{R^2} \right) \right)
\end{align} 
Now, the flux corresponding to the second term above fortuitously combines with the remaining flux term in \eqref{kin-flux-comb-nu'} to yield $0$ total flux, in view of the regularity condition \eqref{gamma-nu condition}. Further, the first term is a total divergence term and converges to $0$ at the boundaries  of the axes ($R \to \infty$), in view of the asymptotic behaviour of $N'^R$ and $\mbo{\nu'}$. 
 \iffalse
 The term that remains is: 
 \begin{align}
 \text{Flux} (J_4, \Gamma) =&  \lim_{\theta \to 0, \pi} 2 R^3 \cos \theta \Bigg \{  \ptl_t Y^R \ptl_R \mbo{\nu}' - \left(  1- \frac{R^2_+}{R^2} \right)   \frac{R^4 \ptl_R \mbo{\nu}'}{ \left( (r^2+a^2)^2 -a^2 \Delta
 	\sin^2 \theta  \right)}     \notag\\
 &  \quad \quad \ptl_R \left( - \mathcal{Y}^R  \frac{1 + \frac{R^2_+}{R^2}} {{1-\frac{R^2_+}{R^2}}} \right) \Bigg \}
 \end{align}
 \fi

Subsequently, 
\begin{align}
& \text{Flux} (J_4, \mathcal{H}^+) = \int^t_0 \int^\pi _0 \left( \lim_{R \to R_+} 2 \mathcal{L}_{N'} \mbo{\nu}' \bar{\mu}_q q^{ab} \ptl_a N \right)  d\theta dt \notag\\
&=  \int^t_0 \int^\pi _0 \left( \lim_{R \to R_+}  2 \mathcal{L}_{N'} \mbo{\nu}' \bar{\mu}_{q_0} \ptl_R N \right) R d\theta dt  \notag\\
=&  \int^t_0 \int^\pi _0 \left( 2 \sin \theta \left( 1+ \frac{R^2_+}{R^2} \right) ( \ptl_t Y^c - N^2 e^{-2 \mbo{\nu}} q^{ac} \ptl_a Y^t) \ptl_c \mbo{\nu}' \right)  d\theta dt 
\end{align}
the integrand is 
\begin{align}
% & \text{Flux} (J_4, \mathcal{H}^+) \notag\\
&= 2 \sin \theta \left( 1+ \frac{R^2_+}{R^2} \right) \Bigg \{ \ptl_t Y^R \ptl_R \mbo{\nu}' -   \left (1- \frac{R^2_+}{R^2} \right)^2 \cdot \frac{R^4 \ptl_R \mbo{\nu}'}{ \left( (r^2+a^2)^2 -a^2 \Delta
	\sin^2 \theta  \right)}   \notag\\
& \quad \ptl_\R \left( Y^t (t=0) - \mathcal{Y}^R \frac{1 + \frac{R^2_+}{R^2}} {{1-\frac{R^2_+}{R^2}} }  - \mathcal{Y}^ \theta \cot \theta \right) \notag\\
&+ \ptl_t Y^\theta \ptl_\theta \mbo{\nu}' -  \left( 1 - \frac{R^2_+}{R^2} \right)^2 \cdot \frac{R^4 \ptl_\theta \mbo{\nu}'}{ ( (r^2+a^2)^2 -a^2 \Delta \sin^2 \theta)} \notag\\
& \quad \ptl_\theta \left(Y^t(t=0) - \mathcal{Y}^R \frac{1+ \frac{R^2_+}{R^2}}{1-\frac{R^2_+}{R^2}} - \mathcal{Y}^\theta \cot \theta  \right) \Bigg\}
\end{align}

The second term in the curly brackets rapidly vanishes due to the decay rate of $Y^ \theta$ close to $\mathcal{H}^+.$ The non vanishing factor in the second line above is accounted for by the vanishing boundary condition of $\ptl_R \mbo{\nu}'$ at the horizon. 
 
\begin{align}
& \text{Flux} (J_4, \bar{\iota}^0) =  \int^t_0 \int^\pi _0 \left( \lim_{R \to \infty} 2 \mathcal{L}_{N'} \mbo{\nu}' \bar{\mu}_{q_0} \ptl_R N \right)  d\theta dt \notag \\
& =  \int^t_0 \int^\pi _0 \left(  \lim_{R \to \infty}  2 \sin \theta \frac{R^2_+}{R} ( \ptl_t Y^c - N^2 e^{-2 \mbo{\nu}} q^{ac} \ptl_a Y^t) \ptl_c \mbo{\nu}' \right)  d \theta dt   \notag\\
\intertext{so that the integrand involves}
 &=  \lim_{R \to \infty} 2 \sin \theta \left( 1+ \frac{R^2_+}{R^2} \right) \Bigg \{ \ptl_t Y^R \ptl_R \mbo{\nu}'  -   \left (1- \frac{R^2_+}{R^2} \right)^2 \cdot \frac{R^4 \ptl_R \mbo{\nu}'}{ \left( (r^2+a^2)^2 -a^2 \Delta
 	\sin^2 \theta  \right)}   \notag\\
 & \quad \ptl_R \left( Y^t (t=0) - \mathcal{Y}^R \frac{1 + \frac{R^2_+}{R^2}} {{1-\frac{R^2_+}{R^2}} }  - \mathcal{Y}^ \theta \cot \theta \right) \notag\\
 &+ \ptl_t Y^ \theta \ptl_\theta \mbo{\nu}' -  \left( 1 - \frac{R^2_+}{R^2} \right)^2 \cdot \frac{R^4 \ptl_\theta  \mbo{\nu}'}{ ( (r^2+a^2)^2 -a^2 \Delta \sin^2 \theta)} \notag\\
 & \quad \ptl_R \left(Y^t(t=0) - \mathcal{Y}^R \frac{1+ \frac{R^2_+}{R^2}}{1-\frac{R^2_+}{R^2}} - \mathcal{Y}^\theta \cot \theta  \right) \Bigg\}.
\end{align}
In the second term in the curly brackets decay only after using the decay rate of $\ptl_ \theta \mbo{\nu}'$ (because $ \ptl_\theta Y^R$ does not decay near $\iota^0$). 
%We have, 
%\begin{align}
%\mathcal{L}_{N'} \mbo{\nu}' =& N'^R \ptl_R \mbo{\nu}' + N'^\theta \ptl_\theta \mbo{\nu}'  \notag\\
%=& \ptl_t Y^R - N^2 q^{Rb} \ptl_b Y^t \ptl_R \mbo{\nu}' + \ptl_t Y^\theta - N^2 q^{\theta b} \ptl_b %Y^t \ptl_\theta \mbo{\nu}'  
%\end{align}

\begin{align}
(J_5)^b \fdg = -2N'^b \bar{\mu}_q q^{ac} \ptl_a \mbo{\nu}' \ptl_c N
\end{align}
\begin{align}
\text{Flux} (J_5, \Gamma)
\end{align}
we have the flux of $J_5$ at the axes given by, 

\begin{align}
\text{Flux} ( J_5, \Gamma) =&  \int^t_0 \int_{(-\infty, -R_+) \cup ( R_+, \infty)} (\lim_{\theta \to 0, \pi}-2 R^2 N'^ \theta q^{ac}_0 \ptl_a \mbo{\nu}' \ptl_c N ) dR dt  \notag\\
=&  \int^t_0 \int_{(-\infty, -R_+) \cup ( R_+, \infty)}   ( \lim_{\theta \to 0, \pi} -2 R^2 \left(\ptl_ t Y^\theta - \frac{N^2}{R^2} e^{-2 \mbo{\nu}'} \ptl_\theta Y^t \right) q^{ac}_0 \ptl_a \mbo{\nu}' \ptl_c N ) dR dt   \notag\\
\intertext{now consider the integrand}
=&  \lim_{\theta \to 0, \pi}-2 R^2 \left(\ptl_ t Y^\theta - \frac{N^2}{R^2} e^{-2 \mbo{\nu}'} \right) \cdot  \notag\\ 
& \cdot \left(  \ptl_R \mbo{\nu}' \sin \theta (1 + \frac{R^2_+}{R^2}) + \frac{1}{R^2} \ptl_ \theta \mbo{\nu}' R \cos \theta (1- \frac{R^2_+}{R^2}) \right)  \notag \\
& \to 0 \quad \text{as} \quad \theta \to 0, \pi 
\end{align}

and 

\begin{align}
\text{Flux} (J_5, \mathcal{H}^+) =& \int^t_0  \int^\pi _0  \left( \lim_{R \to R_+} q_0 ( -2N'^b \bar{\mu}_q q^{ac} \ptl_a \mbo{\nu}' \ptl_c N, \ptl_R) \right)  d\theta dt  \notag \\
=& \int^t_0  \int^\pi _0 \left( \lim_{R \to R_+} - 2 N'^R \bar{\mu}_q q^{ac} \ptl_a \mbo{\nu}' \ptl_c N \right)  d\theta d t  \notag \\
=& \int^t_0 \int^\pi_0 \left(  \lim_{R \to R_+} -2( \ptl_t Y^R - N^2 e^{-2 \mbo{\nu}} \ptl_R Y^t) \bar{\mu}_q q^{ac} \ptl_a \mbo{\nu}' \ptl_c N \right)  d\theta dt 
\end{align}
the integrand contains the terms,  
\begin{align}
&
=- 2 \Bigg \{ \ptl_t \Big( \sum^\infty_{n \to 1} Y^R_n  R R^+_n ( \frac{R^n}{R^n_+}- \frac{R_+^n}{R^n} \cos \theta )   \notag\\
& \quad - R^2  \left( 1- \frac{R^2_+}{R^2}\right)^2 \frac{1}{ (r^2 +a^2)^2 - a^2 \Delta \sin \theta}   \Big) \cdot \notag\\
&\ptl_R  \left( Y^t(t=0)- \frac{1+ \frac{R^2_+}{R^2}}{R (1- \frac{R_+^2}{R^2})} \sum^\infty_{n = 1} \mathcal{Y}^R_n  R R^+_n \left( \frac{R^n}{R^n_+}  \frac{R_+^n}{R^n} \cos \theta \right) \right) - \cot \theta \sum^\infty_{n=1} \mathcal{Y}_n^\theta R^n_+\left( \frac{R^n}{R^n_+} + \frac{R^n_+}{R^n} \right) \cos n\theta    \Bigg\}  \notag\\
& \notag \cdot \left( \ptl_R \mbo{\nu}'  \sin \theta (1 + \frac{R^2_+}{R^2}) + \frac{1}{R^2} \ptl_\theta \mbo{\nu}' R \cos \theta (1- \frac{R^2_+}{R^2}) \right) \notag\\
&\to 0 \quad \text{as} \quad R \to R_+ 
\end{align}
due to vanishing term in the final bracket. 

We are now left with the final flux term: 
\begin{align}
&\text{Flux} (J_5, \bar{\iota}^0) =  \int^t_0 \int^\pi_0  \left( \lim_{R \to \infty} -2( \ptl_t Y^R - N^2 e^{-2 \mbo{\nu}} \ptl_R Y^t) \bar{\mu}_q q^{ac} \ptl_a \mbo{\nu}' \ptl_c N \right)  d\theta dt  \notag\\
& =   \int^t_0 \int^\pi_0  \Big( \lim_{R \to \infty} -2 \left( \ptl_t Y^R - N^2 e^{-2 \mbo{\nu}} \ptl_R ( Y^t (t=0) - \mathcal Y^R \ptl_R N/N -\mathcal{Y}^\theta \ptl_\theta N/N) \right)   \cdot \notag\\
& \quad \quad \cdot \left( R \ptl_R \mbo{\nu}' \sin \theta (1+ \frac{R^2_+}{R^2}) + \frac{1}{R} \ptl_ \theta \mbo{\nu}' \cos \theta (1- \frac{R^2_+}{R^2}) \right) \Big)  d\theta dt \notag\\
\intertext{again the integrand can be estimated as}
&= \lim_{R \to \infty} -2 \Bigg \{- \ptl_t \left(\sum^{\infty}_{n=1}  Y^R_n R^{-n+1} \cos n \theta  \right) - R^4  \left( 1- \frac{R^2_+}{R^2}\right)^2 \frac{1}{ (r^2 +a^2)^2 - a^2 \Delta \sin \theta}   \notag\\ 
& \quad \cdot \left( Y^t(t=0)- \frac{1+ \frac{R^2_+}{R^2}}{R (1- \frac{R_+^2}{R^2})} \sum^{\infty}_{n=1} \mathcal{Y}^R_n R^{-n+1} \cos n \theta - \cot \theta \sum^\infty_{n=1} \mathcal{Y}^\theta_n R^{-n} \sin n\theta  \right)    \Bigg\} \cdot \notag\\
& \quad \cdot \big \{ R \cdot \mathcal{O} (\frac{1}{R}) \mathcal{O}(1) + \frac{1}{R} \cdot \mathcal{O}(\frac{1}{R})  \cdot \mathcal{O}(1) \big \}  \notag\\
& \to 0 \quad \text{as} \quad R \to \infty \quad \text{and behaves like} \quad \mathcal{O}(\frac{1}{R})
\quad \text{for large} \quad R
\end{align}
Thus it follows that 
\begin{align}
\text{Flux}(J_5, \iota^0) = 0.
\end{align}
%We note that $\Gamma \cap M' \neq \emptyset $

\iffalse
We would like to remind the reader that the (integrated) positive-definite, regularized  Hamiltonian energy $H^{\text{Reg}}$ itself is a well-defined \emph{gauge-invariant} quantity. It may be noted that, in the construction above, the choice of the gauge-transform vector field is at our discretion and as such is not unique. In particular, we can come up with gauge-transform vector field $\bar{Y}$ such that the fluxes vanish \emph{globally} but not necessarily \emph{locally}. To see this, consider the case of 
\begin{align}
P \fdg \bar{Y} \to Y
\end{align}   
such that the conformal Killing vector field $Y$ $(\text{CK} (Y, q_0) =0)$
\begin{align}
Y^{\theta} =& \notag\\
Y^{R}=&
\end{align}
\fi
\subsection*{Acknowledgements} 

In a previous work \cite{GM17_gentitle}, the co-author Vincent Moncrief, had fittingly payed tribute to A. Taub, J. Marsden and S. Dain for their influence on him and for making fundamental contributions in this direction. I take this opportunity to pay tribute to his own outstanding contributions to general relativity and Hamiltonian methods. On an individual note, for his mentorship, encouragement and enjoyable interactions, I am indebted to him. 

This project has spanned several years and my gratitude is also due to all the colleagues who  have continually supported me and to the institutions that hosted me. Special thanks are due to my host Hermann Nicolai at the Albert Einstein Institute  at Golm (DFG grant no. GU 1513/2-1), where significant aspects of this work were completed. The conception and some initial aspects of this work were done while I was a postdoc at the Department of Mathematics, Yale University, where I benefited from conducive working conditions.  

\bibliography{../References(Harvard)/central-bib(Harvard)}
%\bibliography{biblio}
\bibliographystyle{plain}
\end{document}